\documentclass[reqno,11pt,a4paper]{amsart}
\usepackage[a4paper,hmargin=2.8cm,vmargin=3cm]{geometry}
\usepackage[utf8x]{inputenc}
\usepackage{amsmath,amsthm,amsfonts,amssymb}
\usepackage[bookmarksnumbered=true]{hyperref}
\usepackage{xcolor}
\usepackage{stmaryrd}
\usepackage{dsfont} 
\usepackage[textsize=footnotesize,textwidth=2.5cm]{todonotes}
\makeatletter%
\@mparswitchfalse%
\makeatother%
\normalmarginpar
\usepackage[capitalise]{cleveref}
\crefname{equation}{}{}
\usepackage[english]{babel}

\newcommand{\cC}{\mathcal{C}}
\newcommand{\cU}{\mathcal{U}}

\newcommand{\cD}{\mathcal{D}}

\newcommand{\cW}{\mathcal{W}}

\newcommand{\cO}{\mathcal{O}}

\newcommand{\Ncal}{\mathcal{N}}
\newcommand{\Hcal}{\mathcal{H}}

\newcommand{\sgn}{\operatorname{sgn}}
\newcommand{\Ucal}{\mathcal{U}}
\newcommand{\Lcal}{\mathcal{L}}
\newcommand{\Acal}{\mathcal{A}}
\newcommand{\Bcal}{\mathcal{B}}
\newcommand{\Ccal}{\mathcal{C}}
\newcommand{\Mcal}{\mathcal{M}}
\newcommand{\hc}{\textnormal{h.c.}}

\def\norm#1{{\left|\hskip-.05em\left|#1\right|\hskip-.05em\right|}}

\newcommand{\tr}{{\rm tr}}
\newcommand{\HS}{{\rm HS}}

\theoremstyle{plain}

\newtheorem{theorem-princ}{Theorem}

\newtheorem{theorem}{Theorem}[section]  
\newtheorem{cor}[theorem]{Corollary}
\newtheorem{prop}[theorem]{Proposition}
\newtheorem{lemma}[theorem]{Lemma}
\newtheorem{assump}[theorem]{Assumption}
\newtheorem{defi}[theorem]{Definition}

\theoremstyle{remark}

\newtheorem{rmk}[theorem]{Remark}

\numberwithin{equation}{section}

\newcommand{\abs}[1]{\left\lvert #1 \right\rvert}
\newcommand{\prodscal}[2]{\left\langle#1,#2\right\rangle}
\newcommand{\crochetjap}[1]{\left\langle#1\right\rangle}

\def\N{\mathbb{N}}
 
\def\R{\mathbb{R}} 
\newcommand{\Rbb}{\mathbb{R}}
\def\C{\mathbb{C}} 

\newcommand{\Nbb}{\mathbb{N}}

\newcommand{\dom}{{\rm Dom}}
\renewcommand{\div}{{\rm div}}
\newcommand{\op}{{\rm op}}
\newcommand{\Hhb}{{H_{\hbar,b}}}
\newcommand{\qhb}{{\mathfrak{q}_{\hbar,b}}}
\newcommand{\Hloc}{{H^{\rm loc}_{\hbar,b}}}

\newcommand{\test}[1]{\mathcal{C}_0^{\infty} (#1)}
\newcommand{\indic}{\mathds{1}}
\newcommand{\dist}{{\rm dist}}
\newcommand{\supp}{{\rm supp}}
\newcommand{\spec}{{\rm Spec}}
\newcommand{\loc}{{\rm loc}}
\newcommand{\normLp}[3]{\left\lVert #1 \right\rVert_{L^{#2}{#3}}}
\newcommand*{\di}{\mathop{}\!\mathrm{d}}

\renewcommand{\Re}{{\rm Re}\,}
\renewcommand{\Im}{{\rm Im}\,}

\newcommand{\bx}{\mathbf{x}}
\newcommand{\by}{\mathbf{y}}

\usepackage{comment}

 \usepackage{centernot}
\usepackage{framed}

\title[Fermions in a Magnetic Field]{Derivation of Hartree--Fock Dynamics and Semiclassical Commutator Estimates for Fermions in a Magnetic Field}

\author{Niels Benedikter}
\address{Università degli Studi di Milano}
\urladdr{https://nielsbenedikter.de/}
\thanks{N.\,B.\ and N.\,N.\ are supported by ERC StG \textsc{FermiMath} nr.~101040991.}

\author{Chiara Boccato}
\address{Università di Pisa}
\urladdr{http://www.chiaraboccato.com/}
\thanks{C.\,B.\ and D.\,M.\ are supported by MUR–Italian Ministry of University and Research and Next Generation EU within PRIN 2022AKRC5P “Interacting Quantum Systems: Topological Phenomena and Effective Theories”}

\author{Domenico Monaco}
\address{Università degli Studi di Roma ``La Sapienza''}
\urladdr{https://sites.google.com/view/dmonaco}
\thanks{D.\,M.\ is also supported by MUR within PNRR Project nr.~PE0000023-NQSTI, and by Sapienza Universit\`a di Roma within Progetto di Ricerca di Ateneo 2021, 2022 and 2023.}

\author{Ngoc Nhi Nguyen}
\address{Università degli Studi di Milano}
\email{ngoc.nguyen@unimi.it}
\urladdr{https://sites.google.com/view/ngocnhinguyen/english-version}
\thanks{This work has been carried out under the auspices of the GNFM-INdAM (Gruppo
Nazionale per la Fisica Matematica — Istituto Nazionale di Alta Matematica)}

\begin{document}

	\begin{abstract}
	We study the quantum dynamics of a large number of interacting fermionic particles in a constant magnetic field. In a coupled mean-field and semiclassical scaling limit, we show that solutions of the many-body Schrödinger equation converge to solutions of a non-linear Hartree--Fock equation. The central ingredient of the proof are certain semiclassical trace norm estimates of commutators of the position and momentum operators with the one-particle density matrix of the solution of the Hartree--Fock equation. In a first step, we prove their validity for non-interacting initial data in a magnetic field by generalizing a result by Fournais and Mikkelsen \cite{Fournais-Mikkelsen-2020}.
	We then propagate these bounds from the initial data along the Hartree--Fock flow to arbitrary times.
	\end{abstract}

    	\maketitle
	
	\tableofcontents
	
\section{Introduction}
A wide variety of condensed matter systems can be quantum-mechanically described as gases of interacting fermionic particles. This includes ordinary metals, as well as semiconductors. Of particular interest is the description of fermions in two-dimensional materials subject to a homogeneous external magnetic field, due to the possible emergence of (fractional) quantization of the conductivity in the quantum Hall effect. The topological nature of this phenomenon reflects in its robustness against changes in the microscopic details of the model, a form of universality. Of particular relevance is the fact that the Hall conductivity stays quantized even in presence of weak interactions \cite{giuliani2017universality, bachmann2018quantization, monaco2019adiabatic, teufel2020non, wesle2024exact}. Fractionalization of the Hall charge is believed to be a consequence of strong correlations between particles and continues to be a topic of active investigation, despite recent important progress in its mathematical understanding \cite{bachmann2021rational, bachmann2024tensor}. To narrow the range of prerequisites for the emergence of the quantum Hall effect, we consider the widely studied mean-field scaling regime and prove that it gives rise only to weak correlations between particles, i.e., the time evolution of the system's state remains well approximated by a Slater determinant evolving by a Hartree--Fock equation. Our proof adapts the many-body analysis in Fock space of \cite{BPS-2014-MFevol}; this method, however, relies crucially on certain semiclassical properties of the solution of the Hartree--Fock equation, which turn out to be challenging to establish for anything but non-interacting fermions on the torus. At the core of our proof is therefore the task of establishing the semiclassical properties of non-interacting fermions in a magnetic field as initial data, followed by the propagation of these semiclassical properties to arbitrary times along the Hartree--Fock flow which includes both magnetic field and interaction. The proof of the semiclassical properties of the initial data generalizes the method of \cite{Fournais-Mikkelsen-2020}, whereas the propagation generalizes \cite{BPS-2014-MFevol}.
	
	\smallskip

	The starting point of our discussion is the one-particle magnetic Schrödinger operator on $L^2(\R^d)$ for $d \in \Nbb_{\geq 1}$ given by
	\begin{equation}\label{def:magnetic-Sch-op}
		\Hhb := (-i\hbar\nabla-b a(x))^2+ V,
	\end{equation}
	where $\hbar\in(0,\hbar_0]$ is the semiclassical parameter (where a small $\hbar_0>0$ is fixed), $b\geq 0$ tunes the strength of the magnetic field, $a=(a_1,\ldots,a_d):\R^d\to\R^d$ is an external magnetic vector potential, and $V:\R^d\to\R$ is an external scalar potential. We state the precise hypotheses on the potentials $a$ and $V$ in \cref{ass:a-V_gene}, that define $\Hhb$ as a self-adjoint operator on $L^2(\R^d)$, bounded from below. We also make the choice to deal with the case where $\Hhb$ has discrete spectrum $(\lambda_n)_{n\in\N_{\geq 1}}$, for instance under the additional assumption that $V$ is \emph{confining}, i.e.\ when $V(x)\to +\infty$ as $\abs{x}\to\infty$.    
	\smallskip

	To describe the interacting many-body evolution of the fermionic system, we introduce the $N$-particle Hilbert space $L^2_\textnormal{a}(\Rbb^{dN})$ of totally antisymmetric wave functions, i.\,e., those $\psi \in L^2(\Rbb^{dN})$ satisfying
	\[
		\psi(x_{\sigma(1)},\ldots, x_{\sigma(N)}) = \sgn(\sigma) \psi(x_1,\ldots,x_N) \quad \textnormal{for all }\sigma \in \mathcal{S}_N\;.
	\]
	The Hamiltonian for the interacting $N$-body system is defined using a pair interaction potential $W:\Rbb^d \to \Rbb$ as
	\begin{equation}\label{eq:Nbodyham}
	   H_{\hbar,b}^N := \sum_{i=1}^N (H_{\hbar,b})_i + \lambda \sum_{1 \leq i < j \leq N} W(x_i - x_j)	\;.
	\end{equation}
	By $(H_{\hbar,b})_i$ we denote the Hamiltonian $H_{\hbar,b}$ acting on the variables of the $i$-th particle. The coupling of the mean-field scaling limit to the semiclassical scaling \cite{NS81} is realized by choosing
	\begin{equation}	\label{eq:mf}
	\lambda := N^{-1} \;, \qquad \hbar := N^{-1/d} \;, \quad\qquad N \to \infty \;.
	\end{equation}
    (Different fermionic scaling regimes, which are simpler in that the semiclassical scaling is not coupled to the mean-field scaling, have been studied by \cite{FK11,PP16}.)
	In the non-interacting case $W =0$, the eigenvectors of $H_{\hbar,b}^N$ can be obtained as antisymmetric tensor products (\emph{Slater determinants}) of the eigenfunctions of $H_{\hbar,b}$; more specifically, if $H_{\hbar,b} \psi_n = \lambda_n \psi_n$ for~$N$ orthonormal elements $\psi_n$ of $L^2(\Rbb^d)$, then
	\begin{equation}	\label{eq:slater}
		\Psi_N := \bigwedge_{n=1}^N \psi_n\;, \qquad \Psi_N(x_1,\ldots,x_N) = \frac{1}{\sqrt{N!}} \det\left( \psi_i(x_j) \right)_{i,j=1,\ldots,N}\;,
	\end{equation}
	is a normalized eigenvector of $H_{\hbar,b}^N$, with eigenvalue $E_N = \sum_{n=1}^N \lambda_n$. The $\psi_n$ are also called the \emph{orbitals} of the Slater determinant. The \emph{one-particle reduced density matrix} corresponding to $\Psi_N$ is a trace-class operator on the one-particle Hilbert space $L^2(\Rbb^d)$ defined (up to the normalization by $N$) as the partial trace over $N-1$ particles of the projection onto $\Psi_N$, i.\,e.,
	\[
		\gamma^{(1)}_{\Psi_N} := N \tr_{N-1} \lvert \Psi_N \rangle \langle \Psi_N \rvert
	\]
    or equivalently in terms of its integral kernel
    \[
    \gamma^{(1)}_{\Psi_N}(x;y) := N \int_{(\R^d)^{(N-1)}} \di x_2 \ldots \di x_N \Psi_N(x,x_2,\ldots,x_N) \overline{\Psi_N(y,x_2,\ldots,x_N)} \;.
    \]
	If $\Psi_N$ is the Slater determinant \cref{eq:slater}, then $\gamma^{(1)}_{\Psi_N} = \sum_{n=1}^N \lvert \psi_n \rangle \langle \psi_n\rvert$, the orthogonal projection on the orbitals. Vice versa, given any rank-$N$ projection, one can obtain a set of $N$ orbitals from its spectral decomposition and use them to construct a Slater determinant, which is up to a phase factor unique. In particular, in the non-interacting case the ground state of the system described by $H_{\hbar,b}^N$ is the Slater determinant  obtainable from the spectral projection $\indic_{(-\infty,\mu]}(\Hhb)$ where the \emph{chemical potential} $\mu$ is adjusted so that $\operatorname{rank} \indic_{(-\infty,\mu]}(\Hhb) = N$.

	\smallskip

	This simple relation between the one-body theory and the many-body theory breaks down completely when the interaction does not vanish, $W \not= 0$. Since in typical condensed matter systems the number of interacting particles is of order $N \sim 10^{23}$, even a numerical analysis of the interacting many-body problem has to rely on approximation methods. Such methods include, in increasing order of precision: Thomas--Fermi density functional theory and the Vlasov equation for the time evolution \cite{NS81,Spo81}; Hartree--Fock theory; the random-phase approximation \cite{BNP+20,BNP+21,BNP+22,BPSS23,BL23,CHN22,CHN23,CHN23a,CHN24}. The derivation of Thomas--Fermi theory for fermions in magnetic fields was discussed by \cite{FM20a}, while in the same setting \cite{gontier2023density} discusses density functional theory and properties of the reduced Hartree--Fock minimizers. In the present paper we focus on the derivation of Hartree--Fock theory for fermions in a magnetic field.

	General elements of $L^2_\textnormal{a}(\Rbb^{dN})$ can always be written as linear combinations of Slater determinants. The central principle of Hartree--Fock theory is to restrict attention to the submanifold of single Slater determinants, i.\,e., not admitting any linear combinations. The remaining degree of freedom is then the choice of the orbitals $\psi_n$. In time-independent Hartree--Fock theory, this gives rise to a variational problem over the orbitals of the expectation value of $H_{\hbar,b}^N$ as an approximation to the ground state energy; in time-dependent Hartree--Fock theory one aims to find a trajectory in the submanifold of Slater determinants that stays as close as possible to the true many-body evolution in $L^2_\textnormal{a}(\Rbb^{dN})$. Since Slater determinants can be identified with rank-$N$ projections, time-dependent Hartree--Fock theory provides a (non-linear) equation for the evolution of a rank-$N$ projection, the time-dependent Hartree--Fock equation
\begin{equation}\label{eq:HF-dyn}
	i\hbar\partial_t\omega_{N,t}=\Big[h_{\rm HF}(t),\omega_{N,t}\Big],\quad \omega_{N,t=0}=\omega_N\;.
\end{equation}
The operator $h_{\rm HF}(t)= H_{\hbar,b}+ W\ast\rho_t-X_t$ is the Hartree--Fock Hamiltonian, which depends itself on the solution of the equation through the density $\rho_t(\bx)=N^{-1}\omega_{N,t}(\bx,\bx)$ and the exchange operator, i.\,e., the operator $X_t$ with the integral kernel $X_t(\bx,\by)=N^{-1}\omega_{N,t}(\bx,
\by)W(\bx-\by)$ for $\bx,\by\in(\R^d)^N$.

	\smallskip

	We are particularly interested in the description of a \emph{quantum quench}: an initial Slater determinant is prepared for the non-interacting system with $W=0$, then, at time $t=0$, the interaction is abruptly switched on ($W \not=0$) and the system evolves by the many-body Schrödinger equation
	\begin{equation}\label{eq:SE-dyn}
	i\hbar \partial_t \Psi_N(t) = H_{\hbar,b}^N \Psi_N(t)
	\end{equation}
	Since the interaction in the mean-field scaling \cref{eq:mf} is weak, we can expect the time-dependent Hartree--Fock equation to provide a good approximation to the one-particle reduced density matrix $\gamma^{(1)}_{\Psi_N(t)}$ obtained from the many-body Schr\"odinger evolution. This expectation is confirmed by \cref{thm:hf}. Its proof depends crucially on the presence of a subtle semiclassical structure in the initial data, which we believe is of independent interest, also in view of the broader context discussed below in \cref{sec-intro:state_of_the_art}. We establish this semiclassical structure in \cref{thm:cb-princ}. The propagation of the semiclassical structure along the Hartree--Fock flow is also non-trivial; we discuss it in \cref{thm:propag-semicl-magn}.

\subsection{Notation}

	For $x\in\R^d$ we write $\abs{x}:=\big(\sum_{j=1}^d x_j^2\big)^{\frac 12}$ for the Euclidean norm and $\crochetjap{x}:=(1+\abs{x}^2)^{\frac 12}$ for the Japanese bracket. We use $\N$ for the set of natural numbers that contains $0$ and $\N_{\geq 1}$ for natural numbers starting from $1$. For compact operators $H$ acting on $L^2(\R^d)$, we introduce the trace and Hilbert--Schmidt norms by
	\begin{equation*}
		\norm{H}_1=\norm{H}_{\tr}:=\tr\Big(\sqrt{H^*H}\Big) \;, \qquad
		\norm{H}_2=\norm{H}_{\text{HS}}:=\tr\Big(H^*H\Big)^{\frac 12} \;.
	\end{equation*}
	Moreover we denote by $\norm{\cdot}_\op$ the operator norm. If $H=(H_1,\ldots,H_n)$ is an $n$-tuple of operators, then any of the above norms is understood in the following $\ell^2$--sense:
	\begin{equation*}
		\norm{H}_{(\ast)} := \Big(\sum_{j=1}^d\norm{H_j}_{(\ast)}^2\Big)^{\frac 12}\;, \qquad (\ast) \in \{1,2,\op\}.
	\end{equation*}
	
	\subsection{First Main Result: Semiclassical Structure of Spectral Projections}\label{sec-intro:main-result}
    
    Let us recall how to define more formally the magnetic Schr\"odinger operator \eqref{def:magnetic-Sch-op} and present the setting of our first main result.

    \begin{assump}\label{ass:a-V_gene}
    Let $a:\R^d\to\R^d$, $V:\R^d\to\R$ be such that
     $a_j\in L^2_\loc(\R^d)$ for any $j\in\{1,\ldots,d\}$ and that there is a decomposition $V=V_+-V_-$, such that $V_+\in L^1_\loc(\R^d)$ and $V_-\in L^p(\R^d)+L^\infty_\varepsilon(\R^d)$ with
	\begin{equation}\label{eq:cond-p}
			p=1 \ \text{ if }\ d=1\,, \qquad
			p>1 \ \text{ if }\ d=2\,, \qquad
			p\geq \frac{d}{2} \ \text{ if }\ d\geq 3\,.
	\end{equation}
    \end{assump}
    
	 Moreover, we define $\Hhb$ as the Friedrichs extension of a minimal operator $H_{\hbar,b}^{\min}: u\mapsto(-i\hbar\nabla-ba(x))^2u+V(x)$ of domain $\dom(H_{\hbar,b}^{\min})=\test{\R^d}$, and associated to the domain
	\begin{align*}
		\dom(\Hhb):=\{ u\in L^2(\R^d) \: :&\: (-i\hbar\nabla-ba(x))u\in L^2(\R^d), \: \sqrt{V_+}u\in L^2(\R^d), \\&\:  \text{and }(-i\hbar\nabla-ba(x))^2u+V(x)u\in L^2(\R^d)\} \;
		\;,
	\end{align*}
	and for any $u\in L^2(\R^d)$ such that $(-i\hbar\nabla-ba(x))u\in L^2(\R^d), \: \sqrt{V_+}u\in L^2(\R^d)$,  and $v\in\dom(H_{\hbar,b})$
	\begin{align*}
		\int_{\R^d}\overline{(i\hbar\nabla+ba(x))u(x)}(i\hbar\nabla+ba(x))v(x) \di x +\int_{\R^d}V(x)\overline{u(x)}v(x) \di x =\prodscal{u}{\Hhb v}.
	\end{align*}
	The operator $\Hhb$ is semibounded from below and is essentially self-adjoint.
	When the potential is confining, i.\,e.,  $V_+(x) \to \infty$ as $\abs{x}\to \infty$, then $\Hhb$ has compact resolvent. As a consequence, it has a non-decreasing sequence of eigenvalues $(\lambda_n)_{n\in\N}$ of finite multiplicities such that $\lambda_n \to+\infty$ as $n\to +\infty$ and is diagonalizable in an orthonormal basis of eigenfunctions $(\psi_n)_{n\in\N}$. To keep the paper self-contained, we sketch the proof in \cref{sec:basic-prop-magnetic}; for further background, see \cite{Lieb-Seiringer-2010Stability,Lewin-TS_et_MQ,Frank-Laptev-Weidl-2023SchrLT}.

\smallskip
	
	Our first  main result is based on the following general assumptions on the potentials, that are similar to the ones in \cite{Sobolev-1995,Fournais-Mikkelsen-2020}. This hypothesis is a little more general than the ones above required to define the magnetic Schr\"odinger operator, since that includes confining potential $V$. On the other hand, we assume smoothness of $V$ and of $a$ in a compact set of the Euclidean space.
	
	\begin{assump}\label{ass:a-V}
		Let $a:\R^d\to\R^d$, $V:\R^d\to\R$, and $\mu\in\R$ be such that there exist an open set $\Omega_V\subset\R^d$ and a constant $\varepsilon>0$ such that
		\begin{itemize}
            \item $\mu>\inf_{\R}V$,
			\item $a\in\cC^\infty(\Omega_V,\R^d)$ and $a\indic_{\Omega_V^c}\in L^2_\loc(\R^d,\R^d)$,
			\item $V\in\cC^\infty(\Omega_V)$, $V\indic_{\Omega_V^c}\in L^1_\loc(\R^d)$, and $(V-\mu)_-\in L^{1+\frac d2}(\R^d)$,
			\item there exists a bounded set $\Omega_\varepsilon$ such that $\bar{\Omega}_\varepsilon\subset\Omega_V$ and $\{x\in\R^d\: :\: V(x)-\mu\leq\varepsilon\}\subset\Omega_\varepsilon$.
		\end{itemize}
	\end{assump}
	
	We prove that the spectral projector of the magnetic one-body operator (or equivalently, the one-body reduced density matrix of the Slater determinant obtained as the ground state of the corresponding non-interacting many-body system) satisfies commutators bounds, to which we refer as the \emph{semiclassical structure}.	
	
	\begin{theorem-princ}[Semiclassical Structure of Initial Data]\label{thm:cb-princ}
		Let $a:\R^d\to\R^d$, $V:\R^d\to\R$ and $\mu\in\R$ satisfy \cref{ass:a-V}. Let $C_0>0$ and $h_0>0$, let $P_j:=-i\hbar\partial_{x_j}-b a_j(x)$ be the $j$-th coordinate of the magnetic one–particle momentum operator $P=-i\hbar\nabla-b a(x)$ for $j\in\{1,\ldots,d\}$, and let $H_{\hbar,b}=(-i\hbar\nabla-ba(x))^2+V(x)$ be the self-adjoint operator defined as above.
        Let us also assume that $\mu$ belongs to the resolvent set of the operator $P^2$.
		Let $f:\R^d\to\R$ be a smooth function with polynomial growth, i.\,e., such that there exist $m\in\N_{\geq 1}$, $R>0$ and $C_1>0$ such that
		\begin{equation}\label{ass:polynom-growth}
			\abs{f(x)}\leq C_1\crochetjap{x}^m \quad \forall x\in\R^d,\:\abs{x}\geq R\,.
		\end{equation}
		Then, there exists $C>0$ such that, for any $h\in(0,h_0]$, for any $b\in(0,C_0\hbar^{-1})$, and for any $j\in\{1,\ldots,d\}$, one has
		\begin{equation}\label{eq:cb-princ-f(x)}
			\norm{[\indic_{(-\infty,\mu]}(\Hhb),f(x)]}_1\leq C\crochetjap{b}\hbar^{1-d}
		\end{equation}
		and 
		\begin{equation}\label{eq:cb-princ-momentum}
			\norm{[\indic_{(-\infty,\mu]}(\Hhb),P_j]}_1\leq C\crochetjap{b}^2\hbar^{1-d}.
		\end{equation}
	\end{theorem-princ}

\begin{rmk}\label{rmk:cb-princ_proof}
	\begin{enumerate}
    \item
    The assumptions on the regularity of $a$, $V$, and $f$ are due to the method of our proof. As in \cite{Fournais-Mikkelsen-2020}, the idea is to write, for $R=f(x)$ or $R=P$,
	\begin{align*}
		&\norm{[\indic_{(-\infty,\mu]}(\Hhb),R]}_1
		\\&\qquad\leq 	\norm{[\indic_{(-\infty,\mu]}(\Hhb),R\chi(x)]}_1\ +\norm{[\indic_{(-\infty,\mu]}(\Hhb),R(1-\chi(x))]}_1. 
	\end{align*}
	for a well chosen bump function $\chi\in\test{\R^d}$ supported on a bounded neighborhood of the set $\{x\in\R^d\::\: V(x)-\mu\leq \varepsilon\}$.

    \smallskip
    
            We control the first summand by local commutator bounds (\cref{thm:cb-loc-gene}) of order $\cO(\crochetjap{b}^{r+1}\hbar^{1-d})$, where $r=0$ when $R=f(x)$ and $r=1$ when $R=P$. This requires the smoothness of $f$ and the local smoothness of the potentials $a$ and $V$.
\smallskip
        
		The estimates on the second summand are better and are $\cO(\crochetjap{b}^{r+1}\hbar^{2M-d})$ for all $M\in\N$.
		This is a consequence of the so-called Cwikel--Lieb--Rosenblum bounds (\cref{cor:CLR-semicl}) and Agmon estimates (\cref{lemma:Agmon}), that provide a bound on the number of negative eigenvalues of $\Hhb-\mu$ and the exponential decay of eigenfunctions of Schrödinger operators on the classically-forbidden region $\{x\in\R^d\::\: V(x)-\mu> \varepsilon\}$, respectively. All we need is that $a\in L^2_\loc(\R^d)$, $V\in L^1_\loc(\R^d)$ and that
		\begin{align*}
			\sup_{x\in\R^d}\abs{x}^{M}\abs{f(x)}e^{-\abs{x}}
		\end{align*}
		is finite for any $M\in\N$. This is satisfied by \cref{ass:a-V} and the polynomial growth of $f$ stated in \cref{ass:polynom-growth}.
        \item Our bounds are useful as long as the strength $b$ of the magnetic field does not become too large. Scaling limits with strong magnetic field are also of interest, but describe a different phenomenology, for example the gyrokinetic limit \cite{PR24} where $b \sim N$, or a semi-classical description of the  last, partially filled Landau level \cite{Per24} where $b \sim N\hbar$. A large number of different regimes have been identified by \cite{LSY92,LSY94,LSY94a} for atoms in magnetic fields and by \cite{LSY95} for quantum dots in magnetic fields.
\item 
	The assumption \cref{ass:polynom-growth} can be relaxed to less regular $f$ as long as $\hat{f}\in L^1(\R^d)$. This is a consequence of \cref{ass:polynom-growth} with $f(x)=e^{i\alpha\cdot x}$ for all $\alpha\in\R^d$ and the Fourier inversion formula.
	\end{enumerate}
\end{rmk}

\subsection{Commutator Estimates: State of the Art and Applications}\label{sec-intro:state_of_the_art}

 Estimates similar to the ones stated in \cref{thm:cb-princ} have been proved for spectral projections of standard Schrödinger operators with $f(x)=x_j$ for any $j\in\{1,\ldots,d\}$. This has been done for the harmonic oscillator \cite[Section 3.1]{Benedikter-2022effect} and \cite[Section 3]{Lafleche-2024optimal}, where the proof relies on explicit calculations with creation and annihilation operators. The general case has been treated in \cite{Fournais-Mikkelsen-2020} in presence of smooth confining potentials, and more recently in \cite{Cardenas-Lafleche-2025} for confining potentials of class $\cC^{1,1/2}$. We reproduce the spirit of the proof of \cite{Fournais-Mikkelsen-2020} that relies on the knowledge of the optimal Weyl law and multiscale techniques. The techniques presented in \cite{Cardenas-Lafleche-2025} rely on the Birman–Solomyak double integral approach.

As was mentioned above, this kind of bounds is motivated by the derivation of the Hartree--Fock equations for interacting fermions and usually constitute the assumption on the initial data for these equations (see \cref{sec-intro:propag-semicl-structure}); many recent results have been achieved in this respect for Schrödinger operators in the non-relativistic case \cite{BPS-2014-MFevol,BJPS-2016-MFevol-mixed}, in the relativistic case \cite{BPS-2014-MFevol-relativ}, and in large-volume limits \cite{FMA-2023-EDyn,FMA-2024-EDyn}. One of the main steps is the propagation of the bound \cref{eq:cb-princ-f(x)} for $f(x)=\frac{1}{1+\abs{\alpha}}e^{i\alpha\cdot x}$ uniformly with respect to $\alpha\in\R^d$.
Our estimates can be useful to study the same question in the magnetic case for sufficiently regular interaction potentials $W$ (compare \eqref{ass:w}).

Let us now mention other natural settings where this type of estimate appears, although
it is not possible to directly deduce applications of our results. Indeed, they require less
regularity on the function $f$ or that require replacing the trace norm by more general Schatten
norms.
We recall that the \emph{Wigner transform} of an operator $\gamma$ on $L^2(\R^d)$ to a function on phase space is defined by
\begin{equation*}
	\cW_\gamma:(x,\xi)\in\R^d\times\R^d\mapsto\int_{\R^d} e^{\frac i\hbar y\cdot\xi}\gamma\Big(x+\frac y2,x-\frac y2\Big) \di y\,,
\end{equation*}
where $\gamma(x,y)$ denotes the integral kernel of $\gamma$ in the point $(x,y)\in\R^d\times\R^d$ (provided it exists). One observes that
\begin{equation} \label{eq:qu-grad}
	\cW_{[x,\gamma]}=-i\hbar \nabla_\xi\cW_\gamma \quad\text{ and }\quad \cW_{[-i\hbar\nabla,\gamma]}=-i\hbar \nabla_x\cW_\gamma\;,
\end{equation}
which relates the commutator estimates $[\indic_{(-\infty,\mu}(\Hhb),x]$ and $[\indic_{(-\infty,\mu]}(\Hhb),P]$ to the regularity of the Wigner transform  \cite{Cardenas-Lafleche-2025,Lafleche-2024optimal}. Estimates for norms of the gradients \cref{eq:qu-grad} have been derived as a quantum analogue of Sobolev inequalities \cite{Laf24a}. Moreover they have been used to discuss the semiclassical properties of thermal states of fermionic many-body systems \cite{CLS23}.

Furthermore, bounds on Schatten norms of commutators of spectral projectors \cref{eq:cb-princ-f(x)} appear in the study of determinantal processes in Random Matrix theory. As explained in \cite{Deleporte-Lambert-2024Widom}, the commutator $[\indic_{(-\infty,\mu]}(\Hhb),\indic_\Omega]$ for a bounded domain $\Omega\subset\{x\in\R^d\: :\: V(x)\leq \mu\}$ is linked to the trace of the operator $g(\indic_\Omega\indic_{(-\infty,\mu]}(\Hhb) \indic_\Omega)$:
\begin{equation*}
    \norm{[\indic_{(-\infty,\mu]}(\Hhb),\indic_\Omega]}_{\rm HS}^2=\tr(g(\indic_{(-\infty,\mu]}(\Hhb) \indic_\Omega \indic_{(-\infty,\mu]}(\Hhb)))\;.
\end{equation*}
In particular, when $g:t\in[0,1]\mapsto t(1-t)$, the above object represents the variance of the process associated to the integral kernel of $\indic_{(-\infty,\mu]}(\Hhb)$ and also provides information on the trace of $s(\indic_\Omega\indic_{(-\infty,\mu]}(\Hhb) \indic_\Omega)$ with $s:t\in[0,1]\mapsto-t\log(t)-(1-t)\log(1-t)$. The asymptotics in the semiclassical limit $\hbar\to 0$ were conjectured by Widom in the seminal work \cite{Widom-1982} from 1982, and have been related in \cite{Gioev-Klich-2006} to the entanglement entropy of the non-interacting Fermi gas.
In the context of entanglement entropy, the asymptotics in the limit $\hbar\to 0$ of the   trace of operators $g(\indic_\Omega\indic_{(-\infty,\mu]}(\Hhb) \indic_\Omega)$ for certain functions $g:\R\to\R$  are still a subject of active investigation for the non-interacting Fermi gas in various settings: for rectangular domains $\Omega$ in \cite{Gioev-Klich-2006}, for analytic $g$ in \cite{Sobolev-2013Widom}, in a magnetic field without trapping potential \cite{Charles-Estienne-2020,Leschke-Sobolev-Spitzer-2021,Pfeiffer-Spitzer-2024}, and in a trapping potential \cite{Deleporte-Lambert-2024Widom} for continuous~$g$ and smooth boundary of $\Omega$.


\subsection{Second Main Result: Validity of the Hartree--Fock Evolution}\label{sec-intro:propag-semicl-structure}

We state now the second and third results of our paper.

 The propagation of the semiclassical structure from the initial data along the Hartree--Fock flow is crucial in the derivation of the Hartree--Fock equation from the many-body Schrödinger equation.

\begin{theorem-princ}[Propagation of the Semiclassical Structure]
	\label{thm:propag-semicl-magn}
	 	Let $a\in\cC^\infty(\R^d,\R^d)$ be such that $a_j$ is linear with respect to $x$ for any $j\in\{1,\ldots,d\}$, and let $V:\R^d\to\R$ be a polynomial function of degree $m$ such that $V(x)\to \infty$ as $\abs{x}\to \infty$. Let $N\in\N_{\geq 1}$ and let $\hbar=N^{-1/d}$.
		Let $\Hhb$ be the corresponding magnetic Schrödinger operator and $P:=-i\hbar\nabla-ba(x)$. Let the interaction potential $W:\R^d\to\R$ be an even function such that
	\begin{equation}\label{ass:w}
		\int_{\R^d}(1+\abs{p}^2)|\hat{W}(p)|\di p<\infty.
	\end{equation}
	Let us assume that the initial data $\omega_N$ is a projector of rank $N$ that satisfies
	\begin{equation}\label{eq:hyp-propag-semicl-struct}
		\begin{split}
		\norm{[\omega_N,x^\beta]}_1\leq C_{X,\beta}\hbar N \quad \forall \beta\in\N^d,\:\abs{\beta}\leq m,
		\quad\text{ and }\quad
		\norm{[\omega_N,P]}_1\leq C_P\hbar N,
		\end{split}
	\end{equation}
	and that the Hartree--Fock equation \cref{eq:HF-dyn} is well-posed.
	Then, the solution of the Hartree--Fock equation \cref{eq:HF-dyn} for the initial data $\omega_N$ satisfies for any $t\geq 0$ the bounds
	\begin{equation}\label{eq:propag-semicl-struct}
		\sup_{\alpha\in\R^d}\frac{1}{1+\abs{\alpha}}\norm{[\omega_{N,t},e^{i\alpha\cdot x}]}_1\leq C_1e^{\crochetjap{b}C_2 t}\hbar N,\quad
		\norm{[\omega_{N,t},P]}_1 \leq C_1e^{\crochetjap{b}C_2 t}\hbar N,
	\end{equation}
	with $C_1,C_2>0$ that depend on $C_X$, $C_P$, on $m$, on $C_{X,\beta}$ for $\beta\in\N^d$ such that $\abs{\beta}\leq m$ and on the norm given by \cref{ass:w}.
\end{theorem-princ}
\begin{rmk}\label{rmk:propag-semicl-magn}
	\begin{enumerate}
		\item We assume that $a$ is linear so that the magnetic field matrix $B=(B_{ij})_{1\leq i,j\leq d}:=(\partial_{x_i} a_j-\partial_{x_i} a_i)_{1\leq i,j\leq d}$ is constant. 
        An important example in dimension $d=2$ is the symmetric gauge magnetic vector potential $a(x_1,x_2)=\frac 12(-x_2,x_1)$ (that is also equivalent the Landau gauge magnetic vector potential $a(x_1,x_2)=(-x_1,0)$ by a gauge transformation).

    \item
	Examples of admitted external potentials are $\abs{x}^{2m}+V_{\rm pert}$, where $m\in\N_{\geq 1}$ and the perturbation $V_{\rm pert}$ is a smooth real valued function that can be polynomials. For the proof of this case, we only use the commutator bound with the $\beta$-power of the position operator for $\abs{\beta}\leq m$. We can also consider smooth functions $V_{\rm pert}$ such that $\partial_{x_j} V_{\rm pert}=f\circ g$ that $g:\R^d\to\R$ is a sum of a polynomial of degree smaller than $2m$ and $f:\R^d\to\R^d$ be a smooth and globally Lipschitz function. Indeed, all we need in the proof is be able to apply the Duhamel formula directly for the function $t\mapsto\norm{[\omega_{N,t},\nabla V]}_1$, that appears in the bound of the commutator with the momentum (see \cref{eq-proof:semic-str_moment} in the proof). And by \cite[Theorem 1.3.1]{Aleksandrov-Peller-2016}, we use the relation
		\begin{equation}\label{eq:com-lipschitz}
			\norm{[f(A),B]}_1\leq L\norm{[A,B]}_1 \;,
		\end{equation}
      that holds for any $L$-Lipschitz function $f$, any self-adjoint operator $A$, and any bounded linear self-adjoint operator $B$.
		\item
		By \cref{thm:cb-princ}, the condition \cref{eq:hyp-propag-semicl-struct} holds as long as $b=\cO_\hbar(1)$ when $\omega_N$ is the rank-$N$ spectral projection of $\Hhb$ on its $N$ lowest eigenvalues.
		\item The bound $\norm{[\omega_N,x]}_1\leq C_X\hbar N$ implies that
			$\sup_{\alpha\in\R^d}\frac{1}{1+\abs{\alpha}}\norm{[\omega_N,e^{i\alpha\cdot x}]}_1\leq \norm{[\omega_N,x]}_1$, 
		see, e.\,g., \cite[Corollary 1.3]{Fournais-Mikkelsen-2020} or \cite[Theorem 1.3.1]{Aleksandrov-Peller-2016}.

		\item The proof of \cref{thm:propag-semicl-magn} is an adaptation of \cite{BPS-2014-MFevol} and relies on the Duhamel formula for the functions $t\mapsto\sup_{\alpha\in\R^d}\frac{1}{1+\abs{\alpha}}\norm{[\omega_{N,t},e^{i\alpha\cdot x}}_1$ and $t\mapsto\norm{[\omega_{N,t},P}_1$. The proof of \cite{BPS-2014-MFevol} does not treat the presence of an external potential $V$, except for quadratic $V$.
        What changes in the present paper is that we need to propagate the estimates with a general external potential and the magnetic field.  This amounts to the propagation of commutators of the Hartree--Fock solution with monomials $x^\beta$ for $|\beta|$ lower than the degree of the polynomial $V$.
		\item
        The propagation of a semiclassical commutator structure formulated in terms of the gradients \cref{eq:qu-grad}, without magnetic field and without external potential but with a singular interaction potential, was discussed in \cite{CLS22a,CLS24}.
	\end{enumerate}
\end{rmk}

The validity of the time-dependent Hartree--Fock equation follows by the method of \cite{BPS-2014-MFevol}, given that the validity and the propagation of the semiclassical structure have already been established. 

\begin{theorem-princ}[Validity of the Hartree--Fock Equation in a Magnetic Field] \label{thm:hf}
Let $d \in \Nbb_{\geq 1}$, and let $W$ be as in \cref{thm:propag-semicl-magn}. 
Let $\omega_N$ be a sequence of projection operators satisfying the semiclassical commutator bounds \cref{eq:hyp-propag-semicl-struct}. Let $\Psi_N(0)$ be the Slater determinant which is up to a phase uniquely determined by $\omega_N$ and let $\gamma^{(1)}_{\Psi_N(t)}$ be the one-particle reduced density matrix associated to the solution $\Psi_N(t)$ of the time-dependent many-body Schr\"odinger equation \cref{eq:SE-dyn} in the mean-field scaling regime \cref{eq:mf}, with initial data $\Psi_N(0)$. Let $\omega_{N,t}$ be the solution of the time-dependent Hartree--Fock equation \cref{eq:HF-dyn} with the same vector potential, external potential, and interaction potential as in the Schr\"odinger equation, and with initial data $\omega_{N,0} = \omega_N$.

Then there exist constants $C,c_1,c_2 >0$ such that for all $t \in \Rbb$ and for $N \in \Nbb$ we have 
\begin{align*}
\lVert \gamma^{(1)}_{\Psi_N(t)} - \omega_{N,t} \rVert_{\tr} \leq C \hbar^{-1/2} e^{c_1 e^{c_2 t}} \qquad \text{and} \qquad
\lVert \gamma^{(1)}_{\Psi_N(t)} - \omega_{N,t} \rVert_{\HS} \leq C  e^{c_1 e^{c_2 t}} \;.
\end{align*}
\end{theorem-princ}
\begin{rmk}
   The trace norm of the one-particle density matrix is conserved both by the many-body Schr\"odinger evolution and the Hartree--Fock evolution. Thus by \cref{cor:CLR-semicl} for the initial data, we conclude that the density matrices separately are of order
   \begin{equation}\label{eq:orders}
   \lVert \omega_{N,t}\rVert_\tr \sim \hbar^{-d}\;, \qquad \lVert \gamma^{(1)}_{\Psi_N(t)} \rVert_\tr \sim \hbar^{-d} \;.
   \end{equation}
   Since $\omega_N$ is a projection operator, by H\"older's inequality for Schatten norms, we obtain that the Hilbert--Schmidt norm can be trivially estimated by
   \begin{equation}\label{eq:orders2}
   \norm{\omega_{N,t}}_\HS = \sqrt{\tr (\omega_{N,t}^* \omega_{N,t})} \leq \sqrt{\lVert \omega_{N,t}\rVert_\op \lVert \omega_{N,t}\rVert_\tr} \leq \hbar^{-d/2}\;,
   \end{equation}
   Thus \cref{thm:hf} shows that the trace norm estimate gains a factor $\hbar^{d-1/2}$ as $\hbar \to 0$ compared to \cref{eq:orders}, whereas the Hilbert--Schmidt norm estimate gains only $\hbar^{d/2}$ over \cref{eq:orders2}. In particular, our estimates are significantly stronger than those obtained earlier in \cite{FM24} for fermions in a magnetic field in two dimensions: the bounds proved there would correspond in our notation to an estimate of the form $\norm{\gamma_{\Psi_N(t)}^{(1)} - \omega_{N,t}}_\HS \le C \, \hbar^{-1} \, t$, which is linear (rather than super-exponential) in $t$ but is only meaningful for times much shorter than the semiclassical scale.
\end{rmk}
\begin{rmk}
\Cref{thm:hf} can be generalized in various directions. These have been discussed already in \cite{BPS-2014-MFevol} and carry over to the case of fermions in magnetic fields.
\begin{enumerate}
\item One can admit initial data that apart from the $N$ particles in the orbitals of the spectral decomposition of $\omega_N$ also contain a uniformly bounded number of order of particles in other orbitals.
\item One can show that the $k$-particle reduced density matrices $\gamma^{(k)}_{\Psi_N(t)}$ converge to a Wick formula in terms of $\omega_{N,t}$, for any fixed value of $k \in \Nbb$.
\end{enumerate}
Both generalizations come at the cost of a weaker bound for the trace norm, gaining only an order $\hbar^{d/2}$ instead of $\hbar^{d-1/2}$.
\end{rmk}

%

\section{Semiclassical Structures for Spectral Projectors and Proof of \cref{thm:cb-princ}}\label{sec:proof-main-result}

\subsection{The Local Case}\label{sec:cb-local-estimates}

	The following assumption defines a self-adjoint operator which is locally equal to a magnetic Schrödinger operator with smooth and compactly supported magnetic and external potentials. It appeared in \cite{Sobolev-1995}, then in \cite{Fournais-Mikkelsen-2020}.
	
	\begin{assump}\label{ass:loc-magn}
		The operator $H_{\hbar,b}$ on $L^2(\R^d)$ and the open set $\Omega\subset\R^d$ satisfy the following conditions
		\begin{itemize}
			\item[(i)]  $H_{\hbar,b}$ is self-adjoint and semi-bounded from below in $L^2(\R^d)$,
			\item[(ii)] there is an inclusion $\cC_0^\infty(\Omega)\subset\cD(\Hhb)$,
			\item[(iii)] there exist smooth real-valued functions $(a_{\loc})_j\in \test{\R^d}$ for any $j\in\{1,\ldots,d\}$ and $V_{\loc}\in\test{\R^d,\R}$, such that $\Hhb u=\Hloc u$ for any $u\in\test{\Omega}$, where $\Hloc:=(-i\hbar\nabla-ba_{\loc}(x))^2+V_{\loc}(x)$.
		\end{itemize}
	\end{assump}
	
	We state a local version of our main theorem. It is the analogue of \cite[Theorem 3.6]{Fournais-Mikkelsen-2020} for magnetic Schrödinger operators.

	\begin{theorem}[{Local magnetic commutators bounds}]\label{thm:cb-loc-gene}
		Let $C_0>0$ and  $h_0>0$. Let $b>0$ and $\hbar>0$.
		Assume $(\Hhb,\Omega)$ satisfy \cref{ass:loc-magn} for $a_{\loc}\in\test{\R^d,\R^d}$ and $V_{\loc}\in\test{\R^d}$. We define the magnetic momentum $P:= -i\hbar\nabla-b a_{\loc}(x)$. Then, for any $\varphi\in\test{\Omega}$, there exists $C>0$ such that for any $j\in\{1,\ldots,d\}$, any $b\in(0,C_0\hbar^{-1})$, any $h\in(0,h_0]$, and any $m\geq 0$ and any $r\in\{0,1\}$, one has
		\begin{equation}
			\norm{[\indic_{(-\infty,0]}(\Hhb),\varphi P_j^r]}_1\leq C\crochetjap{b}^{1+r}\hbar^{1-d}.
		\end{equation}
	\end{theorem}
	
	To prove the above theorem, we first prove the commutators bounds under a non-criticality condition (see \cref{thm:cb-loc-non-critical-cond}), then we use multiscale techniques from \cite{Sobolev-1995,Fournais-Mikkelsen-2020} to get rid of this condition.

\subsubsection{Preliminaries and auxiliary local bounds}\label{subsec:local-aux}

We state in this section trace norm estimates related to $\varphi P_j^r f(\Hhb)$ for $f\in\test{\R}$, $\varphi\in\test{\R^d}$, and $r\in\{0,1\}$. These estimates are used to prove \cref{thm:cb-loc-gene}. They correspond to the magnetic version of the auxiliary bounds of \cite{Fournais-Mikkelsen-2020}, that are based on ideas in \cite{Sobolev-1995}.

We choose to present basic proofs that do not rely on pseudodifferential calculus (see, e.\,g., \cite{Dimassi-Sjostrand-1999}), although this could shorten the proofs. It is possible to obtain basic proofs since we use the structure of the local Schrödinger operators and identities of the resolvent.


We start by recalling the \emph{diamagnetic inequality}, which allows us to compare the magnetic Laplacian with the usual Laplacian. We use this bound in order to deduce similar properties as for the usual Schrödinger operators (see \cite{Lewin-TS_et_MQ,Lieb-Seiringer-2010Stability}).

\begin{lemma}[{Diamagnetic inequality, \cite[Section 4.4]{Lieb-Seiringer-2010Stability}}]\label{lemma:diamagnetic-ineq}
	Let $a:\R^d\to\R^d$ such that $a_j\in L^2_\loc(\R^d))$ for any $j\in\{1,\ldots,d\}$.
	Then, for any $u\in\dom(\mathfrak{q}_{\hbar,b})$,
	\begin{equation}\label{eq:diamagnetic-ineq_dir}
		\abs{i\hbar\nabla\abs{u}(x)}\leq\abs{(i\hbar\nabla+ba(x))u(x)} \text{ a.e.}
	\end{equation}
	And for any $\lambda>0$, any $\hbar>0$, any $b>0$ the integral kernel of the resolvent satisfies
	\begin{equation}\label{eq:diamagnetic-ineq_res-kin}
		((-i\hbar\nabla-ba(x))^2+\lambda)^{-1}(x,y)\leq (-\hbar^2\Delta+\lambda)^{-1}(x,y) \text{ a.e. in }\R^d\times\R^d.
	\end{equation}
\end{lemma}

We also state useful identities on the magnetic momentum.

\begin{lemma}\label{lemma:prop-magn-momentum}
	Let $b,\hbar>0$, $a\in\cC^2(\R^d,\R^d)$ and $P:=-i\hbar\nabla-ba(x)$.
	Then, for all $k,j\in\{1,\ldots,d\}$ we have
	\begin{equation}\label{eq:com-momenta}
		[P_k,P_j]=-i b\hbar(\partial_{x_j}a_k-\partial_{x_k}a_j)=i\hbar bB_{kj}\,
	\end{equation}
	and
	\begin{equation}\label{eq:com-momenta-square}
		[P_k^2,P_j]
		%
		%
		=i\hbar b (2P_kB_{kj} -i\hbar\partial_{x_k}B_{kj})\;.
	\end{equation}
\end{lemma}
 
We rely on the almost analytic extension of a function and a Cauchy integration formula to define the operator $f(H)$ for self-adjoint operators $H$ and $f\in\test{\R}$.

\begin{defi}\label{def:almost-anal}
	An \emph{almost analytic extension} of $f\in\test{\R}$ is a function $\tilde{f}\in\test{\C}$ such that $\tilde{f}_{|_{\R}}=f$ and such that for all $N\in\N$, there exists $C_N>0$ such that for all $z\in\C$
	\[ \abs{\bar{\partial}\tilde{f}(z)}\leq C_N\abs{\Im(z)}^N ,\]
	where $\bar{\partial}:=\frac 12(\partial_x+i\partial_y)$.
\end{defi}

\begin{theorem}[{Helffer--Sj\"ostrand formula, \cite[Theorem 8.1.]{Dimassi-Sjostrand-1999}}]\label{thm:Helffer-Sjostrand}
	Let $H$ be a self-adjoint operator on a Hilbert space and $f\in\test{\R}$. Then,
	\begin{equation*}
		f(H)=\frac 1\pi\int_\C \bar{\partial}\tilde{f}(z)(H-z)^{-1}L(\di z),
	\end{equation*}
	where $L(\di z)=\di x \di y$ is the Lebesgue measure on $\C$ and $\tilde{f}\in\test{\C}$ is an almost analytic extension of $f\in\test{\R}$.
\end{theorem}

We provide the bounds on the trace norm of the product of the resolvent of the magnetic Schrödinger operator and the magnetic momentum operator. 

\begin{lemma}[{\cite[Lemma 3.5]{Sobolev-1995}}]\label{lemma:aux-loc_0}
	Let $\hbar\in(0,\hbar_0]$ for $\hbar_0>0$, $a\in L^2_\loc=(\R^d,\R^d)$, $V\in L^\infty(\R^d)$. We define the Schrödinger operator $H=(-i\hbar\nabla-ba(x))^2+V$ that acts on $L^2(\R^d)$ and the magnetic momentum operator $P:=-i\hbar\nabla-ba(x)$. Then, for any $\lambda>0$ in the resolvent set of $H$ (for instance $\lambda\geq 1+\normLp{V_-}{\infty}{(\R^d)}$) we have
	\begin{equation*}
		\norm{P_j^r (H+\lambda)^{-1/2}}_1\leq \sqrt{2}\lambda^{\frac{r-1}2} \;.
	\end{equation*}
\end{lemma}

For more general resolvents we have the following bounds, that allow us to track optimally the dependency with respect to the parameters $b>0$.

\begin{lemma}\label{lemma:aux-loc_00}
	Let $\hbar\in(0,\hbar_0]$ for $\hbar_0>0$, $a\in L^2_\loc(\R^d,\R^d)$, and $V\in L^\infty(\R^d)$. We define the Schrödinger operator $H=(-i\hbar\nabla-a(x))^2+V$ that acts on $L^2(\R^d)$ and the magnetic momentum operator $P:=-i\hbar\nabla-ba(x)$. Then, for any $z$ in the resolvent set of $H$, there exists $C>0$ such that for any $\hbar>0$, any $b\geq 0$, any $j,k\in\{1,\ldots,d\}$ and any $r,m\in\{0,1\}$
	\begin{equation*}
		\norm{P_j^r (H-z)^{-1} P_k^m}_1\leq C\frac{\crochetjap{z}}{\dist(z,\C\setminus\spec(H))}\;.
	\end{equation*}
	The constant $C$ depends only on the norm $\normLp{V_-}{\infty}{(\R^d)}$.
\end{lemma}

\begin{proof}
	Let $z\in\C\setminus\spec(H)$ and $\lambda\geq 1+2\normLp{V_-}{\infty}{(\R^d)}$. The resolvent identities \cite[(3.5), (3.6)]{Sobolev-1995} give that
	\begin{align*}
		(H-z)^{-1} = (H+\lambda)^{-1/2}(1+(\lambda+z)(H-z)^{-1})(H+\lambda)^{-1/2}
	\end{align*}
	and there exists a constant $C$ independent of $z$ and $\lambda$ such that
	\begin{align*}
		\norm{(1+(\lambda+z)(H-z)^{-1})}_\op\leq C\frac{\lambda+\crochetjap{z}}{\dist(z,\C\setminus\spec(H))}\;.
	\end{align*}
	By the H\"older inequality and \cref{lemma:aux-loc_0}
	\begin{align*}
		\norm{P_j^r(H-z)^{-1}P_k^m}_1
		&\leq \norm{P_j^r(H+\lambda)^{-1/2}}_1\norm{(1+(\lambda+z)(H-z)^{-1})}_\op\norm{(H+\lambda)^{-1/2}P_k^m}_1
		\\&\leq 2C\lambda^{\frac {r+m}2}\frac{\crochetjap{z}}{\dist(z,\C\setminus\spec(H))} \;.	\qedhere
	\end{align*}
\end{proof}

We next state \cite[Lemma 3.6]{Sobolev-1995}: for simplicity, we take non-real spectral parameters $z$ and we replace $\dist(z,[-1+\min_x V(x),+\infty))$ by $\abs{\Im(z)}$ in the denominator of the right-hand side of the bound in the statement.

\begin{lemma}[{\cite[Lemma 3.6]{Sobolev-1995}}]\label{lemma:aux-loc_000}
	Let $\hbar\in(0,\hbar_0]$ for $\hbar_0>0$, $b>0$, $a\in L^2_\loc=(\R^d,\R^d)$, $V\in L^\infty(\R^d)$. We define the Schrödinger operator $H=(-i\hbar\nabla-ba(x))^2+V$ that acts on $L^2(\R^d)$ and the magnetic momentum operator $P:=-i\hbar\nabla- ba(x)$. Let $\varphi\in\test{\R}$ and $\psi\in\cC^\infty(\R)$ be a bounded function such that all its derivatives are bounded, and such that there exists $c>0$ such that
	\begin{equation}\label{eq:disjoint-supp}
		\dist(\supp(\varphi),\supp(\psi))\geq c.
	\end{equation}
	Then, for any $N>\frac d2$, there exists $C_N>0$ such that for any $\hbar\in(0,\hbar_0]$, any $b>0$, any $r,m\in\{0,1\}$, any $j,k\in\{1,\ldots,d\}$ and any $z\in\C\setminus\R$ 
	\begin{equation*}
		\norm{\varphi P_j^r(H-z)^{-1}P_j^m\psi}_1\leq C_N\hbar^{2N-d}\frac{\crochetjap{z}^{N+\frac{d+r+m}2}}{\abs{\Im(z)}^{1+2N}}\;.
	\end{equation*}
	The constant $C_N$ depends on $N$, the dimension $d$, the constant $c$ and the norms $\normLp{\partial^\alpha\varphi}{\infty}{(\R^d)}$ and $\normLp{\partial^\alpha\psi}{\infty}{(\R^d)}$ of a finite number of derivatives (that depends on $N$ and $c$).
\end{lemma}

The following lemma is magnetic version of \cite[Lemma 2.7]{Fournais-Mikkelsen-2020}. The case $r=0$ has been treated in \cite[Lemma 3.9]{Sobolev-1995}.

\begin{lemma}\label{lemma:aux-loc_1}
	Let $\Hloc$ be a magnetic Sch\"rodinger operator that acts on $L^2(\R^d)$  and let $P:=-i\hbar\nabla-ba_{\loc}(x)$,  with $\hbar,b>0$, $a_{\loc}\in\test{\R^d,\R^d}$, and $V_{\loc}\in\test{\R^d,\R}$.
	Let $f\in\test{\R}$, $\varphi\in\test{\R^d}$, $r\in\{0,1\}$ and $\hbar_0>0$.
	\begin{itemize}
		\item[(i)] There exists $C>0$ such that for any $\hbar\in(0,\hbar_0]$, any $b>0$ and any $j\in\{1,\ldots,d\}$,
		\begin{equation}\label{eq:aux-loc_1-1}
			\norm{\varphi P_j^r f(\Hloc)}_1\leq C.
		\end{equation}
		This constant $C$ depends on the dimension $d$, on the norms $\normLp{(V_\loc)_-}{\infty}{(\R^d)}$ and $\normLp{\partial^\alpha\varphi}{\infty}{(\R^d)}$ for finitely many $\alpha\in\N^d$.
		\item[(ii)] Let $\psi\in\cC^\infty(\R^d)$ be bounded such that its derivatives are bounded and such that there exists $c>0$ such that
		\[ \dist(\supp(\varphi),\supp(\psi))>c .\]
		 Then, for any $N\in\N$, there exists $C_N>0$ such that for any $\hbar\in(0,\hbar_0]$,  any $b>0$ and any $j\in\{1,\ldots,d\}$,
		 \begin{equation}\label{eq:aux-loc_1-2}
		 	\norm{\varphi P_j^r f(\Hloc)\psi}_1\leq C_N\hbar^{2N-d}\;.
		 \end{equation}
		 The constant $C_N$ depends on $N$, the dimension $d$, the norm $\normLp{\partial^\alpha\varphi}{\infty}{(\R^d)}$, the constant $c$, the norm $\normLp{\partial^\alpha\varphi}{\infty}{(\R^d)}$,  $\normLp{\partial^\alpha\psi}{\infty}{(\R^d)}$ for finite number of derivatives (that depends on $N$ and $c$) and on $\supp(f)$.
	\end{itemize}
\end{lemma}

\begin{proof}
	\item[\quad(i)] 	Let $\lambda:=1+2\normLp{V_-}{\infty}{(\R^d)}$. We apply the H\"older inequality, \cref{lemma:aux-loc_0}, and by the continuous functional calculus 
	\begin{align*}
		\norm{\varphi P_j^r f(\Hloc)}_1&\leq 
		\norm{\varphi P_j^r (\Hloc+\lambda)^{-1/2}}_1\norm{(\Hloc+\lambda)^{1/2} f(\Hloc)}_\op
		\\&\leq C\lambda^{\frac{r-1}2}\;.
	\end{align*}
	\item[\quad(ii)] 
	 By the Helffer--Sj\"ostrand formula (\cref{thm:Helffer-Sjostrand}) and \cref{lemma:aux-loc_000} applied to the Schrödinger operator $\Hloc$, and the definition of an almost analytic extension (\cref{def:almost-anal}), we have for $N\geq 0$ that 
	\begin{align*}
		\norm{\varphi P_j^r f(\Hloc)\psi}_1 &=\norm{\frac 1\pi\int_\C  \bar{\partial}\tilde{f}(z) \varphi P_j^r(\Hloc-z)^{-1} \psi L(\di z)}_1
		\\&\leq \frac 1\pi\hat{C}_N\int_{\supp(\tilde{f})}\abs{\Im(z)}^{2N+1}\norm{\varphi P_j^r (\Hloc-z)^{-1} \psi}_1 L(\di z)
		\\&\leq \frac 1\pi \tilde{C}_N\left(\int_{\supp(\tilde{f})}\crochetjap{z}^{N+\frac{d+r}2}L(\di z)\right)\hbar^{2N-d}\;.	\qedhere
	\end{align*}
\end{proof}

The next lemma is a modified version of \cite[Theorem 3.12]{Sobolev-1995}: the case $r=0$ is \cite[Theorem 3.12]{Sobolev-1995} 
and the case $r=1$ corresponds to \cite[Theorem 2.8]{Fournais-Mikkelsen-2020} with the magnetic momentum instead of $-i\hbar\nabla$. It states that in trace norm $\varphi P_j^rf(\Hhb)$ can be replaced by $\varphi P_j^rf(\Hloc)$ in the semiclassical limit. 

\begin{lemma}\label{lemma:aux-loc_2}
	Assume that $(\Hhb,B(0,4R))$ satisfies \cref{ass:loc-magn} for $R>0$. Let  $a_{\loc}\in\test{\R^d,\R^d}$ and $V_{\loc}\in\test{\R^d,\R}$ be the functions that define the local Schrödinger operator $\Hloc$ of this assumption and let $P:=-i\hbar\nabla-ba_{\loc}(x)$.
\begin{itemize}
	\item[(i)] 
	Then, for any $f\in\test{\R}$, $\varphi\in\test{B(0,3R)}$, $\hbar_0>0$ and $N\in\N$, there exists $C_N>0$ such that for any $r\in\{0,1\}$, any $j\in\{1,\ldots,d\}$ and any $\hbar\in(0,\hbar_0]$,
	\begin{equation}\label{eq:aux-loc_2-1}
		\norm{\varphi P_j^r(f(\Hhb)-f(\Hloc))}_1\leq C_N\hbar^{2N-d}\;.
	\end{equation}
	\item[(ii)] 
	In particular, there exists $C>0$ such that for any $r\in\{0,1\}$, any $j\in\{1,\ldots,d\}$ and any $\hbar\in(0,\hbar_0]$,
	\begin{equation}\label{eq:aux-loc_2-2}
		\norm{\varphi P_j^rf(\Hhb)}_1\leq C\hbar^{-d}\;.
	\end{equation}
\end{itemize}
The constants $C_N$ and $C$ depend on the dimension $d$, the norm $\normLp{\partial^\alpha\varphi}{\infty}{(\R^d)}$, the constant $c$, the norm $\normLp{\partial^\alpha\varphi}{\infty}{(\R^d)}$,  $\normLp{\partial^\alpha\psi}{\infty}{(\R^d)}$ for a finite number of derivatives (that depends on $N$ and $c$) and on $\supp(f)$.
\end{lemma}

\begin{proof} We adapt the proof of \cite[Lemma 4.3]{Mikkelsen-2023magn}, that relies on the Helffer--Sj\"ostrand formula, \cref{lemma:aux-loc_000}, and \cref{lemma:aux-loc_1}.
\item[\quad(i)]	Let $\tilde{f}$ an almost extension of $f$, then by the Helffer--Sj\"ostrand formula
	\begin{align*}
		\varphi P_j^r(f(\Hhb)-f(\Hloc)) &= \frac{1}{\pi}\int_\C\bar{\partial}\tilde{f}(z)\varphi P_j^r\big((\Hhb-z)^{-1}-(\Hloc-z)^{-1}\big) L(\di z) .
	\end{align*}
	Moreover, by construction we have
	\begin{align*}
		\dist(\supp(\varphi),\partial B(0,4))\geq R.
	\end{align*}
	Let $z\in\C\setminus(\spec(\Hhb)\cup\spec(\Hloc))$.
	Let $\psi\in\test{\R,[0,1]}$ be such that $\psi=1$ on the set $	\{x\in\R^d\: :\: \dist(x,\supp(\varphi))\leq R/4\}$ and such that $\supp(\psi)\subset
	\{x\in\R^d\: :\: \dist(x,\supp(\varphi))\leq 3R/4\}$, so that
	\begin{align*}
		\dist(\supp(\varphi),\supp(1-\psi))\geq R/4 \;,
	\end{align*}
	and in particular $\varphi(1-\psi)=0$ and $\varphi=\varphi\psi$.
	We then write
	\begin{align*}
		\varphi P_j^r&\big((\Hhb-z)^{-1}-(\Hloc-z)^{-1}\big) 
		\\&= \varphi P_j^r\big((\Hhb-z)^{-1}-(\Hloc-z)^{-1}\psi\big)-\varphi P_j^r(\Hloc-z)^{-1}(1-\psi)
		\\&=\varphi\big(\psi P_j^r(\Hhb-z)^{-1}-P_j^r(\Hloc-z)^{-1}\psi\big)-\varphi P_j^r(\Hloc-z)^{-1}(1-\psi)\;.
	\end{align*}
	On the one hand, \cref{lemma:aux-loc_000} provides $C_N>0$ such that one has for any $N\geq 0$
	\begin{align}\label{eq-proof:aux-loc_2-extra-1}
		\norm{\varphi P_j^r(\Hloc-z)^{-1}(1-\psi)}_1\leq C_N\hbar^{2N-d}\frac{\crochetjap{z}^{N+\frac{d+r}2}}{\abs{\Im(z)}^{1+2N}}\;.
	\end{align}
%
	Let us show that
	\begin{align}\label{eq-proof:aux-loc_2-extra-2}
		\norm{\psi P_j^r(\Hhb-z)^{-1}-P_j^r(\Hloc-z)^{-1}\psi}_1\leq \tilde{C}_N\hbar^{2N+1-d}\frac{\crochetjap{z}^{N+\frac{d+1+r}2}}{\abs{\Im(z)}^{2+2N}}\;.
	\end{align}
	Notice that for all $\alpha\in\N^d$ such that $\abs{\alpha}\geq 1$
	\begin{align}\label{eq-proof:aux-loc_2-supp-disj}
		\dist(\supp(\varphi),\supp(\partial_x^\alpha\psi))\geq R/4\;.
	\end{align}
	This implies
	\begin{align*}
		\varphi[\psi, P_j^r]=\delta_{r,1}i\hbar\varphi(\partial_{x_j}\psi) =0\;,
	\end{align*}
	and
	\begin{align*}
		\varphi\Big(\psi P_j^r&(\Hhb-z)^{-1}-P_j^r(\Hloc-z)^{-1}\psi \Big)
		\\&=\underset{=0}{\underbrace{\varphi[\psi, P_j^r]}}(\Hhb-z)^{-1}+\varphi P_j^r\Big(\psi(\Hhb-z)^{-1}-(\Hloc-z)^{-1}\psi\Big)\;.
	\end{align*}
%
	Let $\tilde{\psi}\in\test{B(0,4R),[0,1]}$ such that $\tilde{\psi}=1$ on $\supp(\psi)$.
 	We apply the same calculations as in \cite[Lemma 4.3]{Mikkelsen-2023magn}
	\begin{align*}
		\psi (\Hhb-z)^{-1}-(\Hloc-z)^{-1}\psi
		&= 	\psi (\Hhb-z)^{-1}-(\Hloc-z)^{-1}\psi\tilde{\psi}
		\\&= 	\psi (\Hhb-z)^{-1}-(\Hloc-z)^{-1}\psi(\Hloc-z)(\Hhb-z)^{-1}\tilde{\psi}
		\\&= (\Hloc-z)^{-1}[(\Hloc-z),\psi ](\Hhb-z)^{-1}\tilde{\psi}\;.
	\end{align*}
%
	We also notice that for any $z\in\C$, one has $[(\Hloc-z),\psi P_j^r]=[\Hloc,\psi P_j^r]$ and that
	\begin{align}\label{eq:com-H-phi}
		[\Hloc,\psi] 
		&=\sum_{k=1}^d[P_k^2,\psi]
		=\sum_{k=1}^d([P_k,\psi]P_k+P_k[P_k,\psi])\nonumber
		\\&= \sum_{k=1}^d\left(-2i\hbar P_k(\partial_{x_k}\psi)-\hbar^2(\partial_{x_k x_k}^2\psi)\right)\;.
	\end{align}
%
	We now use the H\"older inequality, the identity \cref{eq:com-H-phi} and \cref{lemma:aux-loc_000}, which we can apply since \cref{eq-proof:aux-loc_2-supp-disj} holds, to deduce that
	\begin{align*}
		&\norm{\varphi P_j^r\big(\psi(\Hhb-z)^{-1}-(\Hloc-z)^{-1}\psi\big)}_1 \nonumber
		\\&\quad=\norm{\sum_{k=1}^d\varphi P_j^r(\Hloc-z)^{-1}\Big(-i\hbar P_k(\partial_{x_k}\psi)-\hbar^2(\partial_{x_k}^2\psi)\Big)(\Hhb-z)^{-1}\tilde{\psi}}_1  \nonumber
		\\&\quad\leq \hbar \norm{(\Hhb-z)^{-1}}_\op\sum_{k=1}^d\Big(\norm{\varphi P_j^r(\Hloc-z)^{-1} P_k(\partial_{x_k}\psi)}_1
		+\hbar\norm{\varphi P_j^r(\Hloc-z)^{-1}(\partial_{x_k}^2\psi)}_1
		\Big) \nonumber
		\\&\quad\leq \frac\hbar{\abs{\Im(z)}} d\hat{C}_N \hbar^{2N-d}\left(\frac{\crochetjap{z}^{N+\frac{d+1}2}}{\abs{\Im(z)}^{1+2N}}+\hbar\frac{\crochetjap{z}^{N+\frac{d}2}}{\abs{\Im(z)}^{1+2N}}\right) \nonumber
		\\&\quad\leq \tilde{C}_N\hbar^{2N+1-d}\frac{\crochetjap{z}^{N+\frac{d+1+r}2}}{\abs{\Im(z)}^{2+2N}}\;.
	\end{align*}
	As a conclusion, by using the fact that $\tilde{f}$ is almost analytic and by recombining \cref{eq-proof:aux-loc_2-extra-1}, \cref{eq-proof:aux-loc_2-extra-2}, we obtain \cref{eq:aux-loc_2-1}
	\begin{align*}
		\norm{\varphi P_j^r(f(\Hhb)-f(\Hloc))}_1 &= C_N\hbar^{2N-d}\int_{\supp(\tilde{f})}\crochetjap{z}^{N+\frac{d+1+r}{2}} L(\di z) \;.
	\end{align*}
	\item[\quad(ii)] We deduce \cref{eq:aux-loc_2-2} from \cref{eq:aux-loc_1-2}, \cref{eq:aux-loc_2-1}, and the triangle inequality.
\end{proof}

Below is a local version of \cref{thm:cb-loc-gene}, the magnetic version of \cite[Lemma 3.1]{Fournais-Mikkelsen-2020}.

\begin{lemma}\label{lemma:aux-loc_3}
	Assume that $(\Hhb,B(0,4R))$ satisfy \cref{ass:loc-magn} for $R>0$. Let  $a_{\loc}\in\test{\R^d,\R^d}$ and $V_{\loc}\in\test{\R^d,\R}$ that define the previously mentioned local Schrödinger operator $\Hloc$ and let $P:=-i\hbar\nabla- ba_{\loc}(x)$.
		Then, for any $f\in\test{\R}$, $\varphi\in\test{B(0,3R)}$, $\hbar_0>0$, there exists $C>0$ such that for any $\hbar\in(0,\hbar_0]$, any $b>0$, any $r\in\{0,1\}$, any $j\in\{1,\ldots,d\}$
	\begin{equation*}
		\norm{[f(\Hhb),\varphi P_j^r]}_1\leq C\crochetjap{b}^r\hbar^{1-d}\;.
	\end{equation*}
	This constant depends on the dimension $d$, $R>0$, the norms $\normLp{V_-}{\infty}{(\R^d)}$, of $\normLp{\partial^\alpha\varphi}{\infty}{(\R^d)}$ for finitely many $\alpha\in\N^d$ and $\normLp{\partial^kf}{\infty}{(\R)}$ for all $k\in\N$. 
\end{lemma}

\begin{proof}
	By \cref{lemma:aux-loc_2}, it is sufficient to prove the desired bound with the local operator $\Hloc$ instead of $\Hhb$. Then, for any $N\in\N$ there exists a constant $C_N>0$  such that
	\begin{align*}
		\norm{[f(\Hhb),\varphi P_j^r]}_1&\leq \norm{[f(\Hloc),\varphi P_j^r]}_1+\norm{[f(\Hhb)-f(\Hloc),\varphi P_j^r]}_1
		\\&\leq \norm{[f(\Hloc),\varphi P_j^r]}_1+2\norm{\varphi P_j^r(f(\Hhb)-f(\Hloc))}_1
		\\&\leq \norm{[f(\Hloc),\varphi P_j^r]}_1+C_N \hbar^{2N-d}\;.
	\end{align*}
	Let us define $g\in\test{\R,[0,1]}$ such that $g=1$ on the support on $f$, so that
	\begin{align*}
		[f(\Hloc),\varphi P_j^r] 
		&=g(\Hloc)f(\Hloc)\varphi P_j^r-\varphi P_j^r g(\Hloc)f(\Hloc)
%
		%
		\\&= g(\Hloc)[f(\Hloc),\varphi P_j^r]+[g(\Hloc),\varphi P_j^r]f(\Hloc)\;.
	\end{align*}
	Then, by the triangle inequality and the cyclicity of the trace norm
	\begin{align*}
		\norm{[f(\Hloc),\varphi P_j^r]}_1\leq \norm{g(\Hloc)[f(\Hloc),\varphi P_j^r]}_1+\norm{f(\Hloc)[g(\Hloc),\varphi P_j^r]}_1\;.
	\end{align*}
	It remains to bound each term on the right-hand side. Actually, we will bound one of them as the analysis is similar for the other one, since we do not use the relation between $f$ and $g$. To do so, we introduce $\tilde{\varphi}\in\test{\R,[0,1]}$ such that $\tilde{\varphi}=1$ on $\supp(\varphi)$. In particular, $\supp(\varphi)$ and $\supp(1-\tilde{\varphi})$ are disjoint. Then,
	\begin{equation}\label{eq-proof:aux-loc_3}
		\begin{split}
			g(\Hloc)[f(\Hloc),\varphi P_j^r] 
			&= 	g(\Hloc)\tilde{\varphi}[f(\Hloc),\varphi P_j^r]+g(\Hloc)(1-\tilde{\varphi})[f(\Hloc),\varphi P_j^r]
			\\&= 	g(\Hloc)\tilde{\varphi}[f(\Hloc),\varphi P_j^r]+g(\Hloc)(1-\tilde{\varphi})f(\Hloc)\varphi P_j^r\;.
		\end{split}
	\end{equation}
	By the triangle inequality in \cref{eq-proof:aux-loc_3}, the H\"older inequality, Lemma \ref{lemma:aux-loc_1} applied to the trace norms, and by \cref{eq-proof:aux-loc_3-bis}, one has for any $N\in\N$
	\begin{align*}
		\norm{	g(\Hloc)[f(\Hloc),\varphi P_j^r] }_1 &\leq \norm{g(\Hloc)\tilde{\varphi}[f(\Hloc),\varphi P_j^r]}_1+\norm{g(\Hloc)(1-\tilde{\varphi})f(\Hloc)\varphi P_j^r}_1
		\\&\leq \norm{\tilde{\varphi}g(\Hloc)}_1\norm{[f(\Hloc),\varphi P_j^r]}_\op+\norm{g(\Hloc)}_\op\norm{(1-\tilde{\varphi})f(\Hloc)\varphi P_j^r}_1
		&
		\\&\leq C\norm{[f(\Hloc),\varphi P_j^r]}_\op+C_N\hbar^{2N-d}\;.
	\end{align*}
	It only remains to establish the bound 
	\begin{equation}\label{eq-proof:aux-loc_3-bis}
		\norm{[f(\Hloc),\varphi P_j^r]}_\op \leq C\crochetjap{b}^r\hbar \;,
	\end{equation}
	where $C>0$ only depends on the supremum norm of the derivatives of $\partial^\alpha_x\varphi$, $\partial^\beta_x a$ for finitely many $\alpha,\beta\in\N^d$ such that $\abs{\beta}\geq 1$ and $\partial_{x_j}V$. In fact,
		by the Helffer--Sj\"ostrand formula, for $\tilde{f}\in\test{\C}$ an almost extension of $f$,
		\begin{align*}
			[f(\Hloc),\varphi P_j^r] 
			&= \frac 1\pi\int_\C \bar{\partial}\tilde{f}(z)[(\Hloc-z)^{-1},\varphi P_j^r] L(\di z)\;.
		\end{align*}
		Let $z\in\C\setminus\spec(\Hloc)$. Writing the resolvent identities gives that
		\begin{align*}
			[(\Hloc&-z)^{-1},\varphi P_j^r] 
			\\&= [(\Hloc-z)^{-1},\varphi]P_j^r+\varphi[(\Hloc-z)^{-1},P_j^r]
			\\&=(\Hloc-z)^{-1} [\Hloc-z,\varphi](\Hloc-z)^{-1}P_j^r
			-\varphi(\Hloc-z)^{-1}[\Hloc-z,P_j^r](\Hloc-z)^{-1}	
            \\&=(\Hloc-z)^{-1} [\Hloc,\varphi](\Hloc-z)^{-1}P_j^r
			-\varphi(\Hloc-z)^{-1}[\Hloc,P_j^r](\Hloc-z)^{-1}	
            \;.
		\end{align*}
		On the one hand, by \cref{eq:com-H-phi}
		\begin{align*}
			 [\Hloc,\varphi]
			 	= \hbar\sum_{k=1}^d\left(-\hbar\partial_{x_k}^2\varphi+2i\partial_{x_k}\varphi P_k\right) \;.
		\end{align*}
		On the other hand, by the explicit identity on the magnetic momenta \cref{eq:com-momenta-square}
		\begin{align}\label{eq-proof:HP}
			 [\Hloc,P_j^r]
		&			= \delta_{r,1}\left(\sum_{k=1}^d[P_k^2,P_j]+[V_{\loc},P_j]\right)\nonumber
			\\&=i\hbar\delta_{r,1}\left(b\sum_{k=1}^d \big(P_k B_{kj}+B_{kj}P_k\big)+\partial_{x_j}V_{\loc}\right).
		\end{align}
		Since $\cD(\Hhb)\subset\cD(P_j)$, since the functions $a_{\loc}$ and $V_{\loc}$ are smooth and compactly supported, each operator $(\Hloc-z)^{-1}P_j^r$, $\varphi(\Hloc-z)^{-1}$, $ (\Hloc-z)^{-1}\big(-i\hbar\partial_{x_k}^2\varphi+2i\partial_{x_k}\varphi P_k\big)$ and $ (\Hloc-z)^{-1}\Big(b\big(P_k B_{kj}+B_{kj}P_k\big)+\frac 1d\partial_{x_j}V\Big)$ is bounded in operator norm for any $k\in\{1,\ldots,d\}$.
		
        Moreover, by \cref{lemma:aux-loc_00} we have
		\begin{equation*}
			\norm{[(\Hloc-z)^{-1},\varphi P_j^r]}_\op\leq \frac{C\crochetjap{b}^r\hbar\crochetjap{z}^2}{\abs{\Im(z)}^2}.
		\end{equation*}
		By the definition of the almost analytic extension we deduce that $[f(\Hloc),\varphi P_j^r]$ is a bounded operator and its operator norm is $\cO(\hbar)$, completing the verification of \cref{eq-proof:aux-loc_3-bis}.
\end{proof}

We next state the magnetic version of the Hilbert--Schmidt bounds \cite[Lemma 3.2]{Fournais-Mikkelsen-2020}.

\begin{lemma}\label{lemma:aux-loc_4}
	Assume that $(\Hhb,B(0,4R))$ satisfy \cref{ass:loc-magn} for $R>0$. Then, for any $f\in\test{\R}$, $\varphi\in\test{B(0,3R)}$, $\hbar_0>0$, there exists $C>0$ such that for any $b>0$ and any $\hbar\in(0,\hbar_0]$,
	\begin{equation*}
		\norm{[\Hhb,\varphi]f(\Hhb)}_2\leq C\hbar^{1- \frac d2}.
	\end{equation*}
	This constant depends on the dimension $d$, of $R$, of the norms $\normLp{(V_{\loc})_-}{\infty}{(\R^d)}$, of $\normLp{\partial^\alpha\varphi}{\infty}{(\R^d)}$ for finitely many $\alpha\in\N^d$ and of $\normLp{\partial^kf}{\infty}{(\R)}$ for all $k\in\N$.
\end{lemma}

The proof is exactly the same as the one of \cite[Lemma 3.2]{Fournais-Mikkelsen-2020}.

\subsubsection{Magnetic Weyl laws}

We begin with the Weyl law asymptotics established in \cite[Theorem 4.1]{Sobolev-1995} for moderate magnetic fields, that is, under the scaling $b=\cO_\hbar(1)$.

\begin{theorem}[Optimal local Weyl laws]\label{thm:WL-Sobolev-non-critical}
	Let $\hbar_0>0$ and $b_0<1$. Assume that $(\Hhb,B(0,4R))$ satisfies \cref{ass:loc-magn} for $R>0$, for the local operator $\Hloc =(-i\hbar\nabla-ba_\loc(x))^2+V_\loc(x)$ with $b\in(0,b_0]$, and that for any $x\in B(0,2R)$
	\begin{equation}\label{eq:non-critical-cond_bis}
		\abs{V_\loc(x)}+\abs{\nabla V_\loc(x)}^2+\hbar\geq c.
	\end{equation}
	Then, there exists $C>0$ such that for any $\hbar\in(0,\hbar_0]$ and for any $\varphi\in\test{B(0,R/2)}$
	\begin{equation*}
		\abs{\tr\left(\varphi\indic_{(-\infty,0]}(\Hhb)\right)-\frac{1}{(2\pi\hbar)^d}\left(\int_{\R^d}\int_{\R^d} \indic_{(-\infty,0]}(\abs{\xi}^2+V_{\loc}(x))\varphi(x)\di x \di\xi\right)}\leq C\hbar^{1-d}.
	\end{equation*}
	The constant $C$ depends on $d,R$, on the supremum norms of finitely many derivatives of $V$ and the supremum norms of the derivatives of each coordinate of $a$ for order larger of equal than 1.
\end{theorem}

Actually, for the proof of the local commutators bounds (see \cref{thm:cb-loc-non-critical-cond}) we use a similar theorem on small intervals. We do not detail the proof as it is exactly the same as \cite[Proposition 2.3]{Fournais-Mikkelsen-2020}. It relies on the \cref{thm:WL-Sobolev-non-critical} and the Coarea formula.

\begin{prop}[Local Weyl law in a finite interval]\label{thm:WL-finite-int}
	Assume that $(\Hhb,B(0,4R))$ satisfy \cref{ass:loc-magn} for $R>0$. If there exists $\varepsilon>0$ and $c>0$ such that for any $E\in[-2\varepsilon,2\varepsilon]$
	\begin{equation}\label{eq:non-critical-cond_WL}
		\abs{V_\loc(x)-E}+\abs{\nabla V_\loc(x)}^2+\hbar\geq c\;,
	\end{equation}
	then, there exists $C_1,C_2>0$ such that for any closed interval $I\subset(-\varepsilon,\varepsilon)$ and for any $\varphi\in\test{B(0,R/2)}$
	\begin{equation*}
		\abs{\tr\left(\varphi\indic_I(\Hhb)\right)}\leq C_1\abs{I}\hbar^{-d}+C_2\hbar^{1-d}\;.
	\end{equation*}
\end{prop}

\subsubsection{Local estimates with non-criticality condition}\label{subsect:local-estimates-non-criticcal}

We state and prove a local version of the estimates in a neighborhood that satisfies a non-criticality condition with respect to $V$. 

\begin{theorem}\label{thm:cb-loc-non-critical-cond}
	Assume that there exists $R>0$ such that $(\Hhb,B(0,4R))$ satisfy Assumption \ref{ass:loc-magn} for $a_\loc\in\test{\R^d,\R^d}$ and $V_\loc\in\test{\R^d}$, and that there exists $c>0$ such that for any $x\in B(0,2R)$
	\begin{equation}\label{eq:non-critical-cond}
		\abs{V_\loc(x)}+\abs{\nabla V_\loc(x)}^2+\hbar\geq c.
	\end{equation}
	Let $h>0$, $b>0$ and let $P_j:=-i\hbar\partial_{x_j}-b(a_\loc)_j(x)$ for $j\in\{1,\ldots,d\}$. Then, for any $\varphi\in\test{B(0,R/2)}$, there exists $C>0$, any $h\in(0,h_0]$, any $b\in(0,1)$ and $r\in\{0,1\}$, one has
	\begin{equation*}
		\norm{[\indic_{(-\infty,0]}(\Hhb),\varphi P_j^r]}_1 \leq C \hbar^{1-d}.
	\end{equation*}
	Actually, $C$ depends on $d$, $R$, $\normLp{\partial_x^\alpha V}{\infty}{}$,  $\normLp{\partial_x^\alpha \varphi}{\infty}{}$ for finitely many $\alpha\in\N^d$, and on $\normLp{\partial_x^\beta (a_\loc)_j}{\infty}{}$ for any $\beta\in\N^d$ such that $\abs{\beta}\geq 1$ and for any $j\in\{1,\ldots,d\}$.
\end{theorem}

The proof follows the same structure as the one in \cite[Theorem 3.3]{Fournais-Mikkelsen-2020} (we write it in \cref{sec:proof-cb-loc-non-critical-cond} for the paper self-containment) and relies on dyadic decomposition and on the auxiliary results (they are the only ones that differ): \cref{lemma:aux-loc_1}, \cref{lemma:aux-loc_3}, \cref{lemma:aux-loc_4}, and \cref{thm:WL-finite-int}. The only bound to be careful to prove is the following 
\begin{equation}\label{eq-proof:non-critical-cond_3}
	\norm{(\Hhb-i)^{-1}\left[[\varphi P_j^r,\Hloc],\Hhb\right](\Hhb-i)^{-1}}_{\op}\leq C\hbar^2\;.
\end{equation}
Actually, in the magnetic case the term $[P_k^2,P_j]$ is not necessarily equal to 0 (see \cref{lemma:prop-magn-momentum}).

\begin{proof}[Proof of the bound \cref{eq-proof:non-critical-cond_3}]
The idea is to express the double commutator $[[ P_j^r,\Hloc],\Hhb] = [[ P_j^r,\Hloc],\Hloc]$ as a pseudo-differential operator of order $2+r$. More precisely it is a linear combination of the differential operators $P^\alpha$ for $\alpha\in\N^d$ such that $\abs{\alpha}\leq 2+r$ with coefficients that are partial derivatives of $\varphi$
\begin{align*}
	\left[[ \varphi P_j^r,\Hloc],\Hhb\right] 
	&= \hbar^2\Big(\sum_{\abs{\alpha}\leq 2+r} \big(\sum_{\abs{\beta_\alpha}\leq 2+r} c_{\beta_\alpha}\partial_x^{\beta_\alpha}\varphi(x)\big) P^\alpha +\sum_{k=1}^dc_{kj}(\partial_{x_j}^r\partial_{x_k}\varphi)\partial_{x_k}V\Big).
\end{align*}

Then, since $\dom(\Hhb)$ is included into $\dom(P^\alpha)$ for any $\alpha\in\N^d$ such that $\abs{\alpha}\leq 2$, the operators $P^\alpha (\Hhb-i)^{-1}$ and $(\Hhb-i)^{-1}P^\alpha$ are bounded on $L^2(\R^d)$, this gives \cref{eq-proof:non-critical-cond_3}.
%
\item[\quad (i)] 
We recall that
\begin{align}\label{eq-proof:[H,phi]}
	[\Hloc,\varphi]
	= -i\hbar\sum_{k=1}^d\Big(2(\partial_{x_k}\varphi) P_k+i\hbar\partial_{x_k x_k}^2\varphi\Big).
\end{align}
In particular
\begin{align}\label{eq-proof:[[H,phi],V]}
	[[\Hloc,\varphi],V]
	&= 2i\hbar\sum_{k=1}^d ([\partial_{x_k}\varphi,P_k]V+P_k\underset{=0}{\underbrace{[\partial_{x_k}\varphi,V]}}) \nonumber
	\\&= -2\hbar^2\sum_{k=1}^d (\partial_{x_k}\varphi)(\partial_{x_k}V) = -2\hbar^2\nabla_x\varphi\cdot\nabla_x V.
\end{align}
Moreover,  for any $\ell\in\{1,\ldots,d\}$ 
\begin{align}\label{eq-proof:HphiP}
	[[\Hloc,\varphi],P_\ell]
	&= -i\hbar\sum_{k=1}^d \Big( 2[ (\partial_{x_k}\varphi) P_k,P_\ell]+i\hbar[\partial_{x_k x_k}^2\varphi,P_\ell ] \Big)
	\nonumber
	\\&= i\hbar\sum_{k=1}^d\Big( 2\big( [\partial_{x_k}\varphi,P_k]P_\ell +P_k[\partial_{x_k}\varphi,P_\ell] \big) -i\hbar[\partial_{x_k x_k}^2\varphi,P_\ell ]\Big) 
	\nonumber
	\\&= -\hbar^2 \sum_{k=1}^d\Big(2\big(\partial_{x_k x_k}^2\varphi P_\ell+ P_k\partial_{x_\ell x_k}^2\varphi\big)-i\hbar\partial_{x_\ell x_k x_k}^3\varphi\Big) \;,
\end{align}
and then using the identity
\begin{equation}\label{eq-proof:sandwich-P}
	P_\ell \psi(x) P_k = \psi(x)P_\ell P_k- [\psi(x),P_\ell]P_k = \psi(x)P_\ell P_k-i\hbar\partial_{x_\ell}\psi(x) P_k\;,
\end{equation}
we have
\begin{align*}
	[[\Hloc,\varphi],P_\ell^2] 
	&= [[\Hloc,\varphi],P_\ell]P_\ell+P_\ell[[\Hloc,\varphi],P_\ell] 
	\\&= -\hbar^2 \sum_{k=1}^d \Big( 2\big((\partial_{x_k x_k}^2\varphi)P_\ell^2+(\partial_{x_\ell x_k}^2) P_k P_\ell+ P_\ell(\partial_{x_k x_k}^2\varphi)P_\ell+ P_k(\partial_{x_\ell x_k}^2\varphi)P_\ell\big)
	\\&\qquad\qquad\qquad -i\hbar(\partial_{x_\ell x_k x_k}^3\varphi) P_\ell-i\hbar P_\ell (\partial_{x_\ell x_k x_k}^3\varphi)
	\Big)
	\\&= -\hbar^2 \sum_{k=1}^d \Big(4\big((\partial_{x_k x_k}^2\varphi)P_\ell^2+(\partial_{x_\ell x_k}^2\varphi) P_k P_\ell\big) -5i\hbar(\partial_{x_\ell x_k x_k}^3\varphi) P_\ell
	+\hbar^2\partial_{x_\ell x_\ell x_k x_k}^4\varphi	\Big)\;.
\end{align*}
Thus, we get the expression\footnote{Basically, it should be the same expression when $r=0$ as in \cite[(3.22)]{Fournais-Mikkelsen-2020} but with the magnetic momentum
	\begin{align*}
		\left[[\varphi,\Hloc],\Hloc\right] 
		= \hbar^2\sum_{k=1}^d\sum_{\ell=1}^d\left(-2\Big(P_\ell(\partial_{x_\ell x_\ell}^2\varphi)P_k+P_k(\partial_{x_k x_\ell}^2\varphi)P_\ell\Big) +2(\partial_{x_k}\varphi)(\partial_{x_k}V)+\hbar^2\partial_{x_\ell x_\ell x_k x_k}^4\varphi\right)\;.
\end{align*}}
 \begin{equation}\label{eq-proof:HphiH}
 \begin{split}
 	\Big[[\Hloc&,\varphi],\Hloc\Big] 
 	%
 	%
 	\\&=-\hbar^2\Big(\sum_{k=1}^d\sum_{\ell=1}^d \Big(4\big(\partial_{x_k x_k}^2\varphi)P_\ell^2+(\partial_{x_\ell x_k}^2\varphi) P_k P_\ell\big) -5i\hbar(\partial_{x_\ell x_k x_k}^3\varphi) P_\ell\Big)
 	\\&\qquad\quad +\hbar^2\sum_{k=1}^d\sum_{\ell=1}^d\partial_{x_\ell x_\ell x_k x_k}^4\varphi
 	+2\nabla_x\varphi\cdot\nabla_x V\Big)\;.
 \end{split}
 \end{equation}
 We conclude by using that $\dom(\Hhb)\subset \dom(P_k)$ and $\dom(\Hhb)\subset \dom(P_\ell P_k)$ for any $k,l\in\{1,\ldots,d\}$ and then each $P_k (\Hhb-i)^{-1}$ and $P_j P_k (\Hhb-i)^{-1}$ are bounded on $L^2(\R^d)$, this gives \cref{eq-proof:non-critical-cond_3}.
 
\item[\quad (ii)]
By \cref{eq-proof:HphiP} and \cref{eq-proof:HphiH} we obtain
\begin{align*}
	[[\Hloc&,\varphi P_j],\Hloc]
	\\&= [[\Hloc,\varphi] P_j,\Hloc]+[\varphi[\Hloc,P_j],\Hloc]
	\\&= -[\Hloc,[\Hloc,\varphi]]P_j-2[\Hloc,\varphi][\Hloc,P_j] -\varphi[\Hloc,[\Hloc,P_j]]
	\\&= [[\Hloc,\varphi],\Hloc]P_j-2[\Hloc,\varphi][\Hloc,P_j] +\varphi[[\Hloc,P_j],\Hloc]\;.
\end{align*}
Moreover, by redoing the same calculations as in \cref{eq-proof:HP},
\begin{align*}
	[\Hloc,P_j]
	= i\hbar\sum_{k=1}^d \Big(2 bB_{kj} P_k-i\hbar( b\partial_{x_k}B_{kj}+\partial_{x_k} V)\Big)\;.
\end{align*}
This implies
\begin{align*}
	[[\Hloc,P_j],V]
	&= 2i\hbar b\sum_{k=1}^d [B_{kj} P_k,V] = -2\hbar^2 b\sum_{k=1}^d \partial_{x_k}B_{kj}\partial_{x_k}V
\end{align*}
and for any $\ell\in\{1,\ldots,d\}$
\begin{align*}
	[[\Hloc,P_j],P_\ell]
	&= i\hbar\sum_{k=1}^d \Big( 2 b[B_{kj} P_k,P_\ell]-i\hbar[ b\partial_{x_k}B_{kj}+\partial_{x_k} V,P_\ell]\Big)
	\\&= i\hbar\sum_{k=1}^d \Big( 2 b\big([B_{kj},P_\ell]P_k +B_{kj}[P_k,P_\ell]\big)-i\hbar[ b\partial_{x_k}B_{kj}+\partial_{x_k} V,P_\ell]\Big)
	\\&= -\hbar^2\sum_{k=1}^d \Big( 2 b\big((\partial_{x_\ell}B_{kj})P_k + bB_{kj}B_{k\ell}\big)-i\hbar\big( b\partial_{x_\ell x_k}^2B_{kj}+\partial_{x_\ell x_k}^2 V\big)\Big)\;.
\end{align*}
Thus, by using \cref{eq-proof:sandwich-P} we get
\begin{align*}
	[[\Hloc&,P_j],\Hloc] 
	\\&= \sum_{\ell=1}^d ([[\Hloc,P_j],P_\ell]P_\ell+P_\ell[[\Hloc,P_j],P_\ell])+[[\Hloc,P_j],V]
	\\&=-\hbar^2 \sum_{k=1}^d \Big(\sum_{\ell=1}^d \Big( 4b\big((\partial_{x_\ell}B_{kj})P_k + bB_{kj}B_{k\ell}\big)
	-2i\hbar b\big(\partial_{x_\ell x_\ell}^2B_{kj}+b\partial_{x_\ell}(B_{kj}B_{k\ell})\big)
	\Big)
	\\&\qquad\qquad\quad+ 2b\partial_{x_k}B_{kj}\partial_{x_k}V \Big)\;	.
\end{align*}
Finally, by previous the explicit formulas that we get
\begin{align*}
	\norm{[[\Hloc,\varphi],\Hloc](\Hhb-i)^{-1}}_\op +\norm{[[\Hloc,P_j],\Hloc](\Hhb-i)^{-1}}_\op & \leq C\hbar^2,
	\\ \norm{P_j(\Hhb-i)^{-1}}_\op +\norm{[\Hloc,\varphi](\Hhb-i)^{-1}}_\op +\norm{[\Hloc,P_j](\Hhb-i)^{-1}}_\op & \leq C\hbar\;.
\end{align*}
so that we obtain \cref{eq-proof:non-critical-cond_3} by the triangle inequality and the H\"older inequality.
\end{proof}


We now remove the non-criticality condition on the potential by applying the multiscale technique of \cite{Sobolev-1995}. The proof follows the same structure as in \cite[Theorem 3.6]{Fournais-Mikkelsen-2020}. 

The following lemma that can be found in \cite[Theorem 1.4.10]{Hormander-I}.

\begin{lemma}\label{lemma:multiscaling}
	Let $\Omega\subset\R^d$ an open set and let $f\in\cC^1(\bar{\Omega})$ be positive and such that there exists $\rho\in(0,1)$ such that $\abs{\nabla f(x)}\leq\rho$ for all $x\in\Omega$.
	Then:
	\begin{itemize}
		\item[(i)] There exists a sequence $(x_k)_{k\in\N}\subset\Omega$ such that the open balls $B(x_k,f(x_k))$ cover $\Omega$. Furthermore, there exists $N_\rho\in\N$ such that any intersection of more than $N_\rho$ balls of this covering is empty.
		\item[(ii)] One can choose a sequence of smooth functions $(\phi_k)_{k\in\N}\subset\test{\Omega}$ such that $\supp(\phi_k)\subset B(x_k,f(x_k))$ for any $k\in\N$, such that for any $\alpha\in\N^d$, there exists $C_\alpha>0$ such that for all $x \in \Omega$ we have
		\begin{align*}
			\abs{\partial_x^\alpha\phi_k(x)}\leq C_\alpha f(x) \quad
		\text{ and } \quad
        \sum_{k=0}^\infty\phi_k(x)=1\;.
		\end{align*}
	\end{itemize}
\end{lemma}

\begin{proof}[{Proof of Theorem \ref{thm:cb-loc-gene}}]

\item{\bf Step 1.}
We first assume that $\hbar\in(0,\hbar_0]$ and $b\in(0,b_0]$ with $b_0<1$.
Let $\varepsilon>0$ be such that
\begin{equation*}
	\dist(\supp(\varphi),\partial\Omega)\geq\varepsilon.
\end{equation*}
Let us consider the smooth function
\begin{equation*}
	f:x\in\R^d\mapsto M^{-1}\Big(V(x)^2+\abs{\nabla V(x)}^4+\hbar^2\Big)^{\frac 14}\;,
\end{equation*}
where $M>0$ is a constant chosen large enough\footnote{We have taken the same constants as in \cite{Fournais-Mikkelsen-2020}, but we can take any numbers $\tilde{N}>N>1$ such that 
$f(x)\leq \min(1,\varepsilon   {\tilde{N}}^{-1})$ and $\abs{\nabla f(x)}\leq\rho<N^{-1}$.}
so that for any $\hbar\in(0,\hbar_0]$ and any $x\in\bar{\Omega}$
\begin{equation}\label{eq-proof:multiscaling_cond-M}
	f(x)\leq \min\Big(1,\frac\varepsilon   {9}\Big)\quad\text{ and } \abs{\nabla f(x)}\leq\rho<\frac 18\;.
\end{equation}
By construction, it follows that for any $x\in\Omega$, the ball $B(x,8f(x))$ is included in $\Omega$.
By the mean value theorem, for any $y\in\Omega$, for any $x\in B(y,8f(y))$
\begin{equation}\label{eq-proof:multiscaling_growth-f-ball}
	(1-8\rho)f(y)\leq f(x)\leq (1+8\rho)f(y)\;.
\end{equation}
This bound will be useful later in the proof to verify the non-criticality condition for rescaled operators.

By definition of $f$, one has the following estimates
\begin{equation}\label{eq-proof:multiscaling_growth-V}
	\forall x\in\R^d,\quad
	\abs{V(x)}\leq M^2f(x)^2,\quad \abs{\nabla V(x)}\leq Mf(x)\;,
\end{equation}
that we use later to control the derivative of rescaled operators.

Application of \cref{lemma:multiscaling} to $f$ and to the set $\Omega$ provides a sequence $(x_k)_{k\in\N}\subset\Omega$ such that $\bigcup_{k\in\N}B(x_k,f(x_k))\supset\Omega$, an integer $N_{\frac 18}$ such that any intersection of this covering $(B(x_k,f(x_k)))_{k\in\N}$ of more that $N_{\frac 18}$ elements is empty, and a sequence of partition of unity $(\phi_k)_{k\in\N}\subset\test{\Omega}$ such that $\supp(\phi_k)\subset B(x_k,f(x_k))$ for any $k\in\N$, such that for any $\alpha\in\N^d$, there exists $C_\alpha>0$ such that
\begin{equation}\label{eq-proof:multiscaling_growth-phik}
	\forall x\in\Omega,\quad
	\abs{\partial_x^\alpha\phi_k(x)}\leq C_\alpha f(x)^{-\abs{\alpha}}\;,
\end{equation}
Moreover, since $\supp(\varphi)\subset\Omega$ and is bounded, it admits a finite covering $(B(x_k,f(x_k)))_{k\in I}$ for a finite set of index $I\subset\N$. Then, by the triangle inequality, for any $j\in\{1,\ldots,d\}$ and any $r\in\{0,1\}$
\begin{align*}
	\norm{[\indic_{(-\infty,0]}(H_{\hbar,b}),\varphi P_j^r]}_1 \leq \sum_{k\in I}\norm{[\indic_{(-\infty,0]}(H_{\hbar,b}),\phi_k\varphi P_j^r]}_1\;.
\end{align*}
Let us rewrite each term $\norm{[\indic_{(-\infty,0]}(H_{\hbar,b}),\phi_k\varphi P_j^r]}_1$ by rescaling. We define the unitary transformation on $L^2(\R^d)$ defined by
\begin{align*}
	U_k v(x):=(f(x_k))^{d/2}v(f(x_k)x+x_k),
\end{align*}
and the conjugated operator 
\begin{align*}
	\widetilde{H}_{\hbar,b}:= f(x_k)^{-2}U_k H_{\hbar,b}U_k^* \;, 
\end{align*}
that defines a self-adjoint operator on $L^2(\R^d)$ which satisfies \cref{ass:loc-magn} in $B(0,8)$ for the parameters
\begin{equation*}
	\hbar_k:=\frac{\hbar}{f(x_k)^2},\quad b_k=b
\end{equation*}
and for the local functions defined by
\begin{equation*}\forall 1\leq j\leq d,\quad
	(\widetilde{a}_k)_j(x):=f(x_k)^{-1}a(f(x_k)x+x_k),\quad \widetilde{V}_k(x):=f(x_k)^{-2}V(f(x_k)x+x_k).
\end{equation*}
We notice that we have the uniform bound
\begin{align*}
	\hbar_k\leq M^2,\qquad b_k\leq b_0<1\;.
\end{align*}
%
Furthermore, since $U_k$ is unitary,
\begin{align*}
	&\norm{[\indic_{(-\infty,0]}(H_{\hbar,b}),\phi_k\varphi P_j^r]}_1 
	= \norm{U_k[\indic_{(-\infty,0]}(H_{\hbar,b}),\phi_k\varphi P_j^r]U_k^*}_1
	\\&\qquad=f(x_k)^{-2-r} \norm{[f(x_k)^2U_k\indic_{(-\infty,0]}(H_{\hbar,b})U_k^*,( U_k\phi_k\varphi U_k^* )\: f(x_k)(U_k P_j^r U_k^*)]}_1 
	\\
	&\qquad= f(x_k)^{-2-r}\norm{[\indic_{(-\infty,0]}(\tilde{H}_{\hbar,b}),\widetilde{\phi_k\varphi}(-i\hbar_k\partial_{x_j}-b_k \widetilde{a}_j(x))^r]}_1 	
\end{align*}
where $\widetilde{\phi_k\varphi}\in\test{B(0,1)}$ and
\begin{equation*}
	\widetilde{\phi_k\varphi}(x):=(\phi_k\varphi)(f(x_k)x+x_k).
\end{equation*}
By the growth of the derivatives of $\phi_k$ \cref{eq-proof:multiscaling_growth-phik} and by the first bound of \cref{eq-proof:multiscaling_cond-M}, one has the uniform bound for any $\alpha\in\N^d$
\begin{align}\label{eq-proof:multiscaling_growth-phi-ball}
	\normLp{\partial_x^\alpha(\widetilde{\phi_k\varphi})}{\infty}{(\R^d)}&= \sup_{x\in B(0,1)}f(x_k)^{\abs{\alpha}}\abs{\sum_{\beta\leq\alpha} \binom{\alpha}{\beta}(\partial_x\phi_k^{\beta}\varphi)(f(x_k)x+x_k)(\partial_x^{\alpha-\beta}\varphi)(f(x_k)x+x_k)} \nonumber
	\\&
	\leq \sum_{\beta\leq\alpha} \binom{\alpha}{\beta} f(x_k)^{\abs{\alpha-\beta}}\sup_{x\in B(0,1)}\abs{(\partial_x^{\alpha-\beta}\varphi)(f(x_k)x+x_k)}\nonumber
	\\&
	\leq  \varepsilon^{\abs{\alpha}}\sum_{\beta\leq\alpha} \binom{\alpha}{\beta}\normLp{\partial_x^{\beta}\varphi}{\infty}{(\R^d)}\leq C_\alpha\;.
\end{align}
We need to check that the rescaled potential $\widetilde{V}_k$ satisfies the non-criticality condition \cref{eq:non-critical-cond} of \cref{thm:cb-loc-non-critical-cond} and that the derivatives of $\widetilde{V}_k$ and $(\widetilde{a}_k)_j$ are uniformly bounded with respect to $\hbar$ and $b$.

We notice that the point $f(x_k)x+x_k\in B(x_k,8f(x_k))$ for any $x\in B(0,8)$.
By the estimates \cref{eq-proof:multiscaling_growth-V} on the growth of $\partial_x^\alpha\phi_k$ for $\abs{\alpha}\leq1$  and the upper bound of \cref{eq-proof:multiscaling_growth-f-ball}
\begin{align*}
	\normLp{\widetilde{V}_k}{\infty}{(B(0,8))} &\leq f(x_k)^{-2}\sup_{x\in B(0,8)}\abs{V(f(x_k)x+x_k)}
	\leq \sup_{x\in B(0,8)}f(x_k)^{-2} M^2f(f(x_k)x+x_k)^2
	\\&\leq (1+8\rho)^2 M^2,
	\\	\normLp{\nabla\widetilde{V}_k}{\infty}{(B(0,8))} &\leq f(x_k)^{-1}\sup_{x\in B(0,8)} \abs{\nabla V(f(x_k)x+x_k)}
	\leq (1+8\rho) M\;.
\end{align*}
Moreover, by \cref{eq-proof:multiscaling_cond-M} and the fact that $V$ are smooth with a compact support,  for any $\alpha\in\N^d$ such that $\abs{\alpha}\geq 2$
\begin{align*}
	\normLp{\partial^\alpha_x\widetilde{V}_k}{\infty}{(B(0,8))} &\leq f(x_k)^{\abs{\alpha}-2}\normLp{\partial^\alpha_xV}{\infty}{(\R^d))} \leq \normLp{\partial^\alpha_xV}{\infty}{(\R^d))}.
\end{align*}
%
As well, for any $\alpha\in\N^d$ such that $\abs{\alpha}\geq 1$
\begin{align*}
	\normLp{\partial^\alpha_x(\widetilde{a}_k)_j}{\infty}{(B(0,8))} &\leq f(x_k)^{\abs{\alpha}-1}\normLp{\partial^\alpha_xa_j}{\infty}{(\R^d)} \leq \normLp{\partial^\alpha_xa_j}{\infty}{(\R^d)}\; ,
\end{align*}
%
Then, by the lower bound of \cref{eq-proof:multiscaling_growth-f-ball}, one gets the lower bound uniform in $B(0,8)$
\begin{align*}
	\widetilde{V}_k(x)&+\abs{\nabla\widetilde{V}_k(x)}^2+\hbar_k = f(x_k)^{-2}\Big(V(f(x_k)x+x_k)+f(x_k)^2\abs{\nabla V(f(x_k)x+x_k)}^2+\hbar\Big)
	\\&\geq f(x_k)^{-2} M^2 \Big[M^{-2}\Big(V(f(x_k)x+x_k)^2+f(x_k)^2\abs{\nabla V(f(x_k)x+x_k)}^4+\hbar^2\Big)^{\frac 12}\Big]
	\\&\quad= M^2f(x_k)^{-2}f(f(x_k)x+x_k)^2
	\\&\geq M^2(1-8\rho)^2>0\;.
\end{align*}
Finally, we apply \cref{thm:cb-loc-non-critical-cond} to the operator $\widetilde{H}_{\hbar,b}$ on the open ball $B(0,8)$ (we consider the parameter $R=2$), there exists a constant $C>0$ that only depends on $\hbar_0$, $b_0$ and the derivatives of $V,a$ and $\varphi$, such that
\begin{align*}
	\norm{[\indic_{(-\infty,0]}(H_{\hbar,b}),\phi_k\varphi P_j^r]}_1 	&= f(x_k)^{-2-r}\norm{[\indic_{(-\infty,0)}(\tilde{H}_{\hbar,b}),\widetilde{\phi_k\varphi}(-i\hbar_k\partial_{x_j}-b_k \widetilde{a}_j(x))^r]}_1 
	\\&\leq C f(x_k)^{-2-r}\Big(\frac{\hbar}{f(x_k)^2}\Big)^{1-d} \leq C\Big(\int_{B(x_k,f(x_k))}f(x)^{d-r} \di x\Big) \hbar^{1-d}\;.
\end{align*}
Hence, by summing the above estimates on the indexes $k\in I$ we conclude  that
\begin{align*}
	\norm{[\indic_{(-\infty,0]}(H_{\hbar,b}),\varphi P_j^r]}_1 &\leq CN_{\frac 18}\abs{\Omega}\hbar^{1-d}\;.
\end{align*}
This ends the proof of \cref{thm:cb-loc-gene} for $\hbar\in(0,\hbar_0]$ for any $\hbar_0\in(0,1)$.

\item{\bf Step 2.} Let us assume now that $b\in [b_0,c\hbar^{-1})$ for a given $c>0$. We redo the same proof as above. 
We define the positive function
\begin{equation*}
	f(x)=\min\Big(1,\frac \varepsilon 9\Big)\frac{b_0}b\;.
\end{equation*}
We apply \cref{lemma:multiscaling} to $f$, this provides the sequence of points $(x_k)_{k\in\N}\subset\Omega$, the partition of unity $(\phi_k)_{k\in\N}$ and the critical number $N_{\frac1 8}\in\N$ associated to the covering $(B(x_k,f(x_k))_{k\in\N}$ of $\Omega$.
Let us define the unitary transformation on $L^2(\R^d)$ 
\begin{align*}
	U_k v(x):=(f(x_k))^{d/2}v(f(x_k)x+x_k)\;,
\end{align*}
and the a self-adjoint operator on $L^2(\R^d)$,
\begin{align*}
	\widetilde{H}_{\hbar,b}:= U_k H_{\hbar,b}U_k^*\;, 
\end{align*}
for all $k\in\N$. The operator $\widetilde{H}_{\hbar,b}$ satisfies \cref{ass:loc-magn} in $B(0,8)$
for the parameters
\begin{equation*}
	\hbar_k:=\frac{\hbar}{f(x_k)},\quad b_k=bf(x_k)\;,
\end{equation*}
and for the local functions defined by
\begin{equation*}\forall 1\leq j\leq d,\quad
	(\widetilde{a}_k)_j(x):=f(x_k)^{-1}a(f(x_k)x+x_k),\quad \widetilde{V}_k(x):=V(f(x_k)x+x_k)\;.
\end{equation*}
They satisfy the uniform bounds
\begin{align*}
	\hbar_k\leq \hbar b\leq c,\qquad b_k=\min\Big(1,\frac \varepsilon 9\Big)b_0\leq b_0\;.
\end{align*}
As before, let define
\begin{equation*}
	\widetilde{\phi_k\varphi}(x):=(\phi_k\varphi)(f(x_k)x+x_k)\;,
\end{equation*}
so that
\begin{align*}
	\norm{[\indic_{(-\infty,0)}(H_{\hbar,b}),\phi_k\varphi P_j^r]}_1 
	&=f(x_k)^{-r}\norm{[\indic_{(-\infty,0)}(\tilde{H}_{\hbar,b}),\widetilde{\phi_k\varphi}(-i\hbar_k\partial_{x_j}-b_k \widetilde{a}_j(x))^r]}_1\;.
\end{align*}
Let us check that all the supremum norms of the derivatives of $\widetilde{V}_k$, $\widetilde{a}_k$ and $\widetilde{\phi_k\varphi}$ are bounded. 
By definition $f(x)\leq 1$ for any $x\in\R^d$. Then, for any $\alpha\in\N^d$
\begin{align*}
	\normLp{\partial_x^\alpha(\widetilde{\phi_k\varphi})}{\infty}{(\R^d)}
	&\leq \sum_{\beta\leq\alpha} \binom{\alpha}{\beta} f(x_k)^{\abs{\alpha-\beta}}\sup_{x\in B(0,1)}\abs{(\partial_x^{\alpha-\beta}\varphi)(f(x_k)x+x_k)}\nonumber
	\\&
	\leq  \sum_{\beta\leq\alpha} \binom{\alpha}{\beta}\normLp{\partial_x^{\beta}\varphi}{\infty}{(\R^d)}\;,
\end{align*}
and
\begin{align*}
	\normLp{\partial^\alpha_x\widetilde{V}_k}{\infty}{(B(0,8))} &\leq\normLp{\partial^\alpha_xV}{\infty}{(\R^d))} \;.
\end{align*}
Moreover, for any $\alpha\in\N^d$ such that $\abs{\alpha}\geq 1$
\begin{align*}
	\normLp{\partial^\alpha_x(\widetilde{a}_k)_j}{\infty}{(B(0,8))} &\leq f(x_k)^{\abs{\alpha}-1}\normLp{\partial^\alpha_xa_j}{\infty}{(\R^d)} \leq \normLp{\partial^\alpha_xa_j}{\infty}{(\R^d)} \;.
\end{align*}
Thus, since $b_k<1$, by applying \cref{thm:cb-loc-non-critical-cond} to the operator $\tilde{H}_{\hbar,b}$
\begin{align*}
	\norm{[\indic_{(-\infty,0]}(H_{\hbar,b}),\phi_k\varphi P_j^r]}_1 &=f(x_k)^{-r}\norm{[\indic_{(-\infty,0)}(\tilde{H}_{\hbar,b}),\widetilde{\phi_k\varphi}(-i\hbar_k\partial_{x_j}-b_k \widetilde{a}_j(x))^r]}_1
	\\&\leq C f(x_k)^{-r}\Big(\frac{\hbar}{f(x_k)}\Big)^{1-d}
	\\&\leq C\Big(\int_{B(x_k,f(x_k))}f(x)^{d-r-1} \di x\Big) \hbar^{1-d}
	\\&\leq C\crochetjap{b}^{1+r}\Big(\int_{B(x_k,f(x_k))}f(x)^d \di x\Big) \hbar^{1-d}\;.
\end{align*}
Finally, by the triangle inequality and the definition of $N_{\frac 18}$, we get 
\begin{align*}
	\norm{[\indic_{(-\infty,0]}(H_{\hbar,b}),\varphi P_j^r]}_1 &\leq CN_{\frac 18}\abs{\Omega}\crochetjap{b}^{1+r}\hbar^{1-d}\;.
\end{align*}
That ends the proof of \cref{thm:cb-loc-gene}.
\end{proof}


\subsection{Auxiliary Lemmata for the Classically Forbidden Region}

We state several auxiliary results that will be used many times in the final proof.
The first lemma is due to Agmon \cite{Agmon-1985} (see also \cite[Theorem 3.1.1]{Helffer-88}).

\begin{lemma}[Agmon formula]\label{lemma:IMS}
	Let $\Hhb=(-i\hbar\nabla-ba)^2+V$ be a magnetic Schrödinger operator acting on $L^2(\R^d)$ where $b\geq 0$, $a\in L^2_\loc(\R^d,\R^d)$ and where $V\in L^1_{\loc}(\R^d)$.
	Let $\psi\in L^2(\R^d)$ be an eigenfunction of $\Hhb$ associated to the eigenvalue $E\in\R$. Then, for any Lipschitz function $\phi:\R^d\to\R$ and any smooth function $\chi:\R^d\to\R$ such that $\chi e^\phi\psi\in\dom(\qhb)$ we have
	\begin{equation*}\label{eq:IMS}
		\qhb(\chi e^{\phi}\psi) = E\normLp{\chi e^{\phi}\psi}{2}{(\R^d)}^2+\hbar^2\normLp{ e^{\phi} \psi(\nabla\chi+\chi\nabla\phi)}{2}{(\R^d)}^2\;.
	\end{equation*}
    Here the quadratic $(\qhb,\dom(\qhb))$ is defined in \cref{sec:basic-prop-magnetic}.
\end{lemma}

\begin{proof}
	Let $\phi\in L^2(\R^d)$ such that $\Hhb\psi=E\psi$. We have 
	\begin{align*}
		\qhb(\chi e^{\phi}\psi) 
		&=\int_{\R^d}\abs{(i\hbar\nabla+ba(x))(\chi e^\phi\psi)(x)}^2 \di x+\int_{\R^d}V(x)\abs{(\chi e^\phi\psi)(x)}^2\di x
		\\&= \int_{\R^d}\abs{(e^\phi\chi)(i\hbar\nabla+ba(x))\psi+i\hbar e^\phi\psi(\nabla\chi+\chi\nabla\phi)}^2 \di x \\&\qquad+\int_{\R^d}V(x)\abs{(\chi e^\phi\psi)(x)}^2\di x
		\\&= \int_{\R^d}(e^\phi\chi)^2\abs{(i\hbar\nabla+ba(x))\psi}^2\di x+\int_{\R^d}V(x)\abs{(\chi e^\phi\psi)(x)}^2\di x
		\\&\qquad+ 2\Re\Big(\prodscal{(e^\phi\chi)(i\hbar\nabla+ba(x))\psi}{i\hbar e^\phi\psi(\nabla\chi+\chi\nabla\phi)}
		\\&\qquad+\hbar^2\int_{\R^d}\abs{e^\phi\psi(\nabla\chi+\chi\nabla\phi)}^2\di x\;.
	\end{align*}
	Let us write the first term
	\begin{align*}
		\int_{\R^d}(e^\phi\chi)^2&\abs{(i\hbar\nabla+ba(x))\psi}^2\di x+\int_{\R^d}V(x)\abs{(\chi e^\phi\psi)(x)}^2\di x
		\\&= \prodscal{(i\hbar\nabla+ba(x))\psi}{\chi^2e^{2\phi}(i\hbar\nabla+ba(x))\psi}+\prodscal{V\psi}{\chi^2e^{2\phi}\psi}
		\\&=\prodscal{(i\hbar\nabla+ba(x))\psi}{(i\hbar\nabla+ba(x))(\chi^2e^{2\phi}\psi)}+\prodscal{V\psi}{\chi^2e^{2\phi}\psi}
		\\&\qquad -\prodscal{(i\hbar\nabla+ba(x))\psi}{2i\hbar \chi e^{2\phi}(\nabla\chi+\chi\nabla\phi)}
		\\&=\prodscal{(i\hbar\nabla+ba(x))\psi}{(i\hbar\nabla+ba(x))(\chi^2e^{2\phi}\psi)}+\prodscal{V\psi}{\chi^2e^{2\phi}\psi}
		\\&\qquad -2\prodscal{(e^\phi\chi)(i\hbar\nabla+ba(x))\psi}{i\hbar  e^{\phi}(\nabla\chi+\chi\nabla\phi)}
		\\&= E\prodscal{\psi}{\chi^2e^{2\phi}\psi}-2\prodscal{(e^\phi\chi)(i\hbar\nabla+ba(x))\psi}{i\hbar  e^{\phi}(\nabla\chi+\chi\nabla\phi)}
		\\&= E\normLp{\chi e^{\phi}\psi}{2}{(\R^d)}^2-2\prodscal{(e^\phi\chi)(i\hbar\nabla+ba(x))\psi}{i\hbar  e^{\phi}(\nabla\chi+\chi\nabla\phi)}
		\;.
	\end{align*}
	We remember that it is real, then
	\begin{align*}
		\int_{\R^d}(e^\phi\chi)^2&\abs{(i\hbar\nabla+ba(x))\psi}^2\di x+\int_{\R^d}V(x)\abs{(\chi e^\phi\psi)(x)}^2\di x
		\\&= E\normLp{\chi e^{\phi}\psi}{2}{(\R^d)}^2-2\Re\prodscal{(e^\phi\chi)(i\hbar\nabla+ba(x))\psi}{i\hbar  e^{\phi}(\nabla\chi+\chi\nabla\phi)}\;.
	\end{align*}
	Hence, we obtain the desired inequality
	\begin{equation*}
		\qhb(\chi e^{\phi}\psi) = E\normLp{\chi e^{\phi}\psi}{2}{(\R^d)}^2+\hbar^2\normLp{  e^{\phi}\psi(\nabla\chi+\chi\nabla\phi)}{2}{(\R^d)}^2\;. \qedhere
	\end{equation*}
\end{proof}

Similarly as in \cite[Lemma A.1]{Fournais-Mikkelsen-2020}, one has an Agmon type estimate for magnetic Schrödinger operators that holds as long as $V$ is bounded from below outside of a bounded set. These estimates rely on $(i\hbar\nabla+a)^2+V\geq V$ and are probably not optimal.

\begin{lemma}[Agmon estimates]\label{lemma:Agmon}
	Let $\mu>\inf_\R V$ and $\Hhb=(-i\hbar\nabla-ba)^2+V$ be a magnetic Schrödinger operator acting on $L^2(\R^d)$ where $b\geq 0$, $a\in  L^2_\loc(\R^d,\R^d)$, and $V\in L^1_{\loc}(\R^d)$. Assume that there exists $\varepsilon>0$ and an open bounded set $U\subset\R^d$ such that
	\begin{equation}\label{eq:ass-Agmon}
		\{ x\in\R^d \::\: V(x)-\mu<\varepsilon \}\subset U.
	\end{equation}
	Let
	\begin{equation}\label{eq:def-Agmon}
		\tilde{U}:=\{x\in\R^d \: :\:\dist(x,U)< 1\} 
		\quad\text{ and }\quad
		\phi_\varepsilon:x\in\R^d\mapsto\varepsilon\,\dist(x,\tilde{U})
		\;.
	\end{equation}
	Then, there exists a constant $C>0$ independent of $a$ and $b$ such that for any normalized eigenfunction $\psi\in L^2(\R^d)$ of $\Hhb$ associated to an eigenvalue in $(-\infty,\varepsilon/4)$ we have
	\begin{equation*}
		\normLp{e^{\phi_\varepsilon/\hbar}\psi}{2}{(\R^d)}\leq C\;.
	\end{equation*}
\end{lemma}

For the sake of completeness, we provide a proof of these estimates.

\begin{proof}[{Proof of \cref{lemma:Agmon}}]
	Let $E<\varepsilon/4$ and $\psi\in L^2(\R^d)$ be a normalized eigenfunction of $\Hhb$ such that $\Hhb\psi=E\psi$. Let $\chi\in\cC^\infty(\R^d,[0,1])$ be such that $\chi=1$ on $\Omega_\varepsilon^c$ and such that $\supp(\chi)\subset U^c$. By construction, $\supp(\abs{\nabla\chi})\subset \tilde{U}\setminus U$.
	For $\gamma\in (0,1]$, we define the function 
	\begin{align*}
		\phi_{\varepsilon,\gamma}:x\in\R^d\mapsto \frac{\phi_\varepsilon(x)}{1+\gamma\abs{x}^2}\;.
	\end{align*}
	By construction, $\phi_{\varepsilon,\gamma}$ is bounded and almost everywhere differentiable in $\R^d$.
	Moreover, for almost all $x\in\R^d$
	\begin{align*}
		\nabla_x\phi_\varepsilon &= \varepsilon\indic_{\Omega_\varepsilon}(x)\nabla_x\dist(x,U)\;,
	\end{align*}
	and
	\begin{align*}
		\nabla_x\phi_{\varepsilon,\gamma} &= \frac{1}{1+\gamma\abs{x}^2}\nabla_x\phi_\varepsilon -\frac{2\gamma\phi_\varepsilon(x)}{(1+\gamma\abs{x}^2)^2} x
		= \frac{1}{1+\gamma\abs{x}^2}\nabla_x\phi_\varepsilon -\phi_{\varepsilon,\gamma}(x)\frac{2\gamma}{1+\gamma\abs{x}^2} x\;.
	\end{align*}
	Up to a translation, let us assume that $0\in U$ so that $\dist(x,U)\leq \abs{x}$ for any $x\in\R^d$. In particular, one has the inclusion $\supp(\abs{\nabla\phi_{\varepsilon,\gamma}})\subset \tilde{U}\setminus U$ and
	\begin{equation}\label{eq-proof:Agmon}
		\sup_{\gamma\in(0,1]}\normLp{\nabla\phi_{\varepsilon,\gamma}}{\infty}{(\R^d)}\leq 2\varepsilon\;.
	\end{equation}
	By the triangle inequality, one has
	\begin{align*}
		\normLp{e^{\phi_{\varepsilon,\gamma}/\hbar}}{2}{(\R^d)} \leq 1+\normLp{\chi e^{\phi_{\varepsilon,\gamma}/\hbar}}{2}{(\R^d)}\;.
	\end{align*}
	On the one hand, by \cref{eq:ass-Agmon} and since $\supp(\chi)\subset U^c$
	\begin{align*}
		\qhb(\chi e^{\phi_{\varepsilon,\gamma}/\hbar}\psi,\chi e^{\phi_{\varepsilon,\gamma}/\hbar}\psi) &\geq \int_{\R^d} V(x)\abs{\chi(x) e^{\phi_{\varepsilon,\gamma}(x)/\hbar}\psi(x)}^2 \di x 
		\geq \frac\varepsilon 2\normLp{\chi e^{\phi_{\varepsilon,\gamma}/\hbar}\psi}{2}{(\R^d)}^2\;.
	\end{align*}
	On the other hand, given the supports of $\nabla\chi$, $\nabla\phi_{\varepsilon,\gamma}$ and $\phi_{\varepsilon,\gamma}$, one has
	\begin{align*}
		\hbar^2\normLp{\nabla(\chi e^{\phi_{\varepsilon,\gamma}/\hbar})\psi}{2}{(\R^d)}^2 &\leq 2\Big(\normLp{(\chi e^{\phi_{\varepsilon,\gamma}/\hbar}\psi) \nabla\phi_{\varepsilon,\gamma}}{2}{(\R^d)}^2+\hbar^2\normLp{(e^{\phi_{\varepsilon,\gamma}/\hbar}\psi) \nabla\chi}{2}{(\R^d)}^2\Big)
		\\&\leq 2\left(4\varepsilon^2+\normLp{\nabla\chi}{\infty}{(\Omega_\varepsilon\setminus U)}^2\right)\normLp{\psi}{2}{(\Omega_\varepsilon\setminus U)}^2
		\\&\leq 2\left(4\varepsilon^2+\abs{\tilde{U}\setminus U}^2\right)\;.
	\end{align*}
	By gathering everything in the Agmon formula (\cref{lemma:IMS}), one has
	\begin{align*}
		\Big(\frac\varepsilon2-E\Big)\normLp{\chi e^{\phi_{\varepsilon,\gamma}/\hbar}\psi}{2}{(\R^d)}^2 \leq 2\left(4\varepsilon^2+\abs{\tilde{U}\setminus U}^2\right)\;.
	\end{align*}
	Since $E<\varepsilon/4$, this implies the bound
	\begin{align*}
		\normLp{\chi e^{\phi_{\varepsilon,\gamma}/\hbar}\psi}{2}{(\R^d)}^2 \leq 8\left(4\varepsilon+\frac 1\varepsilon\abs{\tilde{U}\setminus U}^2\right)\;,
	\end{align*}
	which is uniform with respect to $\gamma\in(0,1]$. We conclude by taking the limit $\gamma\to 0$.
\end{proof}

We need also a Cwikel--Lieb--Rosenbljum bound (CLR) that provides an upper bound on the integrated Weyl law for magnetic Schrödinger operators.

\begin{lemma}\label{lemma:CLR-Lieb-Seiringer}
	Let $d\geq 1$. Let $\gamma>0$ such that $\frac d2+\gamma\geq 1$, $a\in L^2_\loc(\R^d,\R^d)$, and let $V:\R^d\to\R$ bounded from below and such that $V_+\in L^1_\loc(\R^d)$ and such that $(V+\lambda/2)_-\in L^{\frac d2+\gamma}(\R^d)$. Let $H=(-i\nabla-a(x))^2+V(x)$ be the magnetic Schrödinger operator acting in $L^2(\R^d)$. Then,
	\begin{align*}
		\tr(\indic_{(-\infty,\mu-\lambda]}(H))\leq \lambda^{-\gamma}\frac{2^\gamma\Gamma(\gamma)}{(4\pi)^{d/2}\Gamma\big(d/2+\gamma\big)} \int_{\R^d}\Big(V(x)+\frac\lambda2\Big)_-^{\frac d2+\gamma}\di x\;.
	\end{align*}
\end{lemma}

The proof is an adaptation of \cite[Corollary 4.2]{Lieb-Seiringer-2010Stability}. The constant in front of the integral is the same as for the non-magnetic case and does not depend at all on the magnetic field.

\begin{proof}
	We apply the Birman--Schwinger inequality twice and the diamagnetic inequality for the resolvents (\cref{eq:diamagnetic-ineq_res-kin} on integral kernels)
	this leads to	
	\begin{align*}
		\tr(\indic_{(-\infty,-\lambda]}(H)) &=\tr\Big[\indic_{(-\infty,-\lambda/2]}((-i\nabla-a(x))^2+V(x)-\lambda/2)\Big]
		\\&\leq \tr\Big[\indic_{(-\infty,-\lambda/2]}((-i\nabla-a(x))^2-(V+\lambda/2)_-)\Big]
		\\&\quad= \tr\Big[\indic_{[1,+\infty)}(\sqrt{(V+\lambda/2)_-}((-i\nabla-a(x))^2+\lambda/2)^{-1}\sqrt{(V+\lambda/2)_-}\Big]
		\\&\leq \tr\Big[\indic_{[1,+\infty)}(\sqrt{(V+\lambda/2)_-}(-\Delta+\lambda/2)^{-1}\sqrt{(V+\lambda/2)_-}\Big]
		\\&\leq \tr\Big[\sqrt{(V+\lambda/2)_-}(-\Delta+\lambda/2)^{-1}\sqrt{(V+\lambda/2)_-}\Big]
		\;.
	\end{align*}
	Then, we use the following lemma to the multiplication operator $(V+\lambda/2)_-$ and to $(-\Delta+\lambda/2)^{-1}$ with $m=\frac d2+\gamma$.

	We know that (\cite[Theorem 4.5]{Lieb-Seiringer-2010Stability}) for $A,B$ be positive, self-adjoint operators on a separable Hilbert space and for any $m\geq 1$
		\begin{equation*}
			\tr(B^{1/2}A B^{1/2})^m\leq \tr(B^{m/2}A^m B^{m/2})\;.
		\end{equation*}
	This gives
	\begin{align*}
		\tr(\indic_{(-\infty,-\lambda]}(H))
		&\leq \tr\Big(\sqrt{(V+\lambda/2)_-}^{\frac d2+\gamma}\Big((-\Delta+\lambda/2)^{-1}\Big)^{\frac d2+\gamma}\sqrt{(V+\lambda/2)_-}^{\frac d2+\gamma}\Big)
		\\&\quad= \Big((-\Delta+\lambda/2)^{-1}\Big)^{\frac d2+\gamma}(0,0)\int_{\R^d}\Big(V(x)+\frac\lambda2\Big)_-^{\frac d2+\gamma}\di x
		\\&\quad= \lambda^{-\gamma}\frac{2^\gamma\Gamma(\gamma)}{(4\pi)^{d/2}\Gamma\big(d/2+\gamma\big)} \int_{\R^d}\Big(V(x)+\frac\lambda2\Big)_-^{\frac d2+\gamma}\di x
		\;.
	\end{align*}
	Here $\Big((-\Delta+\lambda/2)^{-1}\Big)^{\frac d2+\gamma}(0,0)$ is the integral kernel of $\Big((-\Delta+\lambda/2)^{-1}\Big)^{\frac d2+\gamma}$ evaluated in the point $(0,0)$.
\end{proof}

We use more exactly the version that involves the semiclassical parameter $\hbar>0$. 

\begin{cor}\label{cor:CLR-semicl}
	Let $d\geq 1$. Let $b>0$, $a\in L^2_\loc(\R^d,\R^d)$, let $V:\R^d\to\R$ bounded from below, and let $\mu>\inf_\R V$,  such that $(V-\mu)_+\in L^1_\loc(\R^d)$ and such that $(V-\mu-\tfrac38\varepsilon)_-\in L^{\frac d2+1}(\R^d)$. Let $H_{\hbar,b}=(-i\hbar\nabla- ba(x))^2+V(x)$ be the magnetic Schrödinger operator acting in $L^2(\R^d)$. Then, for any $\varepsilon>0$, there exist $\hbar_0>0$ and $C=C(d,V,\mu,\varepsilon)>0$ such that for any $\hbar\in(0,\hbar_0]$ and any $b>0$
	\begin{align*}
		\tr(\indic_{(-\infty,\mu+\varepsilon/4]}(H_{\hbar,b}))\leq C\hbar^{-d}\;.
	\end{align*}
\end{cor}

\begin{proof}
	The proof is analogous to \cite[Remark 3.6]{Fournais-Mikkelsen-2020}. It is enough to apply the previous lemma to $\gamma=1$, $\lambda=\frac{\varepsilon}{4\hbar^2}$ and to $H=(-i\nabla- \frac b\hbar a(x))^2+\frac 1{\hbar^2}V(x)\frac 1{\hbar^2}\mu-\frac\varepsilon{2\hbar^2}$
	\begin{align*}
		\tr(\indic_{(-\infty,\mu+\frac\varepsilon4]}(H_{\hbar,b})) &= \tr(\indic_{(-\infty,-\frac\varepsilon{4\hbar^2}]}(H))
		\\&\leq \hbar^{-d}\frac{4\varepsilon^{-\frac d2-1}}{(4\pi)^{d/2}\Gamma\big(d/2+1\big)} \int_{\R^d}\big(V(x)-\mu-\frac{3\varepsilon}{8}\big)_-^{\frac d2+1}\di x\;.	\qedhere
	\end{align*}
\end{proof}

\subsection{Proof of \cref{thm:cb-princ}}
\begin{proof}[Proof of \cref{thm:cb-princ}]

Let $\varepsilon>0$ and the open sets $\Omega_\varepsilon$ and $\Omega_V\subset\R^d$ given by \cref{ass:a-V}.
There exists an open bounded set $\Omega$ such that
\begin{equation}\label{eq-proof:inclusion-sets}
	\{x\in\R^d \: :\: V(x)-\mu<\varepsilon\} \subset\Omega_\varepsilon\subset\overline{\Omega_\varepsilon}\subset\Omega \subset\overline{\Omega}\subset\Omega_V
\end{equation}
Let $\chi,\tilde{\chi}\in\test{\R^d,[0,1]}$ such that $\supp(\chi)\subset\overline{\Omega}$ and $\chi=1$ on $\overline{\Omega_\varepsilon}$, and such that $\supp(\tilde{\chi})\subset\Omega_V$ and $\tilde{\chi}=1$ on $\overline{\Omega}$. By construction, the operator $H_{\hbar,b}$ satisfies \cref{ass:loc-magn} for the bounded set $\Omega$ and for any $u\in\test{\Omega}$, with $a_{\loc}:=\tilde{\chi}a$ and $V_\loc:=\tilde{\chi}(V-\mu)$:
\begin{align*}
	(H_{\hbar,b}-\mu)u=((i\hbar\nabla+ba_{\loc}(x))^2+V_{\loc}(x))u \;
\end{align*}
	Actually, by construction one has $u=\tilde{\chi}u=\tilde{\chi}^2u$ and $\supp(\partial_{x_j}\tilde{\chi})\cap\supp(u)=\emptyset$ for any $j\in\{1,\ldots,d\}$, so that
	\begin{align*}
		(H_{\hbar,b}-\mu)u&=	-\hbar^2\Delta u +2i\hbar b\cdot\nabla u+i\hbar b(\div a) u  +b^2\bar{a}^2 u +(V-\mu)u
		\\&=-\hbar^2\Delta u +2i\hbar b\cdot\nabla(\tilde{\chi} u) +i\hbar b(\div a) \chi u+b^2\abs{a}^2 \tilde{\chi}^2u +(V-\mu)\tilde{\chi}u
		\\&=-\hbar^2\Delta u +2i\hbar b a\cdot(\tilde{\chi}\nabla u) +2i\hbar b(a\cdot\nabla\tilde{\chi})u+i\hbar b(\div a)\tilde{\chi}u +b^2\abs{\tilde{\chi}a}^2u+\tilde{\chi}(V-\mu)u
		\\&=-\hbar^2\Delta u +2i\hbar b (\tilde{\chi}a)\cdot\nabla u +2i\hbar b(a\cdot\nabla\tilde{\chi})u+i\hbar b(\div a)\tilde{\chi} u +b^2\abs{\tilde{\chi}a}^2u+\tilde{\chi}(V-\mu)u
		\\&=-\hbar^2\Delta u +2i\hbar b (\tilde{\chi}a)\cdot\nabla u +i\hbar b\underset{=0}{\underbrace{(a\cdot\nabla\tilde{\chi})u}}+i\hbar b\div(\tilde{\chi}a)u +b^2\abs{\tilde{\chi}a}^2u+\tilde{\chi}(V-\mu)u
		\\&=((i\hbar\nabla+b(\tilde{\chi}a))^2+\tilde{\chi}(V-\mu))u\;.
	\end{align*}
Let $j\in\{1,\ldots,d\}$. We bound the trace norm of the commutators by two localized terms
\begin{align*}
	\norm{[\indic_{(-\infty,\mu]}(H_{\hbar,b}),f(x)]}_1&\leq \norm{[\indic_{(-\infty,\mu]}(H_{\hbar,b}),\chi f(x)]}_1+\norm{[\indic_{(-\infty,\mu]}(H_{\hbar,b}),(1-\chi)f(x)]}_1\\
	\norm{[\indic_{(-\infty,\mu]}(H_{\hbar,b}),P_j]}_1&\leq \norm{[\indic_{(-\infty,\mu]}(H_{\hbar,b}),\chi P_j]}_1+\norm{[\indic_{(-\infty,\mu]}(H_{\hbar,b}),(1-\chi) P_j]}_1\;.
\end{align*}
We first apply \cref{thm:cb-loc-gene} to the operator $H_{\hbar,b}-\mu$ on the open set $\Omega$ that is locally equal to $(-i\hbar\nabla-ba_\loc)^2+V_\loc$, this gives
\begin{align*}
	\norm{[\indic_{(-\infty,\mu]}(H_{\hbar,b}),\chi f(x)]}_1=\norm{[\indic_{(-\infty,0]}(H_{\hbar,b}-\mu),\chi f(x)]}_1\leq C\hbar^{1-d}\;,
    \\\norm{[\indic_{(-\infty,\mu]}(H_{\hbar,b}),\chi P_j]}_1=\norm{[\indic_{(-\infty,0]}(H_{\hbar,b}-\mu),\chi (-i\hbar\partial_{x_j}-b(a_{\loc})_j]}_1\leq C\hbar^{1-d}\;.
\end{align*}
Let us treat the other term localized with the function $1-\chi$ and show that it is negligible. More exactly, let us prove that for any $N\in\N_{\geq 1}$, there exists $C_N>0$ such that
\begin{equation}\label{eq-proof:cb-princ_1}
	\norm{[\indic_{(-\infty,\mu]}(H_{\hbar,b}),(1-\chi)f(x)]}_1,\quad\norm{[\indic_{(-\infty,\mu]}(H_{\hbar,b}),(1-\chi) P_j]}_1\leq C\hbar^{N-d}\;.
\end{equation}
Let us write $R_j$ for $f(x)$ or $P_j$. By the triangle inequality, the cyclicity of the trace and the Cauchy--Schwarz inequality in Schatten spaces
\begin{align*}
	\norm{[\indic_{(-\infty,\mu]}(H_{\hbar,b}),(1-\chi) R_j]}_1 &\leq \norm{\indic_{(-\infty,\mu]}(H_{\hbar,b})(1-\chi) R_j}_1+\norm{(1-\chi) R_j\indic_{(-\infty,\mu]}(H_{\hbar,b})}_1
	\\&\quad= 2 \norm{\indic_{(-\infty,\mu]}(H_{\hbar,b})(1-\chi) R_j}_1
	\\&\quad= 2 \norm{\indic_{(-\infty,\mu]}(H_{\hbar,b})\:\indic_{(-\infty,\mu]}(H_{\hbar,b})(1-\chi) R_j}_1
	\\&\leq 2\tr(\indic_{(-\infty,\mu]}(H_{\hbar,b}))^{1/2}\norm{\indic_{(-\infty,\mu]}(H_{\hbar,b})(1-\chi) R_j}_2^{1/2}\;.
\end{align*}
By the magnetic CLR bound (\cref{cor:CLR-semicl}), we have
\begin{align*}
	\tr(\indic_{(-\infty,\mu]}(H_{\hbar,b}))\leq \tr(\indic_{(-\infty,\mu+\varepsilon/4]}(H_{\hbar,b}))\leq C_{\mu,\varepsilon}\hbar^{-d}\;.
\end{align*}
To control the remaining term we show that for any $N\in\N$ we have
\begin{equation}\label{eq-proof:cb-princ_2}
	\norm{\indic_{(-\infty,\mu]}(H_{\hbar,b})(1-\chi) R_j}_2\leq C_N\hbar^{N-\frac d2}\;.
\end{equation}
In fact, by cyclicity of the trace and by expressing the square of this Hilbert--Schmidt norm with respect to an orthogonal basis of eigenfunctions $(\varphi_n)_{n\in\N}$ of $H_{\hbar,b}$, we have
\begin{equation}\label{eq-proof:cb-princ_3}
\begin{split}
	&\norm{\indic_{(-\infty,\mu]}(H_{\hbar,b})(1-\chi) R_j}_2^2 \\
	&= \tr\Big[\indic_{(-\infty,\mu]}(H_{\hbar,b})(1-\chi) R_j^2(1-\chi)\indic_{(-\infty,\mu]}(H_{\hbar,b})\Big]
	\\& = \sum_{n\in\N,\:\lambda_n\leq 0}\prodscal{\psi_n}{\Big[\indic_{(-\infty,\mu]}(H_{\hbar,b})(1-\chi) R_j^2(1-\chi)\indic_{(-\infty,\mu]}(H_{\hbar,b})\Big]\psi_n}
	\\& =  \sum_{n\in\N,\:\lambda_n\leq\frac \varepsilon4}\prodscal{\psi_n}{\Big[\indic_{(-\infty,\mu]}(H_{\hbar,b})(1-\chi) R_j^2(1-\chi)\indic_{(-\infty,\mu]}(H_{\hbar,b})\Big]\psi_n}\;.
\end{split}
\end{equation}
The inclusion of sets \cref{eq-proof:inclusion-sets} assures that the condition \cref{eq:ass-Agmon} of the magnetic Agmon estimates is satisfied. Then, after applying the Hölder inequality
\begin{align*}
&	\prodscal{\psi_n}{\Big[\indic_{(-\infty,\mu]}(H_{\hbar,b})(1-\chi) f(x)^{2r}(1-\chi)\indic_{(-\infty,\mu]}(H_{\hbar,b})\Big]\psi_n}
	\\&\qquad=\normLp{f(x)^r(1-\chi)\indic_{(-\infty,\mu]}(H_{\hbar,b})\psi_n}{2}{(\R^d)}^2
	\\&\qquad\leq \normLp{(1-\chi)f(x)^r e^{-\phi_{\sqrt{\varepsilon}/8}/\hbar}}{\infty}{(\R^d)}^2\normLp{e^{\phi_{\sqrt{\varepsilon}/8}/\hbar}\indic_{(-\infty,\mu]}(H_{\hbar,b})\psi_n}{2}{(\R^d)}^2.
\end{align*}
We control uniformly the last norm $\normLp{e^{\phi_{\sqrt{\varepsilon}/8}/\hbar}\indic_{(-\infty,\mu]}(H_{\hbar,b})\psi_n}{2}{(\R^d)}$ with \cref{lemma:Agmon} applied to $U=\Omega$ and $\sqrt{\varepsilon}/8$ instead of $\varepsilon$.
By \cref{eq-proof:inclusion-sets} and the definition of $\phi_{\sqrt{\varepsilon}/8}$ we have
\begin{align*}
	\overline{\Omega_\varepsilon} \subset \{x\in\R^d\::\: \dist(x,\Omega)< 1\}=
	\{x\in\R^d\: :\: \phi_{\sqrt{\varepsilon}/8}(x)=0\} \;.
\end{align*}
In other terms, $\phi_{\sqrt{\varepsilon}/8}$ is positive and bounded from below by  on $(\overline{\Omega_\varepsilon})^c$, then it is also the case on $\supp(1-\chi)$.
\begin{align*}
	\normLp{(1-\chi)f(x)^re^{-\phi_{\sqrt{\varepsilon}/8}/\hbar}}{\infty}{(\R^d)} &= \hbar^N\normLp{ \frac{(1-\chi)f(x)}{(\phi_{\sqrt{\varepsilon}/8})^N} \Big(\phi_{\sqrt{\varepsilon}/8}/\hbar\Big)^Ne^{-\phi_{\sqrt{\varepsilon}/8}/\hbar}}{\infty}{(\R^d)}
	\\& \leq C\hbar^N\frac{\sup_{x\in\R^d}\abs{x}^{N}\abs{f(x)}^re^{-\abs{x}}}{\varepsilon^{N/2}\dist(\partial\Omega_\varepsilon,\{x\in\R^d\::\: \dist(x,\Omega)= 1\})^N}
	\\& \leq \tilde{C}\hbar^N\frac{\sup_{x\in\R^d}\abs{x}^{N+rm}re^{-\abs{x}}}{\varepsilon^{N/2}\dist(\partial\Omega_\varepsilon,\{x\in\R^d\::\: \dist(x,\Omega)= 1\})^N}\;.
\end{align*}
Finally, for any $N\in\N$, there exists $C_N=C(N,m,\varepsilon,\Omega)>0$ such that
\begin{equation}\label{eq-proof:cb-princ_cons-Agmon}
	\sup_{r\in\{0,1\}}\sup_{n\in\N}\normLp{f(x)^r(1-\chi)\indic_{(-\infty,\mu]}(H_{\hbar,b})\psi_n}{2}{(\R^d)}^2\leq C_N\hbar^{2N}\;.
\end{equation}
We apply \cref{eq-proof:cb-princ_cons-Agmon} and again \cref{cor:CLR-semicl} into \cref{eq-proof:cb-princ_3} and obtain the bound \cref{eq-proof:cb-princ_2} for $R_j=x_j$
\begin{align*}
	\norm{\indic_{(-\infty,\mu]}(H_{\hbar,b})(1-\chi)f(x)}_2^2 \leq \tr[\indic_{[0,\varepsilon/4]}(H_{\hbar,b}))]C_N\hbar^N \leq \tilde{C}_N\hbar^{2N-d}\;.
\end{align*}
Let us treat the case $R_j=P_j$.
Let $\lambda=1-\inf_{x\in\Omega_\varepsilon}V(x)$ so that $P_j^2\leq H_{\hbar,b}+\lambda$. This bound on operators combined with \cref{lemma:IMS} applied to the Lipschitz function $\phi=0$, to the smooth function $1-\chi$, to the eigenfunction $\indic_{(-\infty,\mu]}(H_{\hbar,b})\psi_n$
\begin{align*}
	&\prodscal{\psi_n}{\Big(\indic_{(-\infty,\mu]}(H_{\hbar,b})(1-\chi) P_j^2(1-\chi)\indic_{(-\infty,\mu]}(H_{\hbar,b})\Big)\psi_n}\\&\qquad\leq \prodscal{\psi_n}{\Big(\indic_{(-\infty,\mu]}(H_{\hbar,b})(1-\chi)(H_{\hbar,b}+\lambda)(1-\chi)\indic_{(-\infty,\mu]}(H_{\hbar,b})\Big)\psi_n}
	\\&\quad\qquad= \prodscal{(1-\chi)\indic_{(-\infty,\mu]}(H_{\hbar,b})\psi_n}{H_{\hbar,b}(1-\chi)\indic_{(-\infty,\mu]}(H_{\hbar,b})\psi_n}
	\\&\qquad\qquad+\lambda\normLp{(1-\chi)\indic_{(-\infty,\mu]}(H_{\hbar,b})\psi_n}{2}{(\R^d)}^2
	\\&\quad\qquad= (\lambda_n+\lambda)\normLp{(1-\chi)\indic_{(-\infty,\mu]}(H_{\hbar,b})\psi_n}{2}{(\R^d)}^2+\hbar^2\normLp{\abs{\nabla(1-\chi)}\indic_{(-\infty,\mu]}(H_{\hbar,b})\psi_n}{2}{(\R^d)}^2
	\\&\qquad\leq \Big(\frac \varepsilon 4+\lambda+\normLp{\nabla(1-\chi)}{\infty}{(\R^d)}^2\Big)\normLp{(1-\chi)\indic_{(-\infty,\mu]}(H_{\hbar,b})\psi_n}{2}{(\R^d)}^2
	\;.
\end{align*}
As before we conclude by using \cref{eq-proof:cb-princ_cons-Agmon} and the CLR estimates (\cref{cor:CLR-semicl}) applied to the potential $V-\mu$ into \cref{eq-proof:cb-princ_3}
\begin{align*}
	\norm{\indic_{(-\infty,\mu]}(H_{\hbar,b})(1-\chi) P_j}_2^2 \leq \tr(\indic_{[0,\varepsilon/4]}(H_{\hbar,b}-\mu))C_N\hbar^N \leq \tilde{C}_N\hbar^{2N-d}\;.
\end{align*}
We gather all bounds to get the desired estimates. This concludes the proof of \cref{thm:cb-princ}.
\end{proof}

\section{Proof of the Propagation of the Semiclassical Structure (\cref{thm:propag-semicl-magn})}\label{sec:propagation-semicl-structure}

\begin{proof}[Proof of \cref{thm:propag-semicl-magn}]

Let $m$ be the degree of the polynomial $V$.
Let us write a bound that allows us to apply the Grönwall lemma to the non-negative function
\begin{equation}\label{eq-proof:semic-str_gronwall-func}
	 t\mapsto\norm{[P,\omega_{N,t}]}_1+\sup_{\alpha\in \R^d}\frac 1{1+\abs{\alpha}}\norm{[e^{i\alpha\cdot x},\omega_{N,t}]}_1+\sum_{\beta\in\N^d,\:\abs{\beta}\leq m}\norm{[x^\beta,\omega_{N,t}]}_1\; .
\end{equation}
Let us prove that for any $t\geq 0$, for any $\alpha\in\R^d$
\begin{equation}\label{eq-proof:semic-str_eix}
	\norm{[e^{i\alpha\cdot x},\omega_{N,t}]}_1 \leq 	\norm{[e^{i\alpha\cdot x},\omega_N]}_1 + 2\abs{\alpha}\int_0^t \di s\norm{[\omega_{N,s},P]}_1 + \frac 1\hbar\int_0^t \di s \norm{[\omega_{N,s},[X_s,e^{i\alpha\cdot\bx}]]}_1
	\;,
\end{equation}
and that
\begin{align}\label{eq-proof:semic-str_moment}
		\norm{[P,\omega_{N,t}]}_1 
		&\leq 	\norm{[P,\omega_N]}_1
		\nonumber
		\\&\quad+2b\sum_{1\leq j<k\leq d}\abs{B_{jk} }\int_0^t \di s \norm{\left[P,\omega_{N,s}\right]}_1
		+\int_0^t \di s \norm{\left[\omega_{N,s},\nabla V\right]}_1
		\nonumber
		\\&\quad+ \left(\int_{\R^d}\di p (1+\abs{p}^2)\abs{\hat{W}(p)}\right)\int_0^t \di s\sup_{p\in \R^d}\frac{1}{1+\abs{p}}\norm{[\omega_{N,s},e^{ip\cdot x}]}_1
		\nonumber
		\\&\quad+ \frac 1\hbar\int_0^t\di s \norm{[\omega_{N,s},[X_s,P]}_1
		\;,
\end{align}
and that for any $\beta\in\N^d$
\begin{align}\label{eq-proof:semic-str_x}
		\norm{[x^\beta,\omega_{N,t}]}_1 
		&\leq \norm{[x^\beta,\omega_N]}_1
		\nonumber
		\\&\quad+\hbar\int_0^t \di s\abs{\beta(\beta-1)}\norm{[\omega_{N,s},x^{\beta-2n_d}]}_1
		\nonumber
		\\&\quad+2b\int_0^t\di s\abs{\beta}\left(\sum_{1\leq j\leq d}\normLp{\partial_{x_j}a_j}{\infty}{}\right)\max\left(\norm{[\omega_{N,s},x^{\beta-n_d}]}_1,\norm{[\omega_{N,s},x^{\beta}]}_1\right)
		\nonumber
		\\&\quad+\frac 1\hbar\int_0^t\di s \norm{[\omega_{N,s},[X_s,x^\beta]}_1
		\;,
\end{align}
where $n_d:=(1,\ldots,1)\in\R^d$.

Let us explain first why it enough to prove \eqref{eq-proof:semic-str_eix}, \eqref{eq-proof:semic-str_moment} and \eqref{eq-proof:semic-str_x}.

On the one hand, let us show that the trace norm with the exchange term can be absorbed by the main trace norm, more precisely for $R$ being $e^{i\alpha\cdot x}$, $P$ or $x^\beta$ for any $\alpha\in\R^d$ and $\beta\in\N^d$
\begin{equation}\label{eq-proof:semic-str_exch-term}
	\begin{split}
		\norm{[\omega_{N,t},[X_t,R]]}_1 \leq \frac 2N\left(\int_{\R^d}|\hat{W}(p)|\di p\right)\norm{[R,\omega_{N,t}]}_1.
	\end{split}
\end{equation}
		In fact, expressing $X_s$ with the Fourier transform of $V$
		\begin{align*}
			X_t &= \frac 1N\int_{\R^d}\di p \hat{W}(p)\Big(e^{ip\cdot x}\omega_{N,t}e^{-ip\cdot x}\Big)\;,
		\end{align*}
		 which yields the expression
		\begin{align*}
			[\omega_{N,t},[X_t,R]] &= \frac 1N\int_{\R^d} \di p\Big[\omega_{N,t},[\hat{W}(p) e^{ip\cdot x}\omega_{N,t}e^{-ip\cdot x},R]\Big]
			||
			\\&=\frac 1N\int_{\R^d} \di p\: \left(\omega_{N,t}[\hat{W}(p) e^{ip\cdot x}\omega_{N,t}e^{-ip\cdot x},R]-[\hat{W}(p) e^{ip\cdot x}\omega_{N,t}e^{-ip\cdot x},R]\omega_{N,t}\right)\;.
		\end{align*}
		Then, by taking the trace norm, using that the operator norm $\norm{\omega_{N,t}}_{\op}\leq 1$ and that $\norm{[\omega_{N,t},R]}_\op=\norm{e^{ip\cdot x}[\omega_{N,t},R]e^{-ip\cdot x}}_\op$, we confirm the claimed \cref{eq-proof:semic-str_exch-term}:
		\begin{align*}
			\norm{[\omega_{N,t},[X_t,R]]}_1 &\leq \frac 2N\int_{\R^d}\di p |\hat{W}(p)| \norm{[e^{ip\cdot x}\omega_{N,t}e^{-ip\cdot x},R]}_1
			\\& \leq \frac 2N\Big(\int_{\R^d} |\hat{W}(p)|\di p\Big) \norm{[\omega_{N,t},R]}_1\;.
		\end{align*}

Furthermore, the assumption on $V$ ensures that that $ \norm{\left[\omega_{N,t},\nabla V\right]}_1$ is bounded by a linear combination of $\norm{\left[\omega_{N,t},x^\beta\right]}_1$ for $\beta\in\N^d$ such that $|\beta|$ is smaller or equal than the highest order of the polynomial function $V$.

Finally, the estimates on the initial data \cref{eq:hyp-propag-semicl-struct}, combined by the one for wave-planes (see point (4) of \cref{rmk:propag-semicl-magn}) imply
\begin{equation*}
	\norm{[P,\omega_N]}_1+\sup_{\alpha\in \R^d}\frac 1{1+\abs{\alpha}}\norm{[e^{i\alpha\cdot x},\omega_N]}_1+\sum_{\beta\in\N^d,\:\abs{\beta}\leq m}\norm{[x^\beta,\omega_N]}_1\leq C_m\hbar N\;.
\end{equation*}
And then, by the Grönwall lemma applied to the function defined in \cref{eq-proof:semic-str_gronwall-func}, under the condition that \eqref{eq-proof:semic-str_eix}, \eqref{eq-proof:semic-str_moment} and \eqref{eq-proof:semic-str_x} hold, we obtain the desired estimates \cref{eq:propag-semicl-struct}.

\smallskip

\item[\bf Beginning of the proof of \eqref{eq-proof:semic-str_eix}, \eqref{eq-proof:semic-str_moment} and \eqref{eq-proof:semic-str_x}.]
Let $R=e^{i\alpha\cdot x}$, $R=P$, or $R=x^\beta$.
We use the fact that $\omega_{N,t}$ satisfies the Hartree--Fock equation \cref{eq:HF-dyn} and the Jacobi identity to get
%
\begin{align*}
	i\hbar\frac d{dt}[R,\omega_{N,t}]
	&= [R,[h_{\rm HF}(t),\omega_{N,t}]]
	%
	= [h_{\rm HF}(t),[R,\omega_{N,t}]]+[\omega_{N,t},[h_{\rm HF}(t),R]]\;.
\end{align*}
Since the multiplication operators $e^{i\alpha\cdot\bx}$, $x^\beta$, and $V+w\ast\rho_t$ commute together, it follows that
\begin{align*}
	i\hbar\frac d{dt}[e^{i\alpha\cdot x},\omega_{N,t}] 
	&=[h_{\rm HF}(t),[e^{i\alpha\cdot x},\omega_{N,t}]]+[\omega_{N,t},[\sum_{k=1}^d P_k^2,e^{i\alpha\cdot x}]]-[\omega_{N,t},[X_t,e^{i\alpha\cdot x}]]\;,
\end{align*}
and
\begin{align*}
	i\hbar\frac d{dt}[P,\omega_{N,t}] 
	&=[h_{\rm HF}(t),[P,\omega_{N,t}]]+[\omega_{N,t},[\Hhb+w\ast\rho_t,P]]-[\omega_{N,t},[X_t,P]]\;,
\end{align*}
and
\begin{align*}
	i\hbar\frac d{dt}[x^\beta,\omega_{N,t}] 
	&=[h_{\rm HF}(t),[x^\beta,\omega_{N,t}]]+[\omega_{N,t},[\sum_{k=1}^d P_k^2,x^\beta]]-[\omega_{N,t},[X_t,x^\beta]]\;.
\end{align*}

\smallskip

\item[{\bf Treatment of the commutator with $e^{i\alpha \cdot x}$ \cref{eq-proof:semic-str_eix}.}] We have
\begin{equation*}\label{eq-proof:semic-str_comm-H-eix}
	[e^{i\alpha\cdot x},\sum_{k=1}^dP_k^2] = \hbar(P\cdot\alpha e^{i\alpha\cdot x}+ e^{i\alpha\cdot x}\alpha\cdot P)\;.
\end{equation*}
This implies
\begin{align*}
	[\omega_{N,t},	[e^{i\alpha\cdot x},\Hhb]]
	&=-\hbar P\cdot\alpha[\omega_{N,t},e^{i\alpha\cdot x}]- [\omega_{N,t},e^{i\alpha\cdot x}]\hbar P\cdot\alpha
	\\&\quad + [\omega_{N,t},P]\cdot\hbar\alpha e^{i\alpha\cdot x} + \hbar\alpha e^{i\alpha\cdot x}\cdot[\omega_{N,t},P]
	\;.
\end{align*}
Let us define the self-adjoint operators
\begin{equation*}
	D_\pm(t) :=  h_{\rm HF}(t)\mp \hbar P\cdot\alpha
	=h_{\rm HF}(t)\pm\hbar(i\hbar\nabla+ba)\cdot\alpha\;.
\end{equation*}
Then, we have
\begin{align*}
	i\hbar&\frac d{dt}\left([e^{i\alpha\cdot x},\omega_{N,t}]\right) +[\omega_{N,t},[X_t,e^{i\alpha\cdot x}]] 
	\\&= [h_{\rm HF}(t),[e^{i\alpha\cdot x},\omega_{N,t}]]+[\omega_{N,t},[\Hhb,e^{i\alpha\cdot x}]]
	\\&= D_+(t)[e^{i\alpha\cdot x},\omega_{N,t}] -[e^{i\alpha\cdot x},\omega_{N,t}]  D_-(t) 
	+[\omega_{N,t},P]\cdot\hbar\alpha e^{i\alpha\cdot x} + \hbar\alpha e^{i\alpha\cdot\bx}\cdot[\omega_{N,t},P]
	\;.
\end{align*}
Let us define the propagators $(\cU_\pm(t,s))_{s\in\R}$ associated to the differential operators $D_\pm(t)$
\begin{equation*}
	\begin{cases}
		i\hbar\partial_t\cU_\pm(t,s) = D_\pm(t)\cU_\pm(t,s), \\
		\cU_\pm(s,s)=\mathbb{I}\;,
	\end{cases}
\end{equation*}
so that by writing Duhamel formula applied to  $\cU_+^*(t,0)[e^{i\alpha\cdot x},\omega_{N,t}]\cU_-(t,0)$ and then composing on the left by $\cU_+(t,0)$ and on the right by $\cU_-^*(t,0)$, we get
\begin{align*}
	[e^{i\alpha\cdot x},\omega_{N,t}] &= \cU_+(t,0)[e^{i\alpha\cdot x},\omega_N]\cU_-^*(t,0)
	\\&\quad
	-i\int_0^t \di s \: \cU_+(t,s)\left\{  [\omega_{N,s},P]\cdot\alpha e^{i\alpha\cdot x} + \alpha e^{i\alpha\cdot x}\cdot[\omega_{N,s},P] \right\}\cU_-^*(t,s)
	\\&\quad
	+\frac i\hbar\int_0^t \di s \:  \cU_+(t,s)[\omega_{N,s},[X_s,e^{i\alpha\cdot x}]] \cU_-^*(t,s) \;.
\end{align*}
Finally, since $\cU_\pm(t,s)$ are unitary, the bounds in the trace norm \cref{eq-proof:semic-str_eix} hold for any $t\in\R_+$.

\smallskip

\item[\bf Treatment of the commutator with the magnetic momentum \cref{eq-proof:semic-str_moment}.]
Let us explicit the second term of the equality associated to the derivative of $[P,\omega_{N,t}]$.
On the one hand, using \cref{lemma:prop-magn-momentum} and the hypothesis that the magnetic field is constant
\begin{equation*}
	\left[\sum_{k=1}^d P_k^2,P\right]=2i\hbar b \Big(\sum_{k=1}^d(B_{kj} P_k)	\Big)_{1\leq j\leq d}\;.
\end{equation*}
On the other hand, we have the explicit expression
\begin{align*}
	[V,P]=i\hbar(\nabla V(x))\;.
\end{align*}
We introduce the propagator $\cU(t,s)$ associated to the self-adjoint operator $h_{\rm HF}(t)$ by
\begin{equation}\label{eq-def:propag-hHF}
	\begin{cases}
		i\hbar\partial_t\cU(t,s)=h_{\rm HF}(t)\cU(t,s),\\
		\cU(s,s)=\mathbb{I}\;.
	\end{cases}
\end{equation}
Then, we can write the derivative of $i\hbar\partial_t\left(\cU^*(t,0)[P,\omega_{N,t}]\cU(t,0)\right)$ and the Duhamel formula; afterward conjugating again by $\cU(t,0)$ we obtain
\begin{align*}
	[P,\omega_{N,t}] &= \cU(t,0)[P,\omega_N]\cU^*(t,0)
	\\&\quad -\frac i\hbar\int_0^t \di s \cU(t,s)\Big(\Big[\omega_{N,s},2i\hbar b\Big(\sum_{k=1}^dB_{kj} P_k	\Big)_{1\leq j\leq d}\Big]
	+i\hbar[\omega_{N,s},\nabla V(x)]
	\\&\qquad\qquad\qquad\qquad\quad
	+\left[\omega_{N,s},[W\ast\rho_s,P]\right]-\left[\omega_{N,s},[X_s,P]\right]\Big)\cU^*(t,s)
	\;.
\end{align*}
Then, by using that $\norm{\cU(t,s)}_\op=1$ , this leads to the bound
\begin{align*}
	\norm{[P,\omega_{N,t}]}_1 &\leq \norm{[P,\omega_N]}_1
	\\&\quad
	+2b \Big(\sum_{1\leq j<k\leq d}|B_{kj}|\Big)\int_0^t \di s \norm{\left[P,\omega_{N,s}\right]}_1
	+\int_0^t \di s \norm{\left[\omega_{N,s},\nabla V(x)\right]}_1
	\\&\quad +\frac1\hbar\int_0^{t} \di s\norm{\left[\omega_{N,s},[W\ast\rho_s,P]\right]}_1
	+\frac 1\hbar\int_0^{t} \di s \norm{\left[\omega_{N,s},[X_s,P]\right]}_1\;.
\end{align*}
Note that $W\ast\rho_t$ commutes with $ba(x)$, so
\begin{align*}
	\norm{[\omega_{N,t},[W\ast\rho_t,P]]}_1 &= \norm{[\omega_{N,t},[W\ast\rho_t,i\hbar\nabla]]}_1
	=\norm{[\omega_{N,t},i\hbar\nabla(w\ast\rho_t)]}_1
	\\&= \norm{\Big[\omega_{N,t},\hbar\int_{\R^d}\di p e^{i x\cdot p} p \hat{W}(p)\hat{\rho}_t(p)\Big]}_{\tr} 
	\\&\leq \hbar\left(\int_{\R^d}\di p |\hat{W}(p)|(1+\abs{p}^2)\right)\sup_{p\in \R^d}\frac{1}{1+\abs{\alpha}}\norm{[\omega_{N,t},e^{i \alpha\cdot x}]}\;.
\end{align*}
Therefore, we get the estimate \cref{eq-proof:semic-str_moment}.

\smallskip

\item[{\bf Treatment of the commutator with the multiplication by the monomials \cref{eq-proof:semic-str_x}.}]
Let us calculate the term $[\omega_{N,t},[\sum_{k=1}^d P_k^2,x^\beta]]$.
For any $j,k\in\{1,\ldots,d\}$ and any $\beta\in\N^d$
\begin{align*}
	[P_k^2,x_j^{\beta_k}] &= \Big[-\hbar^2\partial_{x_k}^2+2ib\hbar a_k(x)\partial_{x_j}+a_k(x)^2,x_j^{\beta_k}\Big]
	\\& = \delta_{jk} \hbar\Big(-\hbar\indic_ {\beta_k\geq 2} \beta_k(\beta_k-1) x_j^{\beta_k-2} +2ib\indic_ {\beta_k\geq 1} \beta_k  x_k^{\beta_k-1}a_k(x)\Big)\;.
\end{align*}
By the assumption on the form of $a_j$ and by the triangle inequality, one has
\begin{equation*}
\begin{split}
	\norm{	\left[\omega_{N,t},\Big[\sum_{k=1}^d P_k^2,x^\beta\Big]\right]}_1&\leq \hbar^2\abs{\beta(\beta-1)}\norm{[\omega_{N,t},x^{\beta-2n_d}]}_1 \\&\quad+2b\hbar\abs{\beta}\Big(\sum_{1\leq k\leq d}\normLp{\partial_{x_k}a_k}{\infty}{}\Big)\max\left(\norm{[\omega_{N,t},x^{\beta-n_d}]}_1,\norm{[\omega_{N,t},x^{\beta}]}_1\right)\,. 
\end{split}
\end{equation*}
Therefore, we apply Duhamel formula to $\cU(t,0)[\omega_{N,t},x^\beta]\cU(t,0)^*$ for the propagator $\cU(t,0)$ given by \cref{eq-def:propag-hHF}, we get the \cref{eq-proof:semic-str_x}.
\end{proof}

\section{Derivation of the Time-Dependent Hartree--Fock Equation (\cref{thm:hf})}\label{sec:validity-HF}
\begin{proof}[Sketch of the Proof of \cref{thm:hf}]
We follow the strategy of \cite{BPS-2014-MFevol} with some simplification due to \cite{BD23}. The first step is lifting the many-body evolution to Fock space, so that the tools of second quantization can be employed. The fermionic Fock space is defined as
\[
\mathcal{F} := \mathbb{C} \oplus \bigoplus_{n=1}^\infty L^2_\textnormal{a}(\Rbb^{dN}) \;.
\]
Elements of the Fock space are sequences $\Psi = (\Psi^{(n)})_{n=0}^\infty$ with $\Psi^{(n)} \in L^2_\textnormal{a}(\Rbb^{dN})$ such that
\[ \lVert \Psi \rVert^2 := \sum_{j=0}^\infty \lVert \Psi^{(n)}\rVert^2_{L^2(\Rbb^{dn})} < \infty \;.\]
On Fock space, the annihilation operators $a_x$ for $x \in \Rbb^d$ can be introduced as densely defined operators by
\[
(a_x \psi)^{(n)}(x_1,\ldots,x_N) := \sqrt{n+1} \psi^{(n+1)}(x,x_1,\ldots,x_n) \;.
\]
Moreover, for a one-particle wave function $f \in L^2(\Rbb^d)$, one defines the corresponding annihilation and creation operators, respectively, by
\[
a(f) := \int \overline{f(x)} a_x \di x \;, \qquad a^*(f) := a(f)^* \;.
\]
One may also introduce $a^*_x$, but this is only defined within inner products where it can be moved to the other argument as $a_x$, or when integrated with a one-particle wave function $f$. It is well-known that in the fermionic case $a(f)$ and $a^*(f)$ are bounded operators, and satisfy the canonical anticommutator relations (CAR), i.\,e., for $f,g \in L^2(\Rbb^d)$ we have
\[
\{a(f),a^*(g) \} := a(f) a^*(g) + a^*(g) a(f) = \langle f,g\rangle\;, \quad \{a(f),a(g)\} = 0 \;, \quad  \{a^*(f),a^*(g)\} = 0\;.
\]
If $O$ is an operator on the one-particle space $L^2(\Rbb^d)$, one defines
\[
\di\Gamma(O) \psi^{(n)} := \sum_{i=1}^n O_i \psi^{(n)}\;, 
\]
where $O_i$ acts on the $i$-th particle. If $O$ has a (distributional) integral kernel $O(x;y)$ one can verify that as a sesquilinear form
\[\di\Gamma(O) := \int O(x;y) a^*_x a_y \di x \di y\;.\]
An example is the particle number operator as the second quantization of the identity
\[
\Ncal = \di\Gamma(\mathds{1}) = \int a^*_x a_x \di x \;.
\]
The vector $\Omega := (1,0,0,\ldots) \in \mathcal{F}$ is called the vacuum vector; on this special vector one defines $a_x \Omega = 0$ for all $x \in \Rbb^d$. In particular it has vanishing particle number in the sense that $\Ncal \Omega = 0$.

The first important observation for us is that, when restricted to the subspace $L^2_\textnormal{a}(\Rbb^{dN})$ of Fock space, the operator 
\[
\mathcal{H}_{\hbar,b} := \di\Gamma(H_{\hbar,b}) + \frac{\lambda}{2} \int W(x-y) a^*_x a^*_y a_y a_x
\]
is identical to the operator $H^N_{\hbar,b}$ defined in \cref{eq:Nbodyham}; moreover the $N$-particle subspace $L^2_\textnormal{a}(\Rbb^{dN})$ is invariant under the action of $\mathcal{H}_{\hbar,b}$. Therefore, if we consider initial data $\Psi_N(0)$ in the $N$-particle subspace, the solution $\Psi_N(t) := e^{-i t\mathcal{H}_{\hbar,b}/\hbar} \Psi_N(0)$ of the Fock-space Schr\"odinger equation $i\hbar \partial_t \Psi_N(t) = \mathcal{H}_{\hbar,b} \Psi_N(t)$ agrees with the solution of the $N$-body Schr\"odinger equation \cref{eq:SE-dyn}. 

The second important observation is that Slater determinants can be written as the application of a unitary transformation in Fock space to the vacuum vector $\Omega$. In fact, if $\Psi_N(0)$ has the orbitals $\psi_n$, for $n = 1,\ldots, N$, we can define the unitary particle--hole transformation $R_N(0): \mathcal{F} \to \mathcal{F}$ by
\[
   R_N(0) := \prod_{n=1}^N \left( a^*(\psi_n) + a(\psi_n) \right) \;.
\]
One easily verifies that, up to a sign that can be absorbed in the ordering of the factors in the definition, we have
\[
R_N(0)\Omega = \prod_{n=1}^N a^*(\psi_n) \Omega = \Psi_N(0) \;.
\]
It is convenient to define the operators
\[ V_N := \sum_{n=1}^N \lvert\psi_n\rangle \langle \overline{\psi_n} \rvert\;, \qquad U_N := \mathds{1}- \sum_{n=1}^N \lvert \psi_n\rangle \langle \psi_n \rvert \]
(notice that $U_N$ is a projection while $V_N$ contains an additional complex conjugation).
One can then show that
\[
   R^*_N a_x R_N = (-1)^N(a(U_{N,x}) - a^*(V_{N,x})) \;, \qquad R^*_N a^*_x R_N = (-1)^N((a^*(U_{N,x}) - a(V_{N,x})) \;,
\]
where in terms of their (distributional) integral kernels we wrote $U_{N,x}(y) = U_N(y;x)$ and $V_{N,x}(y) = V_N(y;x)$. Analogously $R_N(t)$ is defined as the particle-hole transformation such that $R_N(t) \Omega$ is the Slater determinant associated to the orbitals obtainable from the spectral decomposition of the solution of the time-dependent Hartree--Fock equation \cref{eq:HF-dyn}. (The spectral decomposition is not unique, giving rise to a phase ambiguity of the Slater determinant, which may even be non-differentiable. It is well-known that this ambiguity can be resolved by decomposing the initial data to obtain initial orbitals and then evolving these orbitals by an orbital-wise Hartree--Fock equation \cite{lubich2008quantum}.) In the same way, employing the orbitals of the solution of the Hartree--Fock equation, one defines $U_{N,t}$, $V_{N,t}$, and the corresponding $U_{N,t,x}$, $V_{N,t,x}$.

\smallskip

The advantage of having introduced the Fock space formalism and in particular the particle-hole transformation is that we can now introduce the so-called fluctuation dynamics tracking the deviation between many-body and Hartree--Fock evolution by
\[
	\Ucal_N(t,s) := R^*_{N}(t) e^{-i(t-s)\Hcal_{\hbar,b}/\hbar} R_N(s) \;.
\]
By the method of \cite{BD23} it is straightforward to compute the generator of fluctuations
\[
	\Lcal_N(t) := (i \hbar \partial_t R^*_N(t)) R_N(t) + R^*_N(t) \Hcal_{\hbar,b} R_N(t) \;.
\]
More precisely, one is interested in computing $\Lcal_N(t)$ as a normal-ordered expression, i.e.\, one use the CAR to write all its summands with all creation operators to the left of all annihilation operators. Moreover, in view of the proof below it is sufficient to compute only the terms of $\Lcal_N(t)$ that do not commute with the particle number operator $\Ncal$. The obtained summands of $\Lcal_N(t)$ can be grouped into two categories, quadratic in creation or annihilation operators, and quartic in creation or annihilation operators. The key observation is that all the quadratic terms cancel out if the Hartree--Fock equation is inserted in the contribution of $(i \hbar \partial_t R^*_N(t)) R_N(t)$, independent of the choice of the one-particle operator (this was already observed in \cite{BPS-2014-MFevol-relativ} in the context of fermions with pseudo-relativistic kinetic energy). One is left with
\begin{equation}\label{eq:generator}
   \Lcal_N(t) = \Acal_N(t) + \Bcal_N(t) + \Ccal_N(t) + \Mcal_N(t) + \hc
\end{equation}
where
\begin{align}
   \Acal_N(t) & := \frac{1}{2N} \int \di x \di y W(x-y) a^*(U_{N,t,x}) a^*(U_{N,t,y}) a^*(V_{N,t,y}) a^*(V_{N,t,x})\;,\\
   \Bcal_N(t) & := \frac{1}{N} \int \di x \di y W(x-y) a^*(U_{N,t,x}) a^*(U_{N,t,y}) a^*(V_{N,t,x}) a(U_{N,t,y})\;,\\
   \Ccal_N(t) & := \frac{1}{N} \int \di x \di y W(x-y) a^*(U_{N,t,x}) a^*(V_{N,t,x}) a^*(V_{N,t,y}) a(V_{N,t,y})\;.
\end{align}
and $\Mcal_N(t)$ is commutes with the particle number operator, $[\Mcal_N(t),\Ncal] = 0$.

To estimate the difference of $\gamma^{(1)}_{\Psi_N(t)}$ and $\omega_{N,t}$ in terms of the fluctuation dynamics, one uses the fact that the one-particle reduced density matrix can also be computed as
\[\gamma^{(1)}_{\Psi_N(t)}(x;y) = \langle \Psi_N(t), a^*_y a_x \Psi_N(t) \rangle\]
and then uses the unitarity of the particle-hole transformation to obtain
\begin{equation}
\begin{split}
   \gamma^{(1)}_{\Psi_N(t)}(x;y) & = \langle R^*_N(t) e^{-i t \Hcal_{\hbar,b}/\hbar} R_N(0) \Omega, R^*_N(t) a^*_y a_x R_N(t) R^*_N(t) e^{-i t \Hcal_{\hbar,b}/\hbar} R_N(0) \Omega \rangle \\
   & = \langle \Ucal_N(t,0) \Omega, \Big( a^*(U_{N,t,y}) a(U_{N,t,x}) - a^*(\overline{V_{N,t,x}}) a(\overline{V_{N,t,y}}) + \langle \overline{V_{N,t,y}}, \overline{V_{N,t,x}}\rangle\\
   & \qquad \qquad \qquad \ \  + a^*(U_{N,t,y,}) a^*(\overline{V_{N,t,x}}) + a(\overline{V_{N,t,y}}) a(U_{N,t,x}) \Big)\Ucal_N(t,0) \Omega \rangle \;.
\end{split}
\end{equation}
Here one observes that $\langle \overline{V_{N,t,y}}, \overline{V_{N,t,x}}\rangle = \omega_{N,t}(x;y)$, the solution of the Hartree--Fock equation, which can be moved to the left side, leaving the remaining four terms, quadratic in creation or annihilation operators, to be estimated.

\smallskip

The Hilbert--Schmidt norm of the difference can be obtained by pairing $\gamma^{(1)}_{\Psi_N(t)}(x;y) - \omega_{N,t}(x;y)$ with the integral kernel of a Hilbert--Schmidt operator $O$ satisfying $\lVert O\rVert_\HS = 1$, and then taking the supremum over all such operators $O$. The details of the estimate are simpler than for the trace norm and have been shown in \cite[Proof of Theorem 2.1, Step 1]{BPS-2014-MFevol}.

\smallskip

The estimates to obtain the trace norm have also been shown in \cite[Proof of Theorem 2.1, Step 3]{BPS-2014-MFevol}, but since they are slightly more intricate, let us review the main steps. The trace norm of $\gamma^{(1)}_{\Psi_N(t)} - \omega_{N,t}$ is obtained by pairing with a trace-class operator $O$ of norm $\lVert O \rVert_\op = 1$ and taking the supremum over all such operators. The pairing is explicitly given by
\begin{equation}\label{eq:pairing}
\tr\big( O (\gamma^{(1)}_{\Psi_N(t)} - \omega_{N,t} )\big) = \textnormal{I} + \textnormal{II}\;,
\end{equation}
where
\begin{align*}
\textnormal{I} & := \langle \Omega, \Ucal^*_N(t,0) \left( \di\Gamma(U_{N,t} O U_{N,t}) - \di\Gamma(\overline{V_{N,t}} \overline{O^*} V_{N,t})\right) \Ucal_N(t,0) \Omega \rangle \\
\textnormal{II} & := 2 \Re \langle \Omega, U^*_N(t,0) \int \di r_1 \di r_2 (V_{N,t} O U_{N,t})(r_1;r_2) a_{r_1} a_{r_2} \Ucal_N(t,0) \Omega \rangle\;.
\end{align*}
The first line is easily estimated using the standard bound $\lvert \langle \Psi, \di\Gamma(A) \Psi \rangle \rvert \leq \lVert A \rVert_\op \langle \Psi, \Ncal \Psi \rangle$ for any $\Psi \in \mathcal{F}$ and any bounded operator $A$ acting on $L^2(\Rbb^d)$, followed by $\lVert U_{N,t} \rVert =1$ and $\lVert V_{N,t} \rVert = 1$:
\begin{equation}
   \lvert \textnormal{I} \rvert \leq C \lVert O\rVert \langle \Omega, \Ucal^*_N(t,0) \Ncal \Ucal_N(t) \Omega \rangle \;.
\end{equation}
The estimate of the second term is a little more involved if we want to get the optimal bound. We define an auxiliary generator by dropping the $\Acal_N(t)$--term from \cref{eq:generator}, i.\,e.,
\begin{equation}
   \widetilde{\Lcal}_N(t) := \Bcal_N(t) + \Ccal_N(t) + \Mcal_N(t) + \hc
\end{equation}
and we define $\Ucal^{(0)}_N(t,s)$ as the evolution generated by $\Ucal^*_N(t,0) \widetilde{\Lcal}_N(t) \Ucal_N(t,0)$, more explicitly the evolution solving the initial value problem
\[
i \hbar \partial_t \Ucal^{(0)}_N(t,s) = - \Ucal^*_N(t,0) \widetilde{\Lcal}_N(t) \Ucal_N(t,0) \Ucal^{(0)}_N(t,s) \;, \qquad \Ucal^{(0)}_N(s,s) = \mathds{1} \;.
\]
Next we define
\[
\Ucal^{(1)}_N(t,s) := \Ucal_N(t,s) \Ucal^{(0)}_N(t,s)
\]
and verify that its generator is
\[
\begin{split}
   \Lcal^{(1)}_N(t) & = \Lcal_N(t) - \Lcal^{(0)}_N(t) = \Acal_N(t) + \hc \\
    & = \frac{1}{2N} \int \di x \di y W(x-y) a^*(U_{N,t,x}) a^*(U_{N,t,y}) a^*(V_{N,t,y}) a^*(V_{N,t,x}) + \hc
\end{split}
\]
Note that $\Lcal^{(1)}_N(t)$ creates or annihilates four particles at a time, and therefore $[\Lcal^{(1)}_N(t),i^\Ncal] = 0$. By Duhamel's formula one concludes that $\Ucal^{(1)}_N(t,0) i^\Ncal \Ucal^{(1)}_N(t,0) = i^\Ncal$. Next, one expands
\begin{align}
 \textnormal{II} & = 2\Re \left( \textnormal{II}_1 + \textnormal{II}_2 + \textnormal{II}_3 \right)
 \end{align}
 with
 \begin{align}
 \textnormal{II}_1 & := \langle \Omega, \Ucal^{(1)*}_N(t,0) \int \di r_1 \di r_2 (V_{N,t} O U_{N,t})(r_1;r_2) a_{r_1} a_{r_2} \Ucal^{(1)}_N(t,0) \Omega \rangle \;, \\
 \textnormal{II}_2 & := \langle \Omega, \left( \Ucal^*_N(t,0) - U^{(1)*}_N(t,0) \right) \int \di r_1 \di r_2 (V_{N,t} O U_{N,t})(r_1;r_2) a_{r_1} a_{r_2} \Ucal^{(1)}_N(t,0) \Omega \rangle \;, \\
 \textnormal{II}_3 & := \langle \Omega, \Ucal^*_N(t,0) \int \di r_1 \di r_2 (V_{N,t} O U_{N,t})(r_1;r_2) a_{r_1} a_{r_2} \left( \Ucal_N(t,0) - \Ucal^{(1)}_N(t,0) \right) \Omega \rangle \;.
\end{align}
The key point of this expansion is the observation that $\textnormal{II}_1 = 0$; this can be seen by inserting $1= i^0 = i^\Ncal$ before the vacuum vector in the second argument of the inner product, commuting it through $\Ucal^{(1)}_N(t,0)$, pulling it to the left of the operators $a_{r_1} a_{r_2}$ where it becomes $i^{\Ncal + 2}$, commuting it through $\Ucal^{(1)*}_N(t,0)$, and having it act on the vacuum vector in the first argument of the inner product as $i^2 = -1$, so that $\textnormal{II}_1 = -\textnormal{II}_1$.

It remains to estimate $\textnormal{II}_2$ and $\textnormal{II}_3$; both work analogously, so we explain the estimate only for $\textnormal{II}_2$. We expand $\Ucal_N(t,0)$ by the Duhamel formula as
\[
\Ucal_N(t,0) - \Ucal^{(1)}_N(t,0) = - \frac{i}{\hbar} \int_0^t \di s \Ucal(t,s) \widetilde{\Lcal}_N(s) \Ucal^{(1)}_N(s,0)
\]
and insert this formula in $\textnormal{II}$ to get
\begin{align*}
   \textnormal{II}_2 & \leq \frac{4}{\hbar} \left\lvert \left\langle \Omega, \int_0^t \di s \Ucal^{(1)*}_N(s,0) (\Bcal_N(t) + \hc) \Ucal^*_N(t,s) \int \di r_1 \di r_2 (V_{N,t} O U_{N,t})(r_1;r_2) a_{r_1} a_{r_2}  \Ucal^{(1)}_N(t,0) \Omega \right\rangle \right\rvert \\
   & \quad + \frac{4}{\hbar} \left\lvert \left\langle \Omega, \int_0^t \di s \Ucal^{(1)*}_N(s,0) (\Ccal_N(t) + \hc) \Ucal^*_N(t,s) \int \di r_1 \di r_2 (V_{N,t} O U_{N,t})(r_1;r_2) a_{r_1} a_{r_2}  \Ucal^{(1)}_N(t,0) \Omega \right\rangle \right\rvert \;.
\end{align*}
The estimate for the term with $\Ccal_N(t)$ can be obtained similarly to the estimate for the term containing $\Bcal_N(t)$, which we shall now discuss. By expanding the potential into Fourier modes, $W(x-y) = \int \di p \hat{W}(p) e^{ip\cdot x} e^{-ip\cdot y}$, we can write
\begin{align*}
\Bcal_N(t) & = \frac{1}{N} \int \di x \di y W(x-y) a^*(U_{N,t,x}) a^*(U_{N,t,y}) a^*(V_{N,t,x}) a(U_{N,t,y}) \\
& = \frac{1}{N} \int \di p \hat{W}(p) \int \di x a^*(U_{N,t,x}) e^{ip\cdot x} a^*(V_{N,t,x}) \int \di y a^*(U_{N,t,y}) e^{ip\cdot y} a(U_{N,t,y})\\
& = \frac{1}{N} \int \di p \hat{W}(p) \int \di r_1 \di r_2 (U_{N,t} e^{ip\cdot x} V_{N,t})(r_1,r_2) a^*_{r_1} a^*_{r_2} \di\Gamma\left( U_{N,t} e^{ip\cdot x} \overline{U_{N,t}}\right) \;,
\end{align*}

where on the last line $x$ is to be read as a multiplication operator. Using the Cauchy--Schwartz inequality and the well-known estimates \cite[Lemma~3.1]{BPS-2014-MFevol} for $\Psi \in \mathcal{F}$ and $A$ a bounded or Hilbert--Schmidt operator
\begin{align*}
   \lVert \di\Gamma(A) \Psi \rVert & \leq \lVert A \rVert_\op \lVert \Ncal \Psi \rVert \\
   \lVert \int \di r_1 \di r_2 A(r_1;r_2) a_{r_1} a_{r_2} \Psi \rVert & \leq \lVert A \rVert_\HS \lVert \Ncal^{1/2} \Psi \rVert
\end{align*}
one obtains that
\begin{align*}
   \textnormal{II} & \leq \frac{C}{\hbar N} \int_0^t \di s \int \di p \lvert \hat{W}(p) \rvert  \lVert U_{N,t} e^{ip\cdot x} \overline{U_{N,t}} \rVert_\op \lVert U_{N,t} e^{ip\cdot x} V_{N,t} \rVert_\HS \lVert V_{N,t} O U_{N,t} \rVert_\HS \\
   & \qquad \times \langle \Ucal^{(1)*}_N(t,0) \Omega, \Ncal^2 \Ucal^{(1)}_t(t,0) \Omega \rangle \;.
\end{align*}
Two of these norms are straightforward to estimate, namely
\[
\lVert U_{N,t} e^{ip\cdot x} \overline{U_{N,t}} \rVert_\op \leq \lVert U_{N,t}\rVert_\op \lVert e^{ip\cdot x}\rVert_\op \lVert \overline{U_{N,t}} \rVert_\op \leq 1
\] and
\[
\lVert V_{N,t} O U_{N,t} \rVert_\HS \leq \lVert V_{N,t}\rVert_\HS \lVert O\rVert_\op \lVert U_{N,t} \rVert_\op \leq \lVert O \rVert_\op N^{1/2}\;.
\]
The semiclassical structure is crucial to estimate the third of these norms; in fact, recalling $U_{N,t} = 1 - \omega_{N,t}$ and $U_{N,t} V_{N,t} = 0$, by the H\"older inequality for Schatten norms we obtain
\begin{equation*}
\begin{split}
\lVert U_{N,t} e^{ip\cdot x} V_{N,t} \rVert_\HS^2 & = \lVert [U_{N,t}, e^{ip\cdot x}] V_{N,t} \rVert_\HS^2 \leq  \lVert [\omega_{N,t},e^{ip\cdot x}] \rVert_\HS^2
\lVert V_{N,t}\rVert_\op^2 \\
& = \lVert [\omega_{N,t},e^{ip\cdot x}]^* [\omega_{N,t},e^{ip\cdot x}] \rVert_\tr \leq  \lVert [\omega_{N,t},e^{ip\cdot x}] \rVert_\op \lVert [\omega_{N,t},e^{ip\cdot x}] \rVert_\tr \\
& \leq 2 \lVert \omega_{N,t}\rVert_\op \lVert e^{ip\cdot x}\rVert_\op \lVert [\omega_{N,t},e^{ip\cdot x}] \rVert_\tr = 2 \lVert [\omega_{N,t},e^{ip\cdot x}] \rVert_\tr\;.
\end{split}
\end{equation*}
This trace norm is exactly the one controlled by the propagated semiclassical commutator estimate \cref{eq:propag-semicl-struct}. (The corresponding bound with the momentum operator in the exponent is not needed for the many-body analysis, it is only required to be able to close the Gr\"onwall estimate in the proof of \cref{thm:propag-semicl-magn}.) We conclude that
\begin{align*}
   \textnormal{II} & \leq C \lVert O\rVert_\op \hbar^{-1/2} \langle \Ucal^{(1)*}_N(t,0) \Omega, \Ncal^2 \Ucal^{(1)}_t(t,0) \Omega \rangle \;.
\end{align*}
Finally, by a Gr\"onwall argument \cite[Theorem~3.2]{BPS-2014-MFevol} using similar estimates to those employed up to here, one proves that
there exist constants $c_1, c_2 > 0$ such that for all $t \in \Rbb$ we have for all $k\in\N$
\begin{align*}
\langle \Ucal^{(1)*}_N(t,0) \Omega, (\Ncal+1)^k \Ucal^{(1)}_N(t,0) \Omega \rangle \leq \exp(c_1 \exp(c_2 \lvert t\rvert)) \langle  \Omega, (\Ncal+1)^k \Omega \rangle
\end{align*}
uniformly in particle number $N$ and in $\hbar$. This completes the sketch of the proof; for details we refer to \cite{BPS-2014-MFevol} with the simplification of the computation of the generator due to \cite{BD23}.
\end{proof}
\newpage
\appendix

\section{Spectral Theory of Magnetic Schrödinger Operators}\label{sec:basic-prop-magnetic}

We assume that $a:\R^d\to\R^d$ with $a_j\in L^2_\loc(\R^d)$ for any $j\in\{1,\ldots,d\}$ and $V:\R^d\to\R$ with the decomposition $V=V_+-V_-$, such that $V_-\in L^p(\R^d)+L^\infty_\varepsilon(\R^d)$ where $p$ satisfies \cref{def:magnetic-Sch-op}
and such that $V_+\in L^1_\loc(\R^d)$ with the confining property $V_+(x)\to \infty$ as $\abs{x}\to \infty$.

We define $\mathfrak{q}_{\hbar,b}$ as the quadratic form that corresponds to $H_{\hbar,b}^{\min}$, given by on $\test{\R^d}$	
	\begin{equation}\label{def:quadratic-form}
		u\mapsto\int_{\R^d}\abs{(i\hbar\nabla+ba(x))u(x)}^2 \di x +\int_{\R^d}V(x)\abs{u(x)}^2 \di x		\;.
	\end{equation}

 \begin{lemma}\label{lemma:quadr-closable-Friedrich}
 	Under the above assumption, the quadratic form $\mathfrak{q}_{\hbar,b}$ defined in \cref{def:quadratic-form} is closable, with the closure's domain of form
    \begin{equation*}
	 	\dom(\mathfrak{q}_{\hbar,b}):=\{ u\in L^2(\R^d) \: :\: (-i\hbar\nabla-ba(x))u\in L^2(\R^d), \: \sqrt{V_+}u\in L^2(\R^d)\}\;,
	 \end{equation*}

    and admits a Friedrichs extension $H_{\hbar,b}$ with domain
 	\begin{align*}
 		\dom(H_{\hbar,b}):=\{ u\in L^2(\R^d) \: :\: 
 		u\in\dom(\mathfrak{q}_{\hbar,b}) ,\quad (i\hbar\nabla+ba(x))^2u+V(x)u\in L^2(\R^d)\}
        \;.
 	\end{align*}
 \end{lemma}
 
 
 \begin{proof}
 	Since $V$ in bounded from below, the quadratic form $\mathfrak{q}_{\hbar,0,V_+}:u\mapsto \int_{\R^d}\abs{-i\hbar\nabla u}^2\di x+\int_{\R^d}V_+\abs{u}^2\di x$ is bounded from below is closed and as the domain
 	\begin{equation*}
 		\dom(\mathfrak{q}_{\hbar,0,V_+})=\{ u\in H^1(\R^d), \: \sqrt{V_+}u\in L^2(\R^d)\}\;.
 	\end{equation*}
 	Moreover, the form operator by $\mathfrak{q}_{-V_-}:u\mapsto -\int_{\R^d}V_-\abs{u}^2\di x$ is relatively bounded by $-\Delta$, i.\,e., for any $\varepsilon>0$, there exists $C_\varepsilon>0$ such that for any $u\in H^1(\R^d)$
 	\begin{equation*}
 		\abs{\int_{\R^d}V_-(x)\abs{u(x)}^2 \di x}\leq \varepsilon\int_{\R^d}\abs{-\Delta u(x)}^2\di x +C_\varepsilon\int_{\R^d}\abs{u(x)}^2\di x\;.
 	\end{equation*}
 	Then, by the diamagnetic inequality \eqref{eq:diamagnetic-ineq_dir}, the quadratic form $\mathfrak{q}_{\hbar,0,V_+}:u\mapsto \int_{\R^d}\abs{(i\hbar\nabla+b)u}^2\di x+\int_{\R^d}V_+\abs{u}^2\di x$ is bounded from below, is closed, and has the domain
 	\begin{equation*}
 		\dom(\mathfrak{q}_{\hbar,b,V_+})=\{ u\in L^2(\R^d) \: :\: (i\hbar\nabla+ba(x))u\in L^2(\R^d), \: \sqrt{V_+}u\in L^2(\R^d)\}\;,
 	\end{equation*}
 	and the quadratic form $\mathfrak{q}_{-V_-}$ is also relatively bounded by $(-i\hbar\nabla-ba(x))^2$
 	\begin{equation*}
 		\abs{\int_{\R^d}V_-(x)\abs{u(x)}^2 \di x}\leq \varepsilon\int_{\R^d}\abs{(i\hbar\nabla+b a(x))u(x)}^2\di x +C_\varepsilon\int_{\R^d}\abs{u(x)}^2\di x\;.
 	\end{equation*}
 	By Milgram's theorem, there exists a unique self-adjoint operator $H_{\hbar,b}$ on $L^2(R^d)$ with
 	\begin{align*}
 		\dom(H_{\hbar,b}):=\{ u\in L^2(\R^d) \: :\: 
 		u\in\dom(\mathfrak{q}_{\hbar,b}) ,\ (i\hbar\nabla+ba(x))^2u+V(x)u\in L^2(\R^d)\}
 	\end{align*}
    as domain.
 \end{proof}
 %
 
 
\begin{lemma}\label{lemma:compact-res}
	Under the above assumption, the operator $H_{\hbar,b}=(-i\hbar\nabla-ba(x))^2+V(x)$ on $L^2(\R^d)$ has a compact resolvent.
\end{lemma}

We provide a pedestrian proof analogous to the one in \cite[Theorem 5.31]{Lewin-TS_et_MQ}.

\begin{proof}
	Let $(u_n)_{n\in\N}\subset L^2(\R^d)$ such that $u_n\rightharpoonup 0$ weakly in $L^2(\R^d)$. Let us show that the sequence defined by $v_n:=(H_{\hbar,b}+i)^{-1}u_n\to 0$ converges in  $L^2(\R^d)$.
	
By definition, each term $v_n\in\dom(H_{\hbar,b})$, in particular $(\sqrt{V_+}u)_{n\in\N}\subset L^2(\R^d)$, is bounded.

 By the diamagnetic inequality \eqref{eq:diamagnetic-ineq_dir} we have
\begin{align*}
	\int_{\R^d}\abs{(i\hbar\nabla)\abs{v_n}(x)}^2 \di x\leq \int_{\R^d}\abs{(i\hbar\nabla+ba(x))v_n(x)}^2 \di x +\int_{\R^d}\abs{V(x)}\abs{v_n(x)}^2 \di x\;.
\end{align*}
Furthermore, by the inclusion $\dom(H_{\hbar,b})\subset	\dom(\mathfrak{q}_{\hbar,b})$, one has  that $(\mathfrak{q}_{\hbar,b}(v_n))_{n\in\N}$ is uniformly bounded, too. Then the sequence $(\abs{v_n})_{n\in\N}\subset H^1(\R^d)$ is bounded. By Rellich--Kondrachov, for any $\Omega\subset\R^d$ and any $s>0$, the injection $H^s(\R^d)\hookrightarrow L^2(\Omega)$ is continuous; therefore $(\abs{v_n})_{n\in\N}\subset L^2(B(0,R))$ is bounded for any $R>0$.
Eventually we take $n\to \infty$, then $R\to 0$ on the right-hand side of the bound
\begin{align*}
	\int_{\R^d}\abs{v_n(x)}^2\di x \leq \int_{B(0,R)}\abs{v_n(x)}^2\di x + \Big(\inf_{B(0,R)^c}V_+\Big)^{-1}\int_{\R^d}V_+(x)\abs{v_n(x)}^2\di x\;.
\end{align*}
This provides the limit $v_n:=(H_{\hbar,b}+i)^{-1}u_n\to 0$ in $L^2(\R^d)$.
\end{proof}

As a consequence of this lemma, by the spectral theorem, the self-adjoint operator $H_{\hbar,b}$ has  discrete spectrum with eigenvalues $(\lambda_n)_{n\in\N}$ of finite multiplicities
\[ \lambda_1\leq\lambda_1\leq\lambda_2\leq \ldots \lambda_n\to+\infty .\]

\section{Proof of \cref{thm:cb-loc-non-critical-cond}}\label{sec:proof-cb-loc-non-critical-cond}

\begin{proof}[Proof of \cref{thm:cb-loc-non-critical-cond}]
    
By the regularity of $V$ and the non-criticality assumption \cref{eq:non-critical-cond}, there exists $\varepsilon>0$ such that for any $x\in B(0,2R)$ and any $E\in[-\varepsilon,\varepsilon]$,
\begin{equation}
	\abs{V(x)-E}+\abs{\nabla V(x)}^2+\hbar\geq \frac c2.
\end{equation}
To simplify the proof let us assume that $\varepsilon\leq 1$.
Let $g_-\in\test{\R,[0,1]}$ such that $\supp(g_-)\subset[-\varepsilon,0]$ such that $g_-=1$ on $[-\tfrac\varepsilon2,0]$, let $g_+\in\test{\R,[0,1]}$ such that $g_-(\Hhb)+g_+(\Hhb)=\indic_{(-\infty,0]}(\Hhb)$.
Recall that one can write
\begin{align*}
	[\indic_{(-\infty,0)}(\Hhb),\varphi P_j^r ] 
	&= \indic_{(-\infty,0]}(\Hhb)\varphi P_j^r 
	\indic_{(-\infty,0]}(\Hhb)^\perp-\indic_{(-\infty,0]}(\Hhb)^\perp\varphi P_j^r \indic_{(-\infty,0]}(\Hhb)
	\\&= \indic_{(-\infty,0]}(\Hhb)\varphi P_j^r 
	\indic_{(0,+\infty)}(\Hhb)-\indic_{(0,+\infty)}(\Hhb)\varphi P_j^r \indic_{(-\infty,0]}(\Hhb)\;.
\end{align*}
Let us treat the first term (we treat likewise the second one). 
Let $f\in\test{[-2\varepsilon,\varepsilon]}$, such that $0\leq f\leq 1$ and such that $f=1$ on $[-\varepsilon,0]$.
\begin{align*}
	\indic_{(-\infty,0]}(\Hhb)\varphi P_j^r \indic_{(0,+\infty)}(\Hhb)
	&=(g_+(\Hhb)+g_-(\Hhb))\varphi P_j^r \indic_{(0,+\infty)}(\Hhb)
	\\&= [g_+(\Hhb),\varphi P_j^r ]\indic_{(0,+\infty)}(\Hhb)
	\\&\quad+ g_-(\Hhb)[f(\Hhb),\varphi P_j^r ]\indic_{(\varepsilon,+\infty)}(\Hhb)
	\\&\quad+g_-(\Hhb)\varphi P_j^r \indic_{[0,\varepsilon]}(\Hhb)\;.
\end{align*}
The trace norm of the first two terms in the right-hand side of the equality are bounded with \cref{lemma:aux-loc_3} applied to $f=g_\pm$.
\begin{align*}
	\norm{[g_+(\Hhb),\varphi P_j^r ] \indic_{(0,+\infty)}(\Hhb)}_1+\norm{g_-(\Hhb)[f(\Hhb),\varphi P_j^r ]\indic_{(\varepsilon,+\infty)}(\Hhb)}_1\leq C\crochetjap{b}^r\hbar^{1-d}\;.
\end{align*}
Let us focus on the term $g_-(\Hhb)\varphi P_j^r \indic_{[0,\varepsilon]}(\Hhb)$. We define the dyadic decomposition of the real line
\begin{align*}
	\forall n\in\N,\quad \chi_{n,\hbar}^-(t):=
	\begin{cases}
		\indic_{[-\hbar,0]}(t) &\text{if}\ n=0,\\
		\indic_{[-4^n\hbar,-4^{n-1}\hbar]}(t) &\text{if}\ n\geq 1\;,
	\end{cases}
		\quad\text{ and }\quad  \chi_{n,\hbar}^+(t):= \chi_{n,\hbar}^-(-t)\;.
\end{align*}
Let $N_\hbar\in\N$ large enough so that
\begin{align*}
	g_-(\Hhb) = \sum_{n=0}^{N_\hbar}g_-(\Hhb)\chi_{n,\hbar}^-(\Hhb),\quad
	\indic_{[0,\varepsilon]}(\Hhb) =\sum_{m=0}^{N_\hbar} \indic_{[0,\varepsilon]}(\Hhb)\chi_{m,\hbar}^+(\Hhb)\;.
\end{align*}
In particular,
\begin{align*}
	\norm{g_-(\Hhb)\varphi P_j^r\indic_{[0,\varepsilon]}(\Hhb)}_1\leq \sum_{m=0}^ {N_\hbar}\sum_{n=0}^ {N_\hbar}\norm{\chi_{n,\hbar}^-(\Hhb)\varphi P_j^r  \chi_{m,\hbar}^+(\Hhb)}_1 \;.
\end{align*}
Let us separate the sum into several sums:
\begin{itemize}
	\item[(i)] when $m=n=0$,
	\item[(ii)] when $m=0$ and $1\leq n\leq N_\hbar$, or symmetrically $n=0$ and $1\leq m\leq N_\hbar$,
	\item[(iii)] $m,n\geq 1$ such that $m\geq n$ or $n\leq m$,
\end{itemize}
and let us show that the double series is summable and bounded by $\hbar^{1-d}$ up to a constant.

\smallskip

Consider case (i), $n=m=0$. By the Cauchy-Schwartz inequality, the cyclicity of the trace and the magnetic Weyl law in the intervals $[-\hbar,0]$ and $[0,\hbar]$ (\cref{thm:WL-finite-int})
\begin{align*}
	\norm{\chi_{0,\hbar}^-(\Hhb)\varphi \chi_{0,\hbar}^+(\Hhb)}_1\leq \sqrt{\tr\left(\sqrt{\varphi}\indic_{[-\hbar,0]}(\Hhb)\right)}\sqrt{\tr\left(\sqrt{\varphi}\indic_{[0,\hbar]}(\Hhb)\right)}
	\leq C\hbar^{1-d} \;.
\end{align*}
Moreover, by considering $f_\hbar\in\test{\R}$ such that $\supp(f_\hbar)\subset\supp(\chi_{0,\hbar}^+)$ 
\begin{align*}
	&\norm{\chi_{0,\hbar}^-(\Hhb)\varphi P_j\chi_{0,\hbar}^+(\Hhb)}_1 = \norm{\chi_{0,\hbar}^-(\Hhb)\varphi P_j f_\hbar(\Hhb)\chi_{0,\hbar}^+(\Hhb)}_1
	\\&\quad\leq \norm{\chi_{0,\hbar}^-(\Hhb) [\varphi P_j ,f_\hbar(\Hhb)]\chi_{0,\hbar}^+(\Hhb)}_1 +\norm{\chi_{0,\hbar}^-(\Hhb) f_\hbar(\Hhb)\varphi P_j\chi_{0,\hbar}^+(\Hhb)}_1 \;.
\end{align*}
We control the first right-hand side term by the H\"older inequality, by the functional calculus and \cref{lemma:aux-loc_3}
\begin{align*}
	\norm{\chi_{0,\hbar}^-(\Hhb) [\varphi P_j ,f_\hbar(\Hhb)]\chi_{0,\hbar}^+(\Hhb)}_1&\leq \norm{\chi_{0,\hbar}^-(\Hhb)}_{\op}\norm{[\varphi P_j ,f_\hbar(\Hhb)]}_1\norm{\chi_{0,\hbar}^+(\Hhb)}_{\op}
	\\&\leq C\crochetjap{b}^r\hbar^{1-d}\;.
\end{align*}
Furthermore, by the cyclicity of the trace and by \cref{thm:WL-finite-int} instead of \cref{lemma:aux-loc_3}
\begin{align*}
	\norm{\chi_{0,\hbar}^-(\Hhb) f_\hbar(\Hhb)\varphi P_j\chi_{0,\hbar}^+(\Hhb)}_1 &\leq \norm{\chi_{0,\hbar}^-(\Hhb)}_\op\norm{\varphi f_\hbar(\Hhb)}_1\norm{\chi_{0,\hbar}^+(\Hhb)}_\op
	\\&\leq C\hbar^{1-d}\;.
\end{align*}
Hence, there exists $C>0$ such that
\begin{align*}
	\norm{\chi_{0,\hbar}^-(\Hhb) \varphi P_j\chi_{0,\hbar}^+(\Hhb)}_1\leq C\crochetjap{b}^r\hbar^{1-d}.
\end{align*}

For the cases (ii), (iii) , $m\geq n$ and $n\geq m$, let us assume that $m\geq n$. We show that 
\begin{equation}\label{eq-proof:non-critical-cond_dyadic}
	\norm{\chi_{n,\hbar}^\mp(\Hhb)\varphi P_j^r \chi_{m,\hbar}^\pm(\Hhb)}_1
	\leq  \frac C{4^{\frac {3m}4-\frac n2}}\crochetjap{b}^{2+r}\hbar^{1-d} \;.
\end{equation}

By the same calculation in \cite{Fournais-Mikkelsen-2020}, one has
\begin{equation}\label{eq-proof:non-critical-cond_1}
	\norm{\chi_{n,\hbar}^\mp(\Hhb)\varphi P_j^r \chi_{m,\hbar}^\pm(\Hhb)}_1
	\leq  \frac 1{4^{2(m-1)}\hbar^2}\norm{\chi_{n,\hbar}^\mp(\Hhb)[[\varphi P_j^r,\Hhb],\Hhb]\chi_{m,\hbar}^\pm(\Hhb)}_1 \;.
\end{equation}
%
	 The idea consists in inserting $(\Hloc+2^{2n-1}\hbar)$ and the associated resolvent $(\Hloc+2^{2n-1}\hbar)^{-1}$ thanks the functional calculus and the fact that $-2^{2n-1}\hbar\in[-4^{n}\hbar,-4^{(n-1)}\hbar]$. By inserting $(\Hhb+2^{2n-1}\hbar)$ and the inverse, by the functional calculus and the inequality
	 \begin{equation*}
	 	\forall t\in\R,\quad(t+2^{2n-1}\hbar)^{-1}\chi_{m,\hbar}^+(t)\leq (2^{2m-1}\hbar+2^{2n-1}\hbar)^{-1}\chi_{m,\hbar}^+(t)
	 \end{equation*}
	 one has
	 \begin{align*}
	 	&\norm{\chi_{n,\hbar}^-(\Hhb)\varphi P_j^r \chi_{m,\hbar}^+(\Hhb)}_1
	 	\\&\qquad=\norm{\chi_{n,\hbar}^-(\Hhb)\varphi P_j^r (\Hhb+2^{2n-1}\hbar)(\Hhb+2^{2n-1}\hbar)^{-1}\chi_{m,\hbar}^+(\Hhb)}_1
	 	\\&\qquad\leq \frac{1}{2^{2m-1}\hbar+2^{2n-1}\hbar}\norm{\chi_{n,\hbar}^-(\Hhb)\varphi P_j^r(\Hhb+2^{2n-1}\hbar)\chi_{m,\hbar}^+(\Hhb)}_1 \;.
	 	\end{align*}
	 By the triangle inequality, the inequality
	  \begin{equation*}
	 	\forall t\in\R,\quad\chi_{n,\hbar}^-(t)(t+2^{2n-1}\hbar)\leq( 2^{2(n-1)}\hbar)\chi_{n,\hbar}^+(t)\;,
	 \end{equation*}
	 and the functional calculus, one has
	 \begin{align*}
	 	&\norm{\chi_{n,\hbar}^-(\Hhb)\varphi P_j^r(\Hhb+2^{2n-1}\hbar)\chi_{m,\hbar}^+(\Hhb)}_1 \\&\quad\leq 
	 	\norm{\chi_{n,\hbar}^-(\Hhb)[\Hhb,\varphi P_j^r]\chi_{m,\hbar}^+(\Hhb)}_1 +	2^{2n-1}\hbar\norm{\chi_{n,\hbar}^-(\Hhb) \varphi P_j^r \chi_{m,\hbar}^+(\Hhb)}_1
	 \end{align*}
	 Thus,
	 \begin{align*}
	 	\norm{\chi_{n,\hbar}^-(\Hhb)\varphi P_j^r \chi_{m,\hbar}^+(\Hhb)}_1
	 	&\leq  \frac 1{2^{2(m-1)}\hbar}\norm{\chi_{n,\hbar}^-(\Hhb)[\varphi P_j^r,\Hhb]\chi_{m,\hbar}^+(\Hhb)}_1	\;.
	 \end{align*}
	 By iterating the same procedure but with $[\varphi P_j^r,\Hhb]$ instead of $\varphi P_j^r$,
	 \begin{align*}
	 	\norm{\chi_{n,\hbar}^-(\Hhb)[\varphi P_j^r,\Hhb]\chi_{m,\hbar}^+(\Hhb)}_1\leq  \frac 1{2^{2(m-1)}\hbar}\norm{\chi_{n,\hbar}^-(\Hhb)[[\varphi P_j^r,\Hhb],\Hhb]\chi_{m,\hbar}^+(\Hhb)}_1
	 \end{align*}
	 and one gets \cref{eq-proof:non-critical-cond_1}.

Let us estimate the right-hand side trace norm with the double commutator and show that
\begin{equation}\label{eq-proof:non-critical-cond_1-bis}
	\norm{\chi_{n,\hbar}^\mp(\Hhb)[[\varphi P_j^r,\Hhb],\Hhb]\chi_{m,\hbar}^\pm(\Hhb)}_1\leq C2^{m+n}\crochetjap{b}^{2+r}\hbar^{1-d}.
\end{equation}
By recalling that $\varphi\in\test{B(0,3R)}\subset\test{B(0,4R)}$,
\begin{align*}
	\left[[\varphi P_j^r,\Hhb],\Hhb\right] = \left[[\varphi P_j^r ,\Hloc],\Hloc\right],
\end{align*}
on smooth functions supported in $B(0,4R)$.
Let $\psi\in\test{\R^d,[0,1]}$ such that $\supp(\psi)\subset B(0,2R)$ and such that $\psi=1$ on $\supp(\varphi)$. 
By inserting the resolvent $(\Hhb-i)$ and by the H\"older inequality,
\begin{align*}
	&\norm{\chi_{n,\hbar}^-(\Hhb)\left[[\varphi P_j^r,\Hhb],\Hhb\right]\chi_{m,\hbar}^+(\Hhb)}_1
	\\&\quad\leq\prod_{\pm}\norm{\chi_{n,\hbar}^\pm(\Hhb)\psi (\Hhb-i)}_2\norm{(\Hhb-i)^{-1}\left[[\varphi P_j^r,\Hloc],\Hloc\right](\Hhb-i)^{-1}}_{\op}
	.
\end{align*}
Let us first treat the Hilbert--Schmidt norm bounds.
The idea is to separate into two terms on which we apply \cref{thm:WL-finite-int} on the interval $I=[-4^n\hbar,-4^{n-1}\hbar]$ and $[4^{n-1}\hbar,4^n\hbar]$ and \cref{lemma:aux-loc_4} applied to $f=\chi_{n,\hbar}^\pm$
\begin{align*}
	\norm{\chi_{n,\hbar}^\pm(\Hhb)\psi (\Hhb-i)}_2 
	&\leq \tr[\psi\chi_{n,\hbar}^\pm(\Hloc)]^{1/2}\tr[(t\mapsto t\chi_{n,\hbar}^\pm(t))(\Hhb-i)]^{1/2} \\&\qquad+\norm{[\chi_{n,\hbar}^\pm(\Hhb)\Hhb,\psi ]}_2.
\end{align*}
Thus, 
\begin{equation}\label{eq-proof:non-critical-cond_2}
	\norm{\chi_{n,\hbar}^\pm(\Hhb)\psi (\Hhb-i)}_2\leq C2^n\hbar^{\frac{1-d}2}.
\end{equation}
Moreover, by \eqref{eq-proof:non-critical-cond_3} (detailed in \cref{subsect:local-estimates-non-criticcal}) we have
\begin{equation*}
	\norm{(\Hhb-i)^{-1}\left[[\varphi P_j^r,\Hloc],\Hhb\right](\Hhb-i)^{-1}}_{\op}\leq C\hbar^2.
\end{equation*}
And we obtain \eqref{eq-proof:non-critical-cond_1-bis}.
\end{proof}

\subsection*{Acknowledgements}

N.~N.~is grateful to S.~Mikkelsen and L.~Morin for useful discussions. N.~B.\ and N.~N.\ acknowledge financial support from the European Union through the European Research Council's Starting Grant \textsc{FermiMath}, grant agreement nr.~101040991. Views and opinions expressed are those of the authors and do not necessarily reflect those
of the European Union or the European Research Council Executive Agency. Neither
the European Union nor the granting authority can be held responsible for them. N.~B., C.~B., and D.~M.\ were partially supported by Gruppo Nazionale per la Fisica Matematica GNFM -- INdAM in Italy. C.~B.\ and D.~M. acknowledge support by the PRIN 2022AKRC5P ``Interacting Quantum Systems:
Topological Phenomena and Effective Theories'', N.~B.\ acknowledges partial support by the same PRIN grant.

\bibliographystyle{amsalpha}
{\small
\bibliography{biblio-magnetic}
}

\end{document}